\documentclass[sn-mathphys]{sn-jnl}

\usepackage{multirow}

\theoremstyle{thmstyleone}%
\newtheorem{theorem}{Theorem}
\newtheorem{proposition}[theorem]{Proposition}%

\newtheorem{lemma}{Lemma}

\theoremstyle{thmstyletwo}%

\theoremstyle{thmstylethree}%

\newcommand{\dd}{\,\mathrm{d}}
\newcommand{\tr}{\,\mathrm{tr}}

\makeatletter
\newcommand{\mathleft}{\@fleqntrue\@mathmargin0pt}
\newcommand{\mathcenter}{\@fleqnfalse}
\makeatother

\usepackage{hhline}
\usepackage{multirow}
\usepackage{physics}
\usepackage{amsmath}
\begin{document}
\jyear{2022}
\title[Entropy fluctuation formulas of fermionic Gaussian states]{Entropy fluctuation formulas of fermionic Gaussian states}


\author{\fnm{Youyi} \sur{Huang}} 
\author{\fnm{Lu} \sur{Wei}} 

\affil{\center Department of Computer Science\\ Texas Tech University \\ \orgaddress{\city{Lubbock}, \postcode{79409}, \state{Texas}, \country{USA}}}

\abstract{We study the statistical behaviour of quantum entanglement in bipartite systems over fermionic Gaussian states as measured by von Neumann entropy. The formulas of average von Neumann entropy with and without particle number constrains have been recently obtained, whereas the main results of this work are the exact yet explicit formulas of variances for both cases. For the latter case of no particle number constrain, the results resolve a recent conjecture on the corresponding variance. Different than existing methods in computing variances over other generic state models, the key ingredient in proving the results of this work relies on a new simplification framework. The framework consists of a set of new tools in simplifying finite summations of what we refer to as dummy summation and re-summation techniques. As a byproduct, the proposed framework leads to various new transformation formulas of hypergeometric functions.}

\keywords{von Neumann entropy, fermionic Gaussian states, quantum entanglement, random matrix theory, orthogonal polynomials, special functions}


\maketitle

\section{Introduction and main results}\label{sec1}
Quantum entanglement is the most important feature in quantum mechanics. The understanding of the phenomenon of entanglement is crucial in realizing the revolutionary advances of quantum science. In the emerging field of quantum information processing, quantum entanglement is also the resource and medium that enable the underlying quantum technologies.

In this work, we study the statistical behaviour of entanglement over the fermionic Gaussian states. In the past decades, considerable effort has been devoted to investigating the degree of entanglement as measured by different entanglement entropies over the well-known Hilbert-Schmidt ensemble~\cite{HLW06, Page93,Foong94,Ruiz95,VPO16,Wei17,Wei20,HWC21,Lubkin78,Sommers04,Giraud07,MML02,Wei19T}. In particular, these studies focus on the statistical behaviour of entanglement entropies such as von Neumann entropy~\cite{HLW06, Page93,Foong94,Ruiz95,VPO16,Wei17,Wei20,HWC21}, quantum purity~\cite{Lubkin78,Sommers04,Giraud07}, and Tsallis entropy~\cite{MML02,Wei19T}. Driven by the recent breakthrough in probability theory on the Bures-Hall ensemble~\cite{Bortola09,Bortola10,Bortola14,FK16}, considerable progress has been made in understanding the von Neumann entropy~\cite{SK2019,Wei20BHA,Wei20BH,LW21} and quantum purity~\cite{Sommers04,Osipov10,Borot12,LW21} over the Bures-Hall ensemble. Similar investigations are now being carried out over the fermionic Gaussian ensemble, which is a generic state model relevant for different quantum information processing tasks~\cite{BHK21,BHKRV22,LRV20,OGMM19,ST21}. Very recently, the mean values of von Neumann entropy with and without particle number constrains over the fermionic Gaussian ensemble are obtained respectively in~\cite{BHK21} and~\cite{BHKRV22}. As an important step towards characterizing the statistical distribution of von Neumann entropy, we aim to derive the corresponding variances, which describe the fluctuation of the entropy around their mean values. The exact variance of von Neumann entropy over fermionic Gaussian states without particle number constrain has been conjectured in a previous work of the authors~\cite{HW22}. In the current work, we prove the conjecture as well as derive a variance formula for the case of a fixed particle number.


\subsection{Problem formulation}\label{sec2.1}
We introduce the formulation that leads to the fermionic Gaussian states with and without particle number constrains. A system of $N$ fermionic degrees of freedom can be decomposed into two subsystems $A$ and $B$ of the dimensions $m$ and $n$, respectively, with $m+n=N$. Without loss of generality, we assume $m\leq n$ . In the present work, we consider two scenarios of fermionic Gaussian states -- the fermionic Gaussian states with arbitrary number of particles and the fermionic Gaussian states with a fixed number of particles.
\subsubsection*{Case A: Arbitrary number of particles}
 A system of $N$ fermionic modes can be formulated in terms of a set of fermionic creation and annihilation operators $\hat{a}_{i}$ and $\hat{a}_{i}^{\dag}$, $i=1,\dots, N$. Since the modes are fermionic, these operators obey the canonical anti-commutation relation~\cite{OGMM19, BHKRV22},
\begin{equation}\label{facr}
\{\hat{a}_{i},\hat{a}_{j}^{\dag}\}=\delta_{ij}\mathbb{I}, \qquad \{\hat{a}_{i},\hat{a}_{j}\}=0=\{\hat{a}_{i}^{\dag},\hat{a}_{j}^{\dag}\},
\end{equation}
where $\{\hat{A},\hat{B}\}=\hat{A}\hat{B}+\hat{B}\hat{A}$ denotes the anti-commutation relation and $\mathbb{I}$ is an identity operator. Equivalently, one can also describe these fermionic modes in terms of the Majorana operators $\gamma_l$, $l=1,\dots, 2N$, and
\begin{equation}
\hat{\gamma}_{2i-1}=\frac{\hat{a}_{i}^\dag+\hat{a}_{i}}{\sqrt{2}},\qquad \hat{\gamma}_{2i}=\imath\frac{\hat{a}_{i}^\dag+\hat{a}_{i}}{\sqrt{2}}
\end{equation}
with $\imath=\sqrt{-1}$ being the imaginary unit. Note that the Majorana operators are Hermitian satisfying the anti-commutation relation
\begin{equation}\label{macr}
\{\hat{\gamma}_{l},\hat{\gamma}_{k}\}=\delta_{lk}\mathbb{I}.
\end{equation}
By collecting the Majorana operators into a $2N$ dimensional operator-valued column vector ${\gamma}=(\hat{\gamma}_1,\dots,\hat{\gamma}_{2N})^{\dag}$,
a fermionic Gaussian state is then written as the density operator of the form~\cite{ST21,BHKRV22}
\begin{equation}\label{Fgdo}
\rho(\gamma)=\frac{\mathrm{e}^{-\gamma^\dag Q\gamma}}{\tr(\mathrm{e}^{-\gamma^\dag Q\gamma})},
\end{equation}
where the coefficient matrix $Q$ is a $2N\times 2N$ imaginary anti-symmetric matrix as the consequence of the anti-communication relation~(\ref{macr}). There always exsits an orthogonal matrix $M$ that diagnoses the coefficient matrix $Q$ by transforming ${\gamma}$ into another Majorana basis $\mu=(\hat{\mu}_1,\dots,\hat{\mu}_{2N})^\dag=M\gamma$~\cite{BHKRV22}.
A fermionic Gaussian state is labelled by its anti-symmetric covariance matrix~\cite{BHKRV22}
\begin{equation}
 J=-\imath\tanh (Q)=M^T J_0M,
\end{equation}
where $\tanh(x)$ denotes the hyperbolic tangent function~\cite{AS72}, the matrix $J_0$ takes the block diagonal form
\begin{equation}\label{eq}
J_{0}=\left(\begin{array}{ccc}
\tanh(\lambda_1)\mathbb{A} & \dots & 0 \\
\vdots & \ddots & \vdots \\
0 & \dots & \tanh(\lambda_N)\mathbb{A} \\
\end{array}\right),
\end{equation}
and
\begin{equation}\label{eq}
\mathbb{A}=\left(\begin{array}{cc}
0 & 1 \\
-1 & 0 \\
\end{array}\right).
\end{equation}

We consider the von Neumann entropy as the measure of entanglement between the two subsystems. By restricting the matrix $J$ to the entries from subsystems $A$, the restricted matrix $J_A$ becomes the $2m\times 2m$ left-upper block of $J$. The von Neumann entropy of a fermionic Gaussian state of case A can be represented in terms of the real positive eigenvalues $x_i, i=1,\dots,m$ of $\imath J_{A}$ as \cite{BHK21, BHKRV22, HW22}
\begin{equation}\label{eq:von}
S=-\sum_{i=1}^{m}v(x_{i}),
\end{equation}
where
\begin{equation}\label{eq:vx}
v(x)=\frac{1-x}{2}\ln\frac{1-x}{2}+\frac{1+x}{2}\ln\frac{1+x}{2}.
\end{equation}
The resulting joint probability density of the eigenvalues $x_i, i=1,\dots,m$ is proportional to \cite{BHK21}
\begin{equation}\label{eq:ensemble_ap}
\prod_{1\leq i<j\leq m}\left(x_{i}^2-x_{j}^2\right)^{2}\prod_{i=1}^{m}\left(1-x_{i}^2\right)^{n-m}, \qquad x_{i}\in[0,1],
\end{equation}
which is obtained by recursively applying the result~\cite[Proposition A.2]{KFI19}.

\subsubsection*{Case B: Fixed number of particles}
For a fermionic Gaussian state $\ket{F}$ with a fixed particle number $p$, $m \leq p \leq n$, the corresponding covariance matrix $H$ can be expressed via the commutator of fermionic creation and annihilation operators as~\cite{BHKRV22, LRV20}
\begin{equation}
H_{ij}=-\imath\bra{F}\hat{a}_i^{\dag}\hat{a}_j-\hat{a}_j\hat{a}_i^{\dag}\ket{F}.
\end{equation}
Recall the canonical anti-commutation relation~(\ref{facr}), the entries of the matrix $H$ are then of the form
\begin{equation}\label{Fcov}
H_{ij}=-2\imath G_{ij}+\imath\delta_{ij}\mathbb{I},
\end{equation}
where $G_{ij}=\bra{F}\hat{a}_i^{\dag}\hat{a}_j\ket{F}$ denotes the entries of an $N\times N$ matrix $G$ of a fermionic system of $N$ modes. There exists a unitary transformation $U$ that diagonalizes $G$ into the form $U^{\dag}G U$, where the first $p$ diagonal elements are equal to $1$ and the rest are $0$. Therefore, one can write
\begin{equation}
G=U_{N\times p}U_{N\times p}^\dag.
\end{equation}

A fermionic Gaussian state of dimension $N=m+n$ with $p$ particles can be fully characterized by the matrices $H$ and $G$. The von Neumann entropy of the fermionic system in the case B can be represented as \cite{BHKRV22, LRV20}
\begin{equation}~\label{eq:vonf}
S=-\sum_{i=1}^m v(2y_i-1), \qquad y_{i}\in[0,1],
\end{equation}
where $y_i$, $i=1,\dots,m$ are the eigenvalues of the restricted $m\times m$ matrix $G_A=U_{m\times p}U_{m\times p}^\dag$. The eigenvalue distribution of the random matrix $U_{m\times p}U_{m\times p}^\dag$ is the well-known Jacobi unitary ensemble~\cite{Mehta, Forrester}. We denote $x_i$, $i=1,\dots,m$ the eigenvalues of the $m\times m$ left-upper block of matrix $\imath H$. Changing the variables $x_i=2y_i-1$ in~(\ref{eq:vonf}) leads to the von Neumann entropy~(\ref{eq:von}) of case B. The resulting joint probability density of the eigenvalues $x_i,i=1,\dots, m$ is proportional to~\cite{BP21}
\begin{equation}\label{eq:ensemble_fp}
\prod_{1\leq i<j\leq m}\left(x_{i}-x_{j}\right)^{2}\prod_{i=1}^{m}\left({1+x_i}\right)^{p-m}\left(1-x_{i}\right)^{n-p},\qquad x_{i}\in[-1,1].
\end{equation}

It is important to point out that the two joint probability densities~(\ref{eq:ensemble_ap}) and~(\ref{eq:ensemble_fp}) can be compactly represented by a single joint density as
\begin{equation}\label{eq:fg-ensemble}
f_{\mathrm{FG}}(x) \propto \prod_{1\leq i<j\leq m}\left(x_{i}^\gamma-x_{j}^\gamma\right)^{2}\prod_{i=1}^{m}\left(1-x_{i}\right)^{a}\left(1+x_{i}\right)^{b},
\end{equation}
where for the case A we have
\begin{equation}\label{eq:fgap}
\gamma=2,~~~ a=b=n-m\geq 0,~~~x\in[0,1],
\end{equation}
 and for the case B we have
\begin{equation}\label{eq:fgfp}
\gamma=1,~~~ a=n-p\geq 0,~~~ b=p-m\geq 0,~~~x\in[-1,1].
\end{equation}
We omit the normalizations in the density~(\ref{eq:fg-ensemble}) as they will not be made use of  in the subsequent calculations. Note that the variance computation is difficult for an arbitrary $\gamma$ in~(\ref{eq:fg-ensemble}), and one has to consider the cases $\gamma=2$ and $\gamma=1$ separately.

\subsection{Main results}\label{sec2.2}
We now introduce the exact mean and variance formulas of von Neumann entropy for both case A and case B. The mean values have been recently computed~\cite{BHK21,BHKRV22} as summarized in Proposition~\ref{propaa} and Proposition~\ref{propfa} for case A and case B, respectively. The corresponding variance formulas are presented in Proposition~\ref{propav} and Proposition~\ref{propfv} below, which are the main results of the work.

\begin{proposition}[\cite{BHK21}]\label{propaa}
For subsystem dimensions $m\leq n$, the mean value of the von Neumann entropy~(\ref{eq:von}) of fermionic Gaussian states with arbitrary number of particles~(\ref{eq:fgap}) is given by
\begin{align}
\mathbb{E}\!\left[S\right]=&\left(m+n-\frac{1}{2}\right)\psi_{0}(2m+2n)+\left(\frac{1}{4}-m\right)\psi_{0}(m+n)+\left(\frac{1}{2}-n\right)\psi_{0}(2n)\nonumber\\
&-\frac{1}{4}\psi_{0}(n)-m\label{f:apa},
\end{align}
where
\begin{equation}\label{eq:digamma}
\psi_{0}(x)=\frac{\dd\ln\Gamma(x)}{\dd x}
\end{equation}
 is the digamma function.
\end{proposition}

\begin{proposition}[\cite{BHKRV22}]\label{propfa}
For subsystem dimensions $m \leq n$, the mean value of the  von Neumann entropy~(\ref{eq:von}) of fermionic Gaussian states with a fixed particle number~(\ref{eq:fgfp}) is given by
\begin{align}\label{f:fpa}
\mathbb{E}\!\left[S\right]=&-\frac{m (m+n-p)}{m+n}\psi _0(m+n-p)+(m+n) \psi _0(m+n+1)\nonumber\\
&-\frac{m p }{m+n}\psi _0(p+1)-n \psi _0(n+1)-m.
\end{align}
\end{proposition}

\begin{proposition}\label{propav}
For subsystem dimensions $m\leq n$, the variance of the von Neumann entropy~(\ref{eq:von}) of fermionic Gaussian states with arbitrary number of particles~(\ref{eq:fgap}) is given by
\begin{align}\label{f:apv}
\mathbb{V}\!\left[S\right]=~\!\!&\!\left(\frac{1}{2}-m-n\right) \psi _1(2 m+2 n)+\left(n-\frac{1}{2}\right) \psi _1(2 n)+\left(\frac{m (2 m+n-1)}{2 m+2 n-1}-\frac{1}{8}\right)\nonumber\\
 & \times\psi _1(m+n)+\frac{\psi_1(n)}{8}-\frac{1}{2} (\psi_0(2 m+2 n)-\psi_0(2 n)),
\end{align}
where
\begin{equation}\label{eq:trigamma}
\psi_{1}(x)=\frac{\dd^{2}\ln\Gamma(x)}{\dd x^{2}}
\end{equation}
is the trigamma function.
\end{proposition}
\begin{proposition}\label{propfv}
For subsystem dimensions $m \leq n$, the variance of the von Neumann entropy~(\ref{eq:von}) of fermionic Gaussian states with a fixed particle number~(\ref{eq:fgfp}) is given by
\begin{align}\label{f:fpv}
\mathbb{V}\!\left[S\right]=~\!\!&c_0 \psi _1(m+n-p)-(m+n) \psi _1(m+n)+n \psi _1(n)+c_1 \psi _1(p)+c_2 \nonumber\\
&\times (\psi _0(m+n-p)-\psi _0(p))^2+c_3 \left(\psi _0(m+n-p)-\psi _0(p)\right)-\psi _0(m+n)\nonumber\\
&+\psi _0(n)+c_4,
\end{align}
where the coefficients $c_i$ are summarized in Table~\ref{tab_prop4} below with $(a)_{n}=\Gamma(a+n)/\Gamma(a)$
denoting the Pochhammer symbol.
\end{proposition}

\begin{table}[h]
\begin{center}
\begin{minipage}{174pt}
\caption{Coefficients of von Neumann entropy variance in Proposition~\ref{propfv}}\label{tab_prop4}
\begin{tabular}{@{}l l@{}}
\toprule
$c_0~=$& $\displaystyle \frac{m (m+n-p) \left(m^2+2 m n+n^2-n p-1\right)}{(m+n-1)_3}$\\
$c_1~=$& $\displaystyle \frac{m p \left(m^2+m n+n p-1\right)}{(m+n-1)_3}$  \\
$c_2~=$& $\displaystyle \frac{m n p (m+n-p)}{ (m+n) (m+n-1)_3}$  \\
$c_3~=$& $\displaystyle -\frac{m (m+1) (m+n-2 p)}{(m+n)(m+n)_2}$  \\
$c_4~=$& $\displaystyle -\frac{m (2 m+n+2)}{(m+n)(m+n)_2}$ \\
\midrule
\end{tabular}
\end{minipage}
\end{center}
\end{table}
\noindent

The proof to Proposition~\ref{propav} and Proposition~\ref{propfv} will be presented in Section~\ref{sec3}. Note that a special case of equal subsystem dimensions ($m=n$) of the result in Proposition~\ref{propav} has been established very recently~\cite{HW22} by utilizing an existing simplification framework developed in~\cite{Wei17,Wei20,HWC21,Wei22,Wei20BH,LW21,HW22}. However, for the case of arbitrary subsystem dimensions, the existing framework is not sufficient to simplify some of the summations in the variance calculation, where a key technical contribution of the work is to develop a new simplification framework.

\begin{figure}[h]%
\centering
\includegraphics[width=0.9\textwidth]{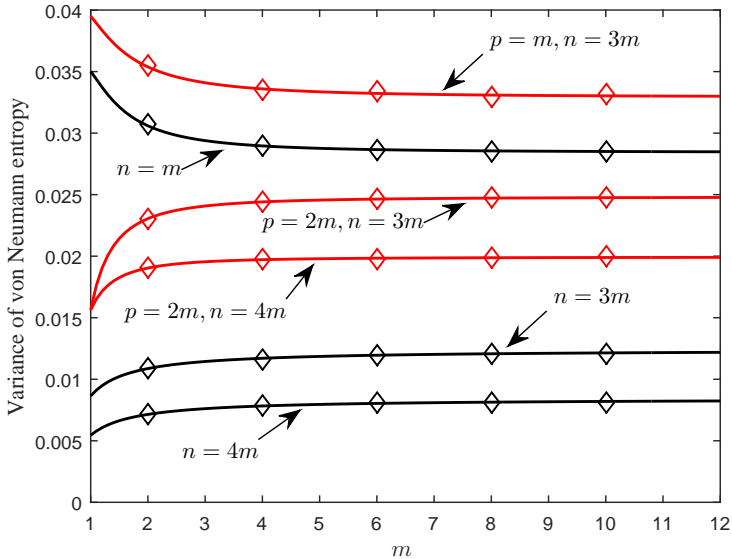}
\caption{Variance of von Neumann entropy: analytical results versus simulations. The black lines represent the obtained analytical result~(\ref{f:apv}) for the cases $n=m$, $n=3m$, and $n=4m$. The red lines are drawn by the result~(\ref{f:fpv}) for the cases $p=m, n=3m$; $p=2m, n=3m$; and $p=2m, n=4m$. The diamond scatters represent numerical simulations.}\label{fig1}
\end{figure}
To illustrate the derived results~(\ref{f:apv}) and~(\ref{f:fpv}), we plot in Figure~\ref{fig1} the exact variance of von Neumann entropy as compared with the simulations\footnote{The simulations performed in figures \ref{fig1}--\ref{fig3} utilize the Mathematica codes provided by Santosh Kumar based on the log-gas approach as discussed in~\citep[Appendix B]{SK2019}.}. In Figure~\ref{fig1}, we observe that the variance in case A approaches to a constant when system dimensions increase with a fixed ratio
\begin{equation}
f_1=\frac{m}{n+m},
\end{equation}
where the variance in case B follows the same behavior with fixed $f_1$ and
\begin{equation}
f_2=\frac{p}{n+m}.
\end{equation}
This phenomenon can be analytically established by the asymptotic results of variances in the literature. For case A, in the asymptotic regime~\cite{BHK21}
\begin{equation}\label{regime:valim}
m\to \infty,\qquad n\to \infty,\qquad0<f_1\leq\frac{1}{2},
\end{equation}
one has~\cite{BHK21}
\begin{equation}\label{eq:valim}
\mathbb{V}\!\left[S\right]=\frac{1}{2}\left(f_1+f_1^2+\ln (1-f_1)\right)+o\left(\frac{1}{m+n}\right),
\end{equation}
whereas for case B, in the asymptotic regime~\cite{BHKRV22}
\begin{equation}\label{regime:valimf}
m\to \infty,\qquad p\to \infty,\qquad n\to \infty,\qquad 0<f_1\leq f_2\leq \frac{1}{2},
\end{equation}
one has~\cite{BHKRV22}
\begin{align}\label{eq:valimf}
\mathbb{V}\!\left[S\right]=\!~&f_1+f_1^2+\ln \left(1-f_1\right)+ f_1 f_2\left(1-f_1\right) \left(1-f_2\right) \ln ^2\frac{1-f_2}{f_2}\nonumber\\
\!~&+f_1^2 \left(2 f_2-1\right) \ln \frac{1-f_2}{f_2}+o\left(\frac{1}{(m+n)^2}\right).
\end{align}
The above asymptotic variances~(\ref{eq:valim}) and~(\ref{eq:valimf}) can be directly recovered by the results in Proposition~\ref{propav} and Proposition~\ref{propfv}, respectively. Moreover, the correction terms of any order can be simply obtained from our exact variance formulas upon using the asymptotic behaviour of polygamma functions
\begin{align}
\psi_0(x)=&\ln (x)-\frac{1}{2 x}-\sum_{l=1}^{\infty}\frac{B_{2l}}{2lx^{2l}},\qquad x\to\infty \label{eq:limpl0},\\
\psi_{1}(x)=&\frac{1+2x}{2x^{2}}+\sum_{l=1}^{\infty}\frac{B_{2l}}{x^{2l+1}},\qquad x\to\infty \label{eq:limpl1},
\end{align}
where $B_k$ is the $k$-th Bernoulli number~\cite{AS72}. For example, utilizing the next order of correction, the asymptotic result~(\ref{eq:valimf}) is refined to
\begin{align}
\mathbb{V}\!\left[S\right]=&f_1^2+f_1+\ln \left(1-f_1\right)+f_1  f_2\left(1-f_1\right) \left(1-f_2\right) \ln ^2\frac{1-f_2}{f_2}+f_1^2\nonumber\\
&\times \left(2 f_2-1\right) \ln \frac{1-f_2}{f_2}+\frac{1}{12 (m+n)^2}\Bigg(\frac{f_1^2}{\left(f_2-1\right){}^2}+\frac{f_1^2}{f_2^2}+12 f_1^2\nonumber\\
&-12 f_1+\frac{1}{\left(f_1-1\right){}^2}+\frac{f_1-3 f_1^2}{f_2-1}+\frac{3 f_1^2-f_1}{f_2}-1 \nonumber\\
&+\frac{2 \left(f_1-1\right) f_1 \left(12 f_2^3-18 f_2^2+4 f_2+1\right)}{\left(f_2-1\right) f_2}\ln \frac{1-f_2}{f_2}\nonumber\\
&+12 \left(f_1-1\right) f_1 \left(f_2-1\right) f_2 \ln ^2\frac{1-f_2}{f_2}\!\Bigg)+o\left(\frac{1}{(m+n)^4}\right).
\end{align}

\begin{figure}[h]%
\centering
\includegraphics[width=0.9\textwidth]{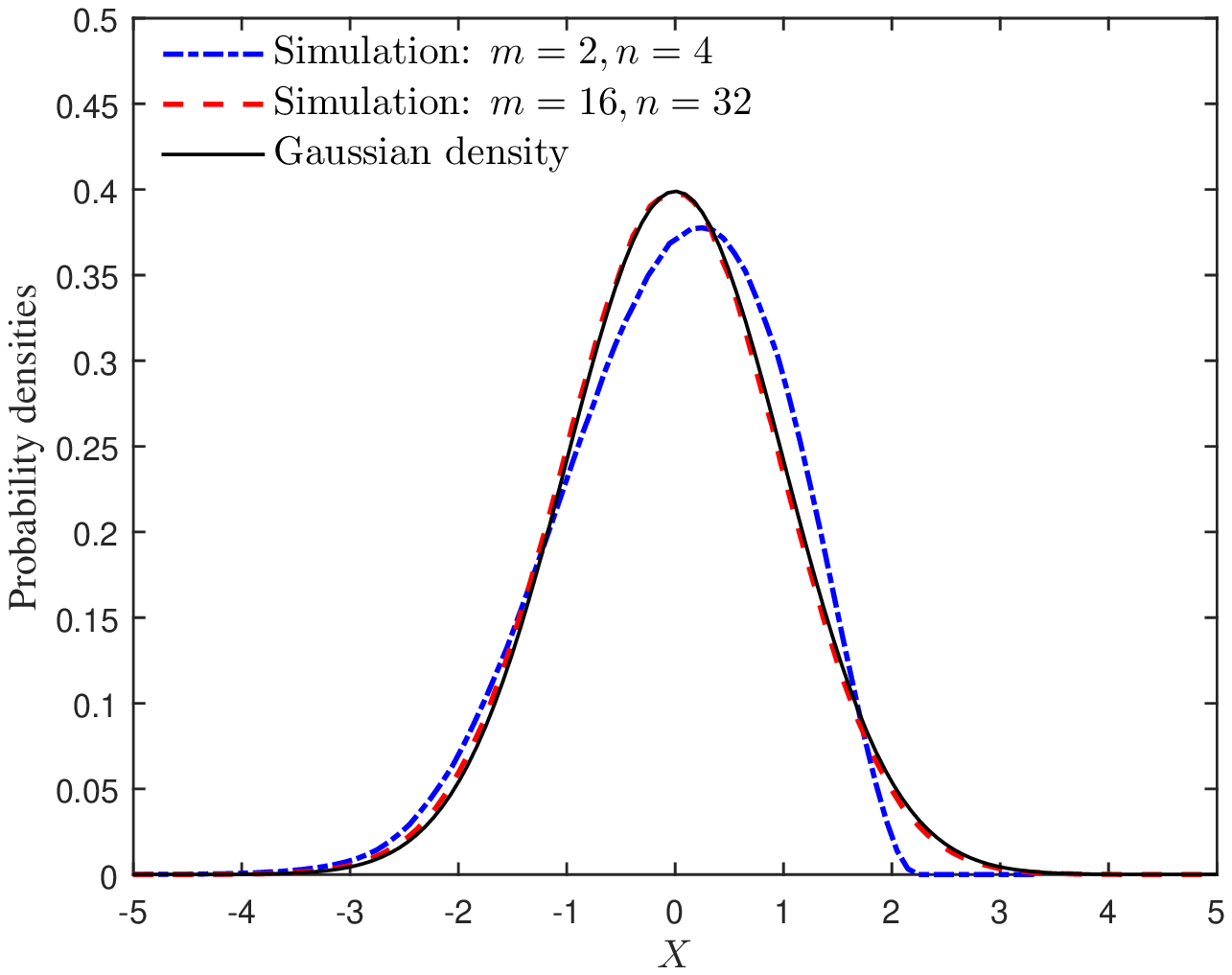}
\caption{Probability densities of standardized von Neumann entropy for case A in (\ref{eq:fgap}): a comparison of Gaussian density~(\ref{eq:iappr}) to the simulation results.  The dash-dot line in blue and the dashed line in red refer to the standardized von Neumann entropy~(\ref{eq:X}) of subsystem dimensions $m=2$, $n=4$, and $m=16$, $n=32$, respectively. The solid black line represents the Gaussian density~(\ref{eq:iappr}).}\label{fig2}
\end{figure}

\begin{figure}[h]%
\centering
\includegraphics[width=0.9\textwidth]{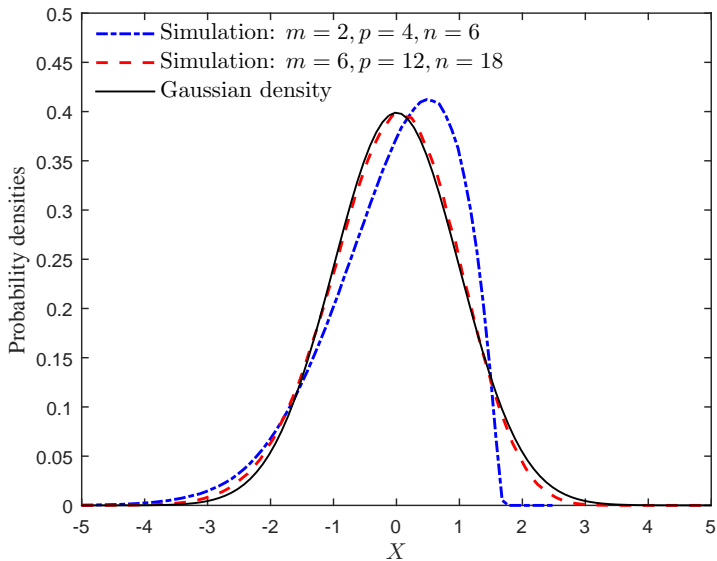}
\caption{Probability densities of standardized von Neumann entropy for case B in (\ref{eq:fgfp}): a comparison of Gaussian density~(\ref{eq:iappr}) to the simulation results.  The dash-dot line in blue and the dashed line in red refer to the standardized von Neumann entropy~(\ref{eq:X}) of dimensions $m=2$, $p=4$, $n=6$, and $m=6$, $p=12$, $n=18$, respectively. The solid black line represents the Gaussian density~(\ref{eq:iappr}).}\label{fig3}
\end{figure}

To understand the distribution of the von Neumann entropy, simple approximations can now be constructed by using the obtained mean and variance formulas. We first standardize the von Neumann entropy as
\begin{equation}\label{eq:X}
X=\frac{S-\mathbb{E}\!\left[S\right]}{\sqrt{\mathbb{V}\!\left[S\right]}},
\end{equation}
where the random variable $X$ is of zero mean and unit variance. We now compare the distribution of $X$ with a standard Gaussian distribution
\begin{equation}\label{eq:iappr}
\phi(x)=\frac{1}{\sqrt{2\pi}}\mathrm{e}^{-\frac{1}{2}x^2}, \qquad x\in(-\infty,\infty).
\end{equation}
In figures~\ref{fig2}--\ref{fig3}, we plot the simulation results of the standardized von Neumann entropy $X$ as compared with a standard Gaussian. Specifically, the ratios are fixed to $f_1=1/3$ for case A in Figure \ref{fig2} and $f_1=1/4$, $f_2=1/2$ for case B in Figure \ref{fig3}. It is observed from the figures that the Gaussian density captures accurately the distribution of the standardized von Neumann entropy $X$ for moderately large dimensions. We also observe that the true distribution of $X$ is non-symmetric and appears to be left-skewed when the subsystem dimensions are small as seen from the dash-dot blue curves. In comparison, when the subsystem dimensions become larger, the distributions of $X$ appear to be closer to the Gaussian distribution. In fact, the Gaussian density as a limiting behavior of von Neumann entropy has been conjectured over different random matrix models of Hilbert-Schmidt ensemble~\cite{Wei20}, Bures-Hall ensemble~\cite{Wei20BH}, and fermionic Gaussian ensemble of an arbitrary number of particles~\cite{HW22}. Here, one is also attempting to conjecture that under the asymptotic regime~(\ref{regime:valimf}), the standardized von Neumann entropy~(\ref{eq:X}) of fermionic Gaussian states with a fixed particle number~(\ref{eq:fgfp}) converges in distribution to a standard Gaussian.

The rest of the paper is organized as follows. The detailed calculations of the main results in Proposition~\ref{propav} and Proposition~\ref{propfv} are provided in Section~\ref{sec3}. In Appendix~\ref{secA}, we list the summation representations of the integrals involved in the variance computations. Some additional finite sum identities utilized in the simplification are listed in Appendix \ref{secB}.

\section{Variance calculation}\label{sec3}
In this section, we prove the results in Proposition~\ref{propav} and Proposition~\ref{propfv}. In Section~\ref{subsec3.1}, we obtain the summation representations of the variances. Tools in simplifying these summations are introduced in Section~\ref{subsec3.2}, where the new simplification framework consisting of six lemmas is presented first. The detailed simplification procedures that lead to the claimed results~(\ref{f:apv}) and~(\ref{f:fpv}) are discussed in Section~\ref{subsec3.3}.
\subsection{Correlation functions and integral calculations}\label{subsec3.1}
Recall the definition~(\ref{eq:von}) of von Neumann entropy
\begin{equation}
S= -\sum_{i=1}^{m}v(x_{i}),
\end{equation}
with
\begin{equation}
v(x)=\frac{1-x}{2}\ln\frac{1-x}{2}+\frac{1+x}{2}\ln\frac{1+x}{2},
\end{equation}
computing its variance requires one and two arbitrary eigenvalue densities of the fermionic Gaussian ensemble~(\ref{eq:fg-ensemble}). Denoting $g_{l}(x_{1},\dots,x_{l})$ as the joint density of $l$ arbitrary eigenvalues, the variance of von Neumann entropy is written as
\begin{equation}\label{eq:vc1}
\mathbb{V}\!\left[S\right]=\mathbb{E}\!\left[S^{2}\right]-\mathbb{E}^{2}\!\left[S\right],
\end{equation}
where
\begin{align}
\mathbb{E}\!\left[S^{2}\right]&=m\int_x v^{2}(x)g_1(x)\dd x+m(m-1)\iint_{x,y} v(x)v(y)g_2\left(x,y\right)\dd x\dd y\label{eq:ES2}\\
\mathbb{E}\!\left[S\right]&=m\int_x v(x)g_1(x)\dd x\label{eq:ES1}.
\end{align}
In (\ref{eq:ES2}) and (\ref{eq:ES1}), the support is $x,y\in[0,1]$ for case A, and the support is $x,y \in [-1,1]$ for case B.

For fermionic Gaussian ensemble~(\ref{eq:fg-ensemble}), it is a well-known result in random matrix theory that  the joint density $g_{l}(x_{1},\dots,x_{l})$ can be written in terms of an $l\times l$ determinant as~\cite{Mehta, Forrester}
\begin{equation}\label{eq:gl}
g_{l}(x_{1},\dots,x_{l})=\frac{(m-l)!}{m!}\det\left(K\left(x_{i},x_{j}\right)\right)_{i,j=1}^{l}.
\end{equation}
The determinant in~(\ref{eq:gl}) is known as the $l$-point correlation function~\cite{Forrester}, where
\begin{equation}\label{eq:corrk}
K\left(x,y\right)=\sqrt{w(x)w(y)}\sum_{k=0}^{m-1}\frac{J_k^{(a,b)}(x) J_k^{(a,b)}(y)}{h_k}
\end{equation}
is the correlation kernel with the weight function
\begin{equation}
w(x)=\left(\frac{1-x}{2}\right)^{a}  \left(\frac{1+x}{2}\right)^{b}.
\end{equation}
For convenience, we summarize in Table~\ref{tab_kernel} the parameters $a$, $b$, and the order $k$ of the Jacobi polynomial $J_k^{(a,b)}(x)$ along with the normalization constants $h_k$ of case A and case B.

\begin{table}[htbp]
\begin{center}
\begin{minipage}{340pt}
\caption{Parameters of correlation kernel~(\ref{eq:corrk}) in case A and case B.}\label{tab_kernel}
\begin{tabular}{@{}lcc@{}}
\toprule
   &Case A: Arbitrary number of particle~(\ref{eq:fgap}) &Case B: Fixed number of particle~(\ref{eq:fgfp}) \\
\midrule
$a$   & $n-m$ &$n-p$  \\
\\
 $b$    &$n-m$ &$p-m$\\
\\
$k$   &$2k$   &$k$  \\
\\
$h_k$ &$\displaystyle \frac{(4k+2a+1)^{-1}\Gamma^{2}(2k+a+1)}{\Gamma(2k+2a+1)\Gamma(2k+1)}$& $\displaystyle \frac{2 \Gamma (k+a+1) \Gamma (k+b+1)}{ (2k+a+b+1) \Gamma (k+1)\Gamma (k+a+b+1)}$  \\
\botrule
\end{tabular}
\end{minipage}
\end{center}
\end{table}
\noindent
The constants $h_k$  in Table~\ref{tab_kernel} are obtained as follows.
For case B, the orthogonality relation~\cite{Forrester}
\begin{align}
&\int_{-1}^{1}\left(\frac{1-x}{2}\right)^{a}\left(\frac{1+x}{2}\right)^{b}J^{(a,b)}_{k}(x)J^{(a,b)}_{l}(x)\dd x\nonumber\\
&=\frac{2\Gamma(k+a+1)\Gamma(k+b+1)}{(2k+a+b+1)\Gamma(k+1)\Gamma(k+a+b+1)}\delta_{kl}, \quad \Re(a,b)>-1,\label{eq:orthogonality}
\end{align}
directly leads to the normalization constant $h_{k}$ of Jacobi polynomials $J^{(a,b)}_{k}$. For case A, by using the parity property of Jacobi polynomials~\cite{Szego}
\begin{equation}\label{eq:parity}
J^{(a,b)}_{k}(-x)=(-1)^{k}J^{(b,a)}_{k}(x),
\end{equation}
the orthogonality relation~(\ref{eq:orthogonality}) can also be written as
\begin{align}
&\int_{0}^{1}\left(\frac{1-x}{2}\right)^{a}\left(\frac{1+x}{2}\right)^{a}J^{(a,a)}_{2k}(x)J^{(a,a)}_{2l}(x)\dd x\nonumber\\
&=\frac{\Gamma(2k+a+1)\Gamma(2k+a+1)}{(4k+2a+1)\Gamma(2k+1)\Gamma(2k+2a+1)}\delta_{kl}, \quad \Re(a)>-1,
\end{align}
which gives the normalization constant $h_{k}$ of the polynomials $J^{(a,a)}_{2k}(x)$ as shown in Table~\ref{tab_kernel}.

By using the joint density~(\ref{eq:gl}) and the results (\ref{eq:ES2})--(\ref{eq:ES1}), the variance~(\ref{eq:vc1}) now boils down to computing two integrals involving the $1$-point and $2$-point correlation functions as, cf.~\cite{Wei17,Wei20BH, BHK21,HW22},
\begin{equation}\label{def:vs}
\mathbb{V}\!\left[S\right]=\mathrm{I_{A}}-\mathrm{I_{B}},
\end{equation}
where
\begin{align}
\mathrm{I_{A}}&=\int_x v^{2}(x)K(x,x)\dd x \label{eq:IA}\\
\mathrm{I_{B}}&=\iint_{x,y} v(x)v(y)K^{2}\left(x,y\right)\dd x\dd y\label{eq:IB}
\end{align}
with $ x,y \in [0,1]$ for case A and  $ x,y \in [-1,1]$ for case B.

We now compute the above two integrals $\mathrm{I_{A}}$ and $\mathrm{I_{B}}$ into the corresponding summation representations. Note that the subsequent calculations of the case A in~(\ref{eq:fgap}) and case B  in~(\ref{eq:fgfp}) are different, which will be performed separately in the following.
\subsubsection*{Case A: Arbitrary number of particles}\label{subsubsec2}
By the definition of the correlation kernel~(\ref{eq:corrk}) and keeping in mind the parity property~(\ref{eq:parity}) of Jacobi polynomials,
 the integral
\begin{equation}
\mathrm{I_{A}}=\int_0^1 v^{2}(x)K(x,x)\dd x
\end{equation} of fermionic Gaussian states with arbitrary number of particles boils down to computing the two parts
\begin{equation}\label{eq:IAex}
\mathrm{I_{A}}=\mathrm{A_{1}}+\mathrm{A_{2}},
\end{equation}
where
\begin{align}
\mathrm{A_{1}}&=\sum_{k=0}^{m-1}\frac{1}{h_{k}}\int_{-1}^{1}\left(\frac{1-x}{2}\right)^{a}\left(\frac{1+x}{2}\right)^{a+2}\ln^{2}\frac{1+x}{2}J_{2k}^{(a,a)}(x)^2\dd x\label{eq:apA1}\\
\mathrm{A_{2}}&=\sum_{k=0}^{m-1}\frac{1}{h_{k}}\int_{-1}^{1}\left(\frac{1-x}{2}\right)^{a+1}\left(\frac{1+x}{2}\right)^{a+1}\ln\frac{1-x}{2}\ln\frac{1+x}{2}J_{2k}^{(a,a)}(x)^2\dd x.\label{eq:apA2}\nonumber\\
\end{align}
Here, we recall that $a=n-m\geq0$ denotes the subsystem difference of fermionic Gaussian states with arbitrary number of particles.

Similarly, the integral
\begin{equation}
\mathrm{I_{B}}=\int_0^1 \int_0^1v(x)v(y)K^{2}\left(x,y\right)\dd x\dd y
\end{equation}
 can be written in terms of the following two integrals
\begin{equation}\label{eq:IBex}
\mathrm{I_{B}}=\mathrm{B_{1}}+\mathrm{B_{2}},
\end{equation}
where
\begin{align}
\mathrm{B_{1}}=&\sum_{k=0}^{m-1}\frac{1}{h_{k}^{2}}\left(\int_{-1}^{1}\left(\frac{1-x}{2}\right)^{a}\left(\frac{1+x}{2}\right)^{a+1}\ln\frac{1+x}{2}J_{2k}^{(a,a)}(x)^2\dd x\right)^{2}\label{eq:apB1}\\
\mathrm{B_{2}}=&\sum_{j=1}^{m-1}\sum_{k=0}^{m-j-1}\frac{2}{h_{k+j}h_{k}}\Bigg(\int_{-1}^{1}\left(\frac{1-x}{2}\right)^{a}\left(\frac{1+x}{2}\right)^{a+1}\ln\frac{1+x}{2}\nonumber\\
&\times J_{2k+2j}^{(a,a)}(x)J_{2k}^{(a,a)}(x)\dd x\Bigg)^{2}\label{eq:apB2}.
\end{align}

Computing the above integrals $\mathrm{A_1}$, $\mathrm{A_2}$, $\mathrm{B_1}$, and $\mathrm{B_2}$ requires the following two integral identities. The first one is
\begin{align}
&\int_{-1}^{1}\left(\frac{1-x}{2}\right)^{a_1}\left(\frac{1+x}{2}\right)^{c}J^{(a_1,b_1)}_{k_1}(x)J^{(a_2,b_2)}_{k_2}(x)\dd x \nonumber \\
&=\frac{2 \left(k_1+1\right)_{a_1}}{\left(b_2+k_2+1\right)_{a_2}}\sum _{i=0}^{k_2} \frac{(-1)^{i+k_2} (i+1)_c \left(i+b_2+1\right)_{a_2+k_2} }{\Gamma \left(k_2-i+1\right) \Gamma \left(a_1+c+i+k_1+2\right)}\nonumber\\
&~~~\!~\times\left(c+i-b_1-k_1+1\right)_{k_1}, \quad \Re(a_1,a_2,b_1,b_2,c)>-1.\label{eq:SIac2}
\end{align}
To show this identity, we first note that the Jacobi polynomial $J_k^{(a,b)}(x)$ supported in $x\in[-1,1]$ admits different representations~\cite{Szego, Forrester}
\begin{align}\label{eq:J1}
J^{(a,b)}_{k}(x)&=\frac{(-1)^{k}(b+1)_{k}}{k!}\sum_{i=0}^{k}\frac{(-k)_{i}(k+a+b+1)_{i}}{(b+1)_{i}\Gamma(i+1)}\left(\frac{1+x}{2}\right)^{i}\\
&=\sum _{i=0}^k \frac{(-1)^i \Gamma (a+k+1) (k+b-i+1)_i}{\Gamma (i+1) \Gamma (a+i+1) \Gamma (k-i+1)}\left(\frac{1-x}{2}\right)^i \left(\frac{1+x}{2}\right)^{k-i}.\nonumber\\ \label{eq:J2}
\end{align}
The identity~(\ref{eq:SIac2}) is then obtained by using the definition~(\ref{eq:J1}) for the polynomial $J^{(a_2,b_2)}_{k_2}$ before applying the well-known integral identity~\cite{Szego, Forrester}
\begin{align}
&\int_{-1}^{1}\left(\frac{1-x}{2}\right)^{a}\left(\frac{1+x}{2}\right)^{c}J^{(a,b)}_{k}(x)\dd x \nonumber \\
&=\frac{2 \Gamma (c+1) (k+1)_a (c-b-k+1)_k}{\Gamma (a+c+k+2)}, \quad \Re(a,b,c)>-1.\label{eq:SIac}
\end{align}
The second identity is
\begin{align}
&\int_{-1}^{1}\left(\frac{1-x}{2}\right)^{d}\left(\frac{1+x}{2}\right)^{c}J^{(a_1,b_1)}_{k_1}J^{(a_2,b_2)}_{k_2}(x)\dd x \nonumber\\
&=\frac{2 \Gamma \left(a_2+k_2+1\right) \Gamma \left(b_2+k_2+1\right)}{\Gamma \left(c+d+k_1+k_2+2\right)}\sum _{i=0}^{k_2} \frac{(-1)^i \Gamma \left(d-a_1+i+1\right)}{\Gamma (i+1) \Gamma \left(a_2+i+1\right)}\nonumber\\
&~~~\!~\times\frac{\Gamma \left(c-b_1-i+k_2+1\right)}{\Gamma \left(k_2-i+1\right) \Gamma \left(b_2-i+k_2+1\right)}\sum _{j=0}^{k_1} \frac{(-1)^j \left(k_1-j+1\right)_{d+i}}{\Gamma (j+1)}\nonumber\\
&~~~\!~\times\frac{\left(c-i+j-b_1-k_1+k_2+1\right)_{b_1+k_1}}{\Gamma \left(d-a_1+i-j+1\right)}, \quad \Re(a_1,a_2,b_1,b_2,c,d)>-1,\label{eq:SIdc2}
\end{align}
which is obtained by using the definition~(\ref{eq:J2}) for the polynomial $J^{(a_2,b_2)}_{k_2}$  before applying the identity~\cite[Equation (62)]{HW22}
\begin{align}
&\int_{-1}^{1}\left(\frac{1-x}{2}\right)^{d}\left(\frac{1+x}{2}\right)^{c}J^{(a,b)}_{k}(x)\dd x  \nonumber \\
&=\frac{2 \Gamma (c-b+1) \Gamma (d-a+1)}{\Gamma (c+d+k+2)}\sum _{i=0}^{k} \frac{(-1)^i \Gamma (c+i+1) \Gamma (d-i+k+1)}{\Gamma (i+1) \Gamma (k-i+1)}\nonumber\\
&~~~\!~\times\frac{1}{ \Gamma (d-a-i+1) \Gamma (c-b+i-k+1)}, \quad \Re(a,b,c,d)>-1\label{eq:SIcd}.
\end{align}

The integral $\mathrm{A_1}$ is now calculated by applying the identity~(\ref{eq:SIac2}), where we need to assign
\begin{equation}\label{eq:sp1}
a_1=b_1=a_2=b_2=a,~~~ k_1=k_2=2k,
\end{equation}
and take twice derivatives of $c$ before setting $c=a+2$. Under the same specialization~(\ref{eq:sp1}), the integral $\mathrm{A_2}$ is calculated by taking  derivatives of both $c$ and $d$ of the identity~(\ref{eq:SIdc2}) before setting $c=d=a+1$, while the integral $\mathrm{B_1}$ is calculated by taking derivative of $c$ of the identity~(\ref{eq:SIac2}) before setting $c=a+1$. According to the result~(\ref{eq:apB2}), the integral $\mathrm{B_2}$ is calculated by specializing
\begin{equation}\label{eq:sp2}
a_1=b_1=a_2=b_2=a,~~~ k_1=2k+2j,~~~ k_2=2k,
\end{equation}
in the identity~(\ref{eq:SIac2}), and taking derivative of $c$ before setting $c=a+1$.

In writing down the summation forms of $\mathrm{A_1}$, $\mathrm{A_2}$, $\mathrm{B_1}$, and $\mathrm{B_2}$, one will have to resolve the indeterminacy by using the following asymptotic expansions of gamma and polygamma functions of negative arguments~\cite{AS72} when $\epsilon \rightarrow 0$,
\label{eq:pgna}\begin{align}
\Gamma(-l+\epsilon)&=\frac{(-1)^{l}}{l!\epsilon}\left(1+\psi_{0}(l+1)\epsilon+o\left(\epsilon^2\right)\right)\label{eq:pgna1}\\
\psi_{0}(-l+\epsilon)&=-\frac{1}{\epsilon}+\psi_{0}(l+1)+\left(2\psi_{1}(1)-\psi_{1}(l+1)\right)\epsilon+o\left(\epsilon^2\right)\label{eq:pgna2}\\
\psi_{1}(-l+\epsilon)&=\frac{1}{\epsilon^{2}}-\psi_{1}(l+1)+\psi_{1}(1)+\zeta(2)+o\left(\epsilon\right).\label{eq:pgna3}
\end{align}

The resulting summation forms of $\mathrm{A_1}$, $\mathrm{A_2}$, $\mathrm{B_1}$, and $\mathrm{B_2}$ are summarized in~(\ref{eq:A1S})--(\ref{eq:B2S}) in Appendix A.

\subsubsection*{Case B: Fixed number of particles}\label{subsubsec2}
By the definition of correlation kernel~(\ref{eq:corrk}), the $\mathrm{I_{A}}$ integral~(\ref{eq:IA}) of case B boils down to computing the two parts
\begin{equation}\label{eq:IAexf}
\mathrm{I_{A}}=\mathcal{A}_{1}+\mathcal{A}_{2},
\end{equation}
where
\begin{align}
\mathcal{A}_1 =&\sum_{k=0}^{m-1}\frac{1}{h_{k}}\int_{-1}^{1}\left(\left(\frac{1+x}{2}\right)^{2}\ln^{2}\frac{1+x}{2}+\left(\frac{1-x}{2}\right)^{2}\ln^{2}\frac{1-x}{2}\right)\nonumber\\
 &\times\left(\frac{1-x}{2}\right)^{a}\left(\frac{1+x}{2}\right)^{b}J_{k}^{(a,b)}(x)^2\dd x\label{eq:fpA1ab}\\
\mathcal{A}_2=&\sum_{k=0}^{m-1}\frac{2}{h_{k}}\int_{-1}^{1}\left(\frac{1-x}{2}\right)^{a+1}\left(\frac{1+x}{2}\right)^{b+1}\ln\frac{1-x}{2}\ln\frac{1+x}{2}J_{k}^{(a,b)}(x)^2\dd x.\nonumber\\
 \label{eq:fpA2}
\end{align}
Due to the parity property~(\ref{eq:parity}), $\mathcal{A}_1$ admits a symmetric structure as
\begin{equation}\label{eq:fpA1}
\mathcal{A}_1=\mathcal{A}_1^{(a,b)}+\mathcal{A}_1^{(b,a)},
\end{equation}
where
\begin{equation}\label{eq:fpA1ab}
\mathcal{A}_1^{(a,b)}=\sum_{k=0}^{m-1}\frac{1}{h_{k}}\int_{-1}^{1}\left(\frac{1-x}{2}\right)^{a}\left(\frac{1+x}{2}\right)^{b+2}\ln^{2}\frac{1+x}{2}J_{k}^{(a,b)}(x)^2\dd x.
\end{equation}
The summations in~(\ref{eq:fpA2}) and (\ref{eq:fpA1ab}) can be evaluated by using the following confluent form of Christoffel-Darboux formula~\cite{Forrester}
\begin{equation}\label{eq:cd}
\sum_{k=0}^{m-1}\frac{J_k^{(a,b)}(x)^2}{h_k}=\alpha_1 J_{m-1}^{(a+1,b+1)}(x)
J_{m-1}^{(a,b)}(x)- \alpha_2 J_{m-2}^{(a+1,b+1)}(x) J_m^{(a,b)}(x),
\end{equation}
where
\begin{align}
\alpha_1&=\frac{m (a+b+m)(a+b+m+1)}{h_{m-1} (a+b+2 m-1)_2}\\
\alpha_2&=\frac{m (a+b+m)^2}{h_{m-1} (a+b+2 m-1)_2}.
\end{align}
Consequently, we have
\begin{align}
\mathcal{A}_1^{(a,b)}=~&\alpha_1\int_{-1}^1 \left(\frac{1-x}{2}\right)^a \left(\frac{1+x}{2}\right)^{b+2} \ln ^2\frac{1+x}{2}J_{m-1}^{(a+1,b+1)}(x) J_{m-1}^{(a,b)}(x) \dd x\nonumber\\
&\!-\alpha_2\int_{-1}^1 \left(\frac{1-x}{2}\right)^a \left(\frac{1+x}{2}\right)^{b+2} \ln ^2\frac{1+x}{2}J_{m-2}^{(a+1,b+1)}(x) J_{m}^{(a,b)}(x) \dd x\nonumber\\
\label{eq:A1cd}\\
\mathcal{A}_2 =~&2\alpha_1\int_{-1}^1 \left(\frac{1-x}{2}\right)^{a+1} \left(\frac{1+x}{2}\right)^{b+1} \ln \frac{1-x}{2}\ln \frac{1+x}{2}\nonumber\\
&\!\times J_{m-1}^{(a+1,b+1)}(x) J_{m-1}^{(a,b)}(x) \dd x\nonumber\\
&\!-2\alpha_2\int_{-1}^1 \left(\frac{1-x}{2}\right)^{a+1} \left(\frac{1+x}{2}\right)^{b+1} \ln \frac{1-x}{2}\ln \frac{1+x}{2}\nonumber\\
&\!\times  J_{m-2}^{(a+1,b+1)}(x) J_{m}^{(a,b)}(x) \dd x.\label{eq:A2cd}
\end{align}

Similarly, the $\mathrm{I_{B}}$ integral~(\ref{eq:IB}) can be written in terms of the following two integrals
\begin{equation}\label{eq:IBexf}
\mathrm{I_{B}}=\mathcal{B}_{1}+\mathcal{B}_{2},
\end{equation}
where
\begin{align}
\mathcal{B}_{1}=&\sum_{k=0}^{m-1}\frac{1}{h_{k}^{2}}\Bigg(\int_{-1}^{1}\left(\frac{1-x}{2}\ln\frac{1-x}{2}+\frac{1+x}{2}\ln\frac{1+x}{2}\right)\left(\frac{1-x}{2}\right)^{a}\left(\frac{1+x}{2}\right)^{b}\nonumber\\
&\times J_{k}^{(a,b)}(x)^{2}\dd x\Bigg)^{2}\\
\mathcal{B}_{2}=&\sum_{j=1}^{m-1}\sum_{k=0}^{m-j-1}\frac{2}{h_{k+j}h_{k}}\Bigg(\int_{-1}^{1}\left(\frac{1-x}{2}\ln\frac{1-x}{2}+\frac{1+x}{2}\ln\frac{1+x}{2}\right)~~~~~~~~~~~~\nonumber\\
&\times\left(\frac{1-x}{2}\right)^{a}\left(\frac{1+x}{2}\right)^{b}J_{k+j}^{(a,b)}(x)J_{k}^{(a,b)}(x)\dd x\Bigg)^{2}.
\end{align}
Note that the integrals in~$\mathcal{B}_{1}$ or~$\mathcal{B}_{2}$ also admit the same symmetric structure as~(\ref{eq:fpA1}) due to the parity property of Jacobi polynomials.
Because of the symmetry, the remaining computation process is similar to that of case A. Specifically, the integrals in $\mathcal{A}_1^{(a,b)}$ are calculated by taking twice derivative of $c$ of the identity~(\ref{eq:SIac2}), where one needs to assign
\begin{equation}\label{eq:sp3}
a_1=a,~~~b_1=b,~~~a_2=a+1,~~~b_2=b+1,~~~k_1=k_2=m-1,
\end{equation}
and
\begin{equation}\label{eq:sp4}
a_1=a,~~~b_1=b,~~~a_2=a+1,~~~b_2=b+1,~~~k_1=m,~~~k_2=m-2,
\end{equation}
respectively corresponding to the first and second integrals in~(\ref{eq:A1cd}) before setting $c=b+2$. By the result~(\ref{eq:A2cd}), the two integrals in $\mathcal{A}_2 $ are calculated by taking derivatives of $c$ and $d$ of identity~(\ref{eq:SIdc2}) respectively with the specializations~(\ref{eq:sp3}) and~(\ref{eq:sp4}), before setting $c=b+1$, $d=a+1$. The integral
\begin{equation}\label{eq:fpB1ab}
\int_{-1}^{1}\frac{1+x}{2}\ln\frac{1+x}{2}\left(\frac{1-x}{2}\right)^{a}\left(\frac{1+x}{2}\right)^{b}J_{k}^{(a,b)}(x)^{2}\dd x
\end{equation}
in $\mathcal{B}_{1}$ is calculated by specializing
\begin{equation}\label{eq:sp5}
a_1=a_2=a,~~~b_1=b_2=b,~~~ k_1=k_2=k
\end{equation}
in the identity~(\ref{eq:SIac2}), and taking derivative of $c$ before setting $c=b+1$.
The integral
\begin{equation}
\int_{-1}^{1}\frac{1+x}{2}\ln\frac{1+x}{2}\left(\frac{1-x}{2}\right)^{a}\left(\frac{1+x}{2}\right)^{b}J_{k+j}^{(a,b)}(x)J_{k}^{(a,b)}(x)\dd x
\end{equation}
in $\mathcal{B}_{2}$ is calculated by specializing
\begin{equation}\label{eq:sp6}
a_1=a_2=a,~~~b_1=b_2=b~~~ k_1=k+j,~~~k_2=k
\end{equation}
in the identity~(\ref{eq:SIac2}), and taking derivative of $c$ before setting $c=b+1$.

After resolving the indeterminacy of gamma and polygamma functions by using~(\ref{eq:pgna1})--(\ref{eq:pgna3}), one arrives at the summation representations~(\ref{eq:fA1S})--(\ref{eq:fB2S}) as listed in Appendix A.

\subsection{Tools for simplification of summations}\label{subsec3.2}
The major task in obtaining the variance formulas in propositions \ref{propav} and \ref{propfv} is to simplify the summation representations~(\ref{eq:A1S})--(\ref{eq:fB2S}) into the closed-form results~(\ref{f:apv}) and~(\ref{f:fpv}). We summarize the necessary simplification tools in this section, where the existing simplification framework is briefly reviewed in Section~\ref{subsubsec 3.2.1} and the new simplification framework is introduced in Section~\ref{subsubsec 3.2.2}.
\subsubsection{Existing simplification framework}\label{subsubsec 3.2.1}
The existing simplification tools have been utilized in various moments calculations over different ensembles, including the Hilbert-Schmidt ensemble~\cite{Wei17,Wei20,HWC21,Wei22}, the Bures-Hall ensemble~\cite{Wei20BH,LW21}, and the fermionic Gaussian ensemble~\cite{HW22}.

In the existing framework, we mainly have two types of finite sum identities. The first type is of the form
\begin{equation}\label{eq:ftp}
\sum_{i=1}^{m}i^c \psi_{j_1}^{b_1}(i+a_1)\psi_{j_2}^{b_2}(i+a_2)\cdots \psi_{j_m}^{b_m}(i+a_m),
\end{equation}
where $a$, $b$, $c$, $j$ are non-negative integers. The main idea in deriving the identities of this type of sums is to change the summation orders and make use of the obtained lower order summation formulas in a recursive manner.
For example, the summation
\begin{equation}\label{eq:type120}
\sum_{i=1}^{m} \psi_{0}^2(i+a)
\end{equation}
will be derived by using the identity~(\ref{eq:B1}) of the lower order summation $\sum_{i=1}^{m} \psi_{0}(i+a)$. Specifically, by using the finite sum form of the digamma function
\begin{equation}\label{eq:pl0}
\psi_{0}(l)=-\gamma+\sum_{k=1}^{l-1}\frac{1}{k},
\end{equation}
the summation~(\ref{eq:type120}) can be rewritten as
\begin{equation}\label{eq:type123}
\psi_{0}(a)\sum_{i=1}^{m}\psi_{0}(i+a)+\sum_{j=1}^{m}\frac{1}{a+j-1}\sum_{i=j}^{m}\psi_{0}(i+a),
\end{equation}
where we have changed the summation order of the double sum. The remaining sums can be simplified by using the lower order identity~(\ref{eq:B1}), leading to the result in~(\ref{eq:B30}).

The second type of summations is of the form
\begin{equation}\label{eq:stp}
S_f(m,n)=\sum_{i=1}^{m} \frac{(n-i)!}{(m-i)!}f(i),\qquad m\leq n,
\end{equation}
where $f(i)$ can be the product of polygamma and rational functions in $i$. A well-known result of the second type summation is the identity
\begin{equation}\label{eq:type20}
\sum _{i=1}^{m} \frac{(n-i)!}{(m-i)!}=\frac{n!}{(m-1)! (n-m+1)},
\end{equation}
which is a special case of the Chu-Vandermonde identity~\cite{Luke}. In the existing simplification framework, the identity~(\ref{eq:type20}) is a fundamental result in obtaining several other second type summations. For example, when
\begin{equation}
f(i)=\frac{1}{i},
\end{equation}
 the summation
\begin{equation}\label{eq:resumex1}
S_{1/i}(m,n)=\sum _{i=1}^{m} \frac{(n-i)!}{ (m-i)!}\frac{1}{i}
\end{equation}
is computed by first obtaining the recurrence relation
\begin{equation}\label{eq:type21}
S_{1/i}(m,n)=\frac{n}{m}S_{1/i}(m-1,n-1)+\frac{n-m}{m}\sum _{i=1}^{m}\frac{(n-1-i)!}{ (m-i)!}.
\end{equation}
After recurring $m$ times, and using the existing result~(\ref{eq:type20}), one obtains the closed-form result of the summation (\ref{eq:resumex1}) as listed in~(\ref{eq:B21}).

\subsubsection{New simplification framework}\label{subsubsec 3.2.2}
The existing simplification framework is useful in simplifying the summations in~(\ref{eq:A1S}),~(\ref{eq:B1S}),~(\ref{eq:fA1S}), and (\ref{eq:fB1S}). However, the framework is insufficient to simplify some of the summations in~(\ref{eq:A2S}),~(\ref{eq:B2S}),~(\ref{eq:fA2S}), and (\ref{eq:fB2S}). The main technique difficulty is that these summations, after exhausting all the possibilities of changing the order of summations, are not related to the first or second type summations of the existing framework. To convert these sums into the ones of the existing framework, the following new simplification framework is needed.

The first technique in the new framework is what we refer to as ``dummy summation". 
 The idea of the dummy summation is as follows. For a summation $F$,
\begin{equation}\label{eq:series1}
F=\sum_{i}f_i,
\end{equation}
where
\begin{equation}
\sum_{i}=\sum_{i_1}\sum_{i_2}\dots\sum_{i_a}
\end{equation}
denotes a finite nested sum of $a$ indexes $i=\{i_1,i_2,\dots,i_a\}$. Each index $i_k$, $k=1,2,\dots, a$ may depend on its previous ones  $i_1, i_2,\dots, i_{k-1}$.
For the summation $F$ that is not simplifiable under the existing framework by any changes of the order of summations, one may interpret the summand $f_i$ into an additional nested finite sum over a set of new indexes $j=\{j_1,j_2,\dots, j_b\}$ as
\begin{equation}\label{eq:dummya}
f_i=\sum_j g_{i, j},
\end{equation}
such that the resulting representation
\begin{equation}\label{eq:series2}
F=\sum_{i,j} g_{i, j}
\end{equation}
admits further simplifications by using the existing tools of first and second type sums~(\ref{eq:ftp}) and~(\ref{eq:stp}) after appropriately changing some of the summation orders. Note that the indexes in the set $j$ in~(\ref{eq:dummya}) may depend on the ones in the set $i$.

To illustrate this new technique, we consider the following two examples that will be utilized in the simplifications in Section \ref{subsec3.3}.
The first example is the summation
\begin{equation}\label{eq:dummye1}
F=\sum_{i_1=1}^m f_{i_1},
\end{equation}
where
\begin{equation}\label{eq:dummye1sd}
f_{i_1}= i_1 \psi_0(i_1+a).
\end{equation}
In this case, we  have $i=\{i_1\}$ in the definition~(\ref{eq:series1}).
By interpreting $i_1$ in the summand~(\ref{eq:dummye1sd}) as a dummy sum
\begin{equation}\label{eq:dummye1i}
i_1= \underbrace{1+1+\dots+1}_{i_1},
\end{equation}
the original single sum in~(\ref{eq:dummye1}) now becomes a double sum
\begin{equation}\label{eq:dummye1d}
F=\sum_{i_1=1}^m \sum_{j_1=1}^{i_1} g_{i_1,j_1},
\end{equation}
where
\begin{equation}
g_{i_1,j_1}=\psi_0(i_1+a)
\end{equation}
with $j=\{j_1\}$ in the definition~(\ref{eq:series2}). After changing the summation order in~(\ref{eq:dummye1d}) as
\begin{equation}\label{eq:dummy_ex1}
F=\sum_{j_1=1}^m \sum_{i_1=j_1}^m \psi_0(i_1+a),
\end{equation}
the sum over $i_1$ can be simplified by using an existing identity~(\ref{eq:B1}), leading to
\begin{equation}\label{eq:dummy_ex11}
F=\frac{1-a}{2} \sum _{j_1=1}^m \psi _0\left(j_1+a\right)+\frac{m}{4}\left(2 (a+m) \psi _0(a+m+1)-m-1\right).
\end{equation}
Simplifying the single summation in~(\ref{eq:dummy_ex11}) by the identity~(\ref{eq:B1}) directly gives the result listed in~(\ref{eq:B3}). In the above example, the dummy summation~(\ref{eq:dummye1i}) converts the original sum~(\ref{eq:dummye1}) into a double sum~(\ref{eq:dummye1d}), which appears more complicated but is simplifiable using an available identity~(\ref{eq:B1}) of the existing framework. Note also that the dummy summation technique may not be critical in this example, which may be derived in other ways.

In the second example, we introduce a dummy summation~(\ref{eq:dummya})  that is crucial as a subroutine in simplifying the summations in~(\ref{eq:A2S}) and~(\ref{eq:fA2S}). The essential idea in creating this dummy summation is to interpret a Gamma ratio
\begin{equation}
\frac{\Gamma (i)}{\Gamma (c+i)}=\frac{1}{i(i+1)\cdots(i+c-1)},\qquad c,i\in \mathbb{Z}^{+}
\end{equation}
as its partial fraction decomposition
\begin{equation}
a_1\frac{1}{i} + a_2 \frac{1}{i+1} +\dots + a_c \frac{1}{i+c-1}
\end{equation}
with  $a_j,j=1,2,\dots, c$ denoting the coefficients of the decomposition. The dummy summation that corresponds to the above interpretation is
\begin{equation}\label{eq:dumys}
\frac{\Gamma (i)}{\Gamma (c+i)}=\sum_{j=1}^c a_j\frac{1}{i+j-1},
\end{equation}
where
\begin{equation}
a_j= \frac{(-1)^{j+1}}{\Gamma (j) \Gamma (c-j+1)}.
\end{equation}
The coefficients $a_j$ are computed by evaluating a unit-argument hypergeometric function~\cite{Brychkov08}
\begin{equation}
\sum _{j=1}^c a_j\frac{1}{i+j-1}=\frac{\!_2F_1(-c,i-1;i;1)}{\Gamma (c+1)},
\end{equation}
where one utilizes the well-known identity~\cite{Brychkov08}
\begin{equation}\label{eq:2f1}
_2F_1(a,b;c;1)=\frac{\Gamma(c)\Gamma(c-a-b)}{\Gamma(c-a)\Gamma(c-b)}, \quad \Re(c)>\Re(a+b).
\end{equation}
The dummy summation in~(\ref{eq:dumys}) is critical in proving the lemmas \ref{lemma2}--\ref{lemma4} as discussed later in this section. By utilizing these lemmas, the summations in~(\ref{eq:A2S}) and~(\ref{eq:fA2S}) involving ratios of Gamma functions then become simplifiable under the existing simplification framework as will be shown in Section~\ref{subsec3.3}.

The second technique in the new framework is what we refer to as the ``re-summation" technique. The idea of this new technique is as follows. For a finite summation $G_m$ that appears not summable under the existing framework with $m$ denoting one of its parameters or the finite upper limit of the summation. The re-summation technique aims to find alternative forms of the summation $G_m$ to reveal the potential cancellations with other sums, such that the remaining terms can be simplified by using the existing tools. Specifically, the re-summation of $G_m$ can be generated by iterating a suitably chosen recurrence relation
\begin{equation}\label{eq:re-sumr}
G_m=c_{m-1}G_{m-1}+r_{m-1},
\end{equation}
where $c_{i}$ denotes the coefficient and $r_{i}$ is the remainder of the recurrence relation. Here, the remainder can contain summations, see the example ~(\ref{eq:re-sumer}) below. Each iteration is to replace the term $G_{m-i}$ with its previous one $G_{m-i-1}$. Keep iterating until $G_{m-i}$ vanishes, we then obtain an alternative form of $G_m$, which is considered a useful one if it facilitates the cancellation with other sums.

To illustrate the re-summation technique, we consider the difference of two summations
\begin{equation}\label{ex:resumlc}
\mathcal{G}=G_m-G^\prime_m,
\end{equation}
where
\begin{align}
G_m&=\sum _{i=1}^m \frac{(-1)^i \Gamma (a-i+m+1)}{\Gamma (m-i+1)(a-i+2 m+1)}\label{ex:resums1}\\
G^\prime_m&=\sum _{i=1}^m \frac{(-1)^i \Gamma (a-i+m+1)}{\Gamma (m-i+1)i}\label{ex:resums2}.
\end{align}
We first notice that the two summations in~(\ref{ex:resums1}) and (\ref{ex:resums2}) above can not be simplified individually by using the identities of the existing framework. In addition, these two summations do not appear to cancel directly in the combination~(\ref{ex:resumlc}). By iterating $m$ times the following tailor made recurrence relations (in terms of the various choices of coefficients and reminders), cf.~(\ref{eq:re-sumr}),
\begin{align}
G_m&=c_{m-1}G_{m-1}+r_{m-1}\\
G^\prime_m&=c_{m-1}G^\prime_{m-1}+r^\prime_{m-1},
\end{align}
where
\begin{align}
c_{m-1}&=\frac{a+m}{m}\\
r_{m-1}&=-\frac{a}{m}\sum _{i=1}^{m+1} \frac{(-1)^i \Gamma (a-i+m+1)}{\Gamma (m-i+2)}+\frac{a+m}{m}\left(\frac{\Gamma (a+m)}{\Gamma (m) (a+2 m-1)}\right.\nonumber\\
&~~~\!~\!\left.-\frac{\Gamma (a+m+1)}{\Gamma (m+1) (a+2 m)}\right)\label{eq:re-sumer}\\
r^\prime_{m-1}&=\frac{a}{m}\sum _{i=1}^m \frac{(-1)^i \Gamma (a-i+m)}{\Gamma (m-i+1)},
\end{align}
we obtain the re-summation results
\begin{align}
G_m=~\!&\frac{a \Gamma (a+m+1)}{\Gamma (m+1)}\sum _{i=1}^m \sum _{j=1}^{i+1} \frac{(-1)^{j+1}\Gamma (i) \Gamma (a+i-j+1)}{\Gamma (i-j+2) \Gamma (a+i+1)}\nonumber\\
&+\frac{\Gamma (a+m+1)}{\Gamma (m+1)} \sum _{i=1}^m \left(\frac{1}{a+2 i-1}-\frac{a+i}{i (a+2 i)}\right)\label{eq:resumr1}\\
G^\prime_m=~\!&\frac{a \Gamma (a+m+1)}{\Gamma (m+1)} \sum _{i=1}^m \sum _{j=2}^{i+1} \frac{(-1)^{j+1} \Gamma (i) \Gamma (a+i-j+1)}{\Gamma (i-j+2) \Gamma (a+i+1)}\label{eq:resumr2}.
\end{align}
The relationship between the two summations now become obvious. After inserting the results~(\ref{eq:resumr1}) and (\ref{eq:resumr2}) into~(\ref{ex:resumlc}), the cancellation occurs directly between the two double summations in~(\ref{eq:resumr1}) and (\ref{eq:resumr2}). The remaining terms give the closed-form result
\begin{align}\label{eq:lemma5c=0}
&\sum _{i=1}^m \frac{(-1)^i \Gamma (a-i+m+1)}{\Gamma (m-i+1)}\left(\frac{1}{ a-i+2 m+1}-\frac{1}{i}\right)\nonumber\\
& =\frac{\Gamma (a+m+1)}{\Gamma (m+1)} \left(\psi _0(a+2 m+1)-\psi _0(a+m+1)\right),
\end{align}
which is a special case of Lemma \ref{lemma5} when $c=0$. This result is also utilized in simplifying the double summations in~(\ref{eq:fB2S}). Note that the emphasis here is the inner cancellation among multiple sums when utilizing re-summation technique. On the other hand, this technique concerning only one summation has been studied in~\cite{Wei17,Wei20,HWC21}.


Using the two techniques in the new simplification framework, we obtain the following six lemmas. More precisely, proving the Lemma~\ref{lemma1} utilizes the re-summation technique. The lemmas \ref{lemma2}--\ref{lemma4} are obtained by using the dummy summation together with the result of Lemma~\ref{lemma1}. The Lemma~\ref{lemma5} and Lemma~\ref{lemma6} are established by using the re-summation technique. In the six lemmas obtained by the new framework, each single sum is  converted to another one, which will be applied in Section \ref{subsec3.3} to convert summations involving multiple ratios of Gamma functions into the ones that are simplifiable by the existing framework.

\begin{lemma}\label{lemma1}
For any complex numbers $a,b,c\notin \mathbb{Z}^{-}$,
we have
\begin{align}\label{eq:lemma1}
& \sum _{i=1}^m \frac{1}{\Gamma (i) \Gamma (a+i) \Gamma (m+1-i) \Gamma (m+b+1-i)(c+i) }\nonumber \\
 &=\frac{1}{\Gamma (b+m) \Gamma (c+m+1) \Gamma (a+b+m)} \sum _{i=1}^m \frac{\Gamma (c-i+m+1) \Gamma (a+b-i+2 m)}{\Gamma (m-i+1) \Gamma (a-i+m+1)}.
\end{align}
\end{lemma}
\begin{proof}
Proving the Lemma~\ref{lemma1} uses the re-summation technique. The  left side of the identity~(\ref{eq:lemma1}) can be rewritten as
\begin{equation}\label{eq:lemma1n}
G_{m}=\sum _{i=1}^{m} \frac{1}{\Gamma (i) \Gamma (a+i) \Gamma (m-i) \Gamma (m+b+1-i)(m-i)(c+i) }.
\end{equation}
In (\ref{eq:lemma1n}), by performing a partial fraction decomposition
\begin{equation}
\frac{1}{(m-i)(c+i)}=\frac{1}{c+m}\left(\frac{1}{c+i}+\frac{1}{m-i}\right),
\end{equation}
we obtain the recurrence relation, cf.~(\ref{eq:re-sumr}),
\begin{equation}\label{eq:lemma1r}
G_m=c_{m-1}G_{m-1}+r_{m-1},
\end{equation}
where
\begin{align}
c_{m-1}=&\frac{1}{c+m}\\
r_{m-1}=&\frac{1}{c+m}\sum _{i=1}^m \frac{1}{\Gamma (i) \Gamma (a+i) \Gamma (m-i+1) \Gamma (m+b+1-i)}\label{eq:lemma1rm}.
\end{align}
The summation in~(\ref{eq:lemma1rm}) is related to a hypergeometric function $\, _2F_1(1-b-m,1-m;a+1;1)$ that admits a closed-form representation by using the identity~(\ref{eq:2f1}). Therefore, we have
\begin{align}
&\sum _{i=1}^m \frac{1}{\Gamma (i) \Gamma (a+i) \Gamma (m-i+1) \Gamma (m+b+1-i)}\nonumber\\
&=\frac{\Gamma (a+b+2m-1)}{\Gamma (m) \Gamma (a+m) \Gamma (b+m) \Gamma (a+b+m) }, \label{eq:ofgi}
\end{align}
and
\begin{equation}\label{eq:lemma1r0}
r_{m-1}=\frac{\Gamma (a+b+2m-1)}{(c+m)\Gamma (m) \Gamma (a+m) \Gamma (b+m) \Gamma (a+b+m) }.
\end{equation}
Iterate $m$ times the recurrence relation~(\ref{eq:lemma1r}), the desired identity~(\ref{eq:lemma1}) is established. This completes the proof of Lemma~\ref{lemma1}.
\end{proof}

\begin{lemma}\label{lemma2}
For any complex numbers $a,b\notin \mathbb{Z^{-}}$, and any $c\in \mathbb{Z^{+}}$, we have
\begin{align}
&\sum _{i=1}^m \frac{1}{\Gamma (c+i) \Gamma (a+i) \Gamma (m+1-i) \Gamma (m+b+1-i) }\nonumber\\
&=\frac{1}{\Gamma (m+b) \Gamma (m+a+b) \Gamma (c) \Gamma (m+c)}\nonumber\\
&~~~~\times\sum _{i=1}^m \frac{\Gamma (m+a+b+i-1) \Gamma (m+c-i)}{\Gamma (a+i) \Gamma (m-i+1)}\label{eq:lemma2}.
\end{align}
\end{lemma}
\begin{proof}
We prove Lemma~\ref{lemma2} by using the dummy summation technique.
Denote
\begin{equation}\label{eq:lemma2f}
F=\sum _{i=1}^m f_i
\end{equation}
with
\begin{equation}
f_i= \frac{1}{\Gamma (c+i) \Gamma (a+i) \Gamma (m-i+1) \Gamma (b+m-i+1)}.
\end{equation}
We first introduce an additional summation in~(\ref{eq:lemma2f}) by replacing $\Gamma (i)/\Gamma (c+i)$ with the summation form in~(\ref{eq:dumys}).
The summation $F$ now becomes
\begin{equation}\label{eq:lemma2d}
F=\sum _{i=1}^m\sum _{j=1}^c g_{i,j}
\end{equation}
with
\begin{equation}
g_{i,j} =\frac{(-1)^{j+1}}{\Gamma (j) \Gamma (c-j+1)(i+j-1)} \frac{1}{\Gamma (i) \Gamma (a+i) \Gamma (m-i+1)\Gamma (b+m-i+1) }.\label{eq:lemma2ds}
\end{equation}
In the double summation~(\ref{eq:lemma2d}), we apply the Lemma~\ref{lemma1} to first evaluate the sum over $i$, leading to
\begin{align}
F=&\frac{1}{\Gamma (b+m) \Gamma (a+b+m)}\sum _{i=1}^m \frac{\Gamma (a+b-i+2 m)}{\Gamma (m-i+1) \Gamma (a-i+m+1)}\nonumber\\
&\times \sum _{j=1}^c \frac{(-1)^{j+1} \Gamma (m-i+j)}{\Gamma (j) \Gamma (c-j+1)\Gamma (j+m)},\label{eq:lemma22}
\end{align}
where it is now possible to evaluate the sum over j as
\begin{align}
&\sum _{j=1}^c \frac{(-1)^{j+1} \Gamma (m-i+j)}{\Gamma (j) \Gamma (c-j+1) \Gamma (j+m)}\nonumber\\
&=\frac{\Gamma (m-i+1)}{\Gamma (c) \Gamma (m+1)} \, _2F_1(1-c,m-i+1;m+1;1).\label{eq:lemma2c}\\
&=\frac{\Gamma (c+i-1) \Gamma (m-i+1)}{\Gamma (c) \Gamma (i) \Gamma (c+m)}.\label{eq:lemma2c1}
\end{align}
The result (\ref{eq:lemma2c1}) is obtained by using the identity~(\ref{eq:2f1}) to evaluate the unit argument hypergeometric function in~(\ref{eq:lemma2c}). Inserting~(\ref{eq:lemma2c1}) into~(\ref{eq:lemma22}) and shifting the index $i\to m-i+1$, Lemma \ref{lemma2} is proved.
\end{proof}

\begin{lemma}\label{lemma3}
For any complex numbers $a,b\notin \mathbb{Z^{-}}$, and any $c\in \mathbb{Z^{+}}$, we have
\begin{align}
 &\sum _{i=1}^m \frac{1}{\Gamma (c+i) \Gamma (a+i) \Gamma (m-i+1) \Gamma (b-i+m+1) i }\nonumber \\
 &=\frac{1}{\Gamma (a) \Gamma (a+m) \Gamma (1+b+m) \Gamma (b+c+m)}\sum _{i=1}^m \frac{\Gamma (a-i+m) \Gamma (b+c+i+m)}{ \Gamma (c+i) \Gamma (m-i+1)i}\nonumber\\
 &~~~\!~+\frac{\psi_0 (a)-\psi_0 (a+m)}{\Gamma (a) \Gamma (c) \Gamma (m+1) \Gamma (b+m+1)}.\label{eq:lemma3}
\end{align}
\end{lemma}
\begin{proof}
Similar to the proof of Lemma~\ref{lemma2}, we introduce an additional summation in
\begin{equation}
F=\sum _{i=1}^m \frac{1}{\Gamma (c+i) \Gamma (a+i) \Gamma (m-i+1) \Gamma (b-i+m+1) i }\label{eq:lemma3f}
\end{equation}
by replacing the term $\Gamma (i)/\Gamma (c+i)$ with the dummy summation~(\ref{eq:dumys}). Consequently, one has
\begin{equation}
F=\sum _{i=1}^m\sum _{j=1}^c g_{i,j}\label{eq:lemma3d}
\end{equation}
with
\begin{equation}
g_{i,j}= \frac{(-1)^{j+1} }{\Gamma(j) \Gamma(c-j+1)(i+j-1)} \frac{1}{\Gamma (i+1) \Gamma (a+i) \Gamma (m-i+1) \Gamma (b-i+m+1)},\label{eq:lemma3ds}
\end{equation}
which is the same form as the summand in~(\ref{eq:lemma2ds}).
Therefore, one uses the results~(\ref{eq:lemma2d})--(\ref{eq:lemma2c})  to obtain the identity~(\ref{eq:lemma3}), and Lemma \ref{lemma3} is proved.
\end{proof}

\begin{lemma}\label{lemma4}
For any complex numbers $a,b\notin \mathbb{Z^{-}}$, and any $c, d\in \mathbb{Z^{+}}$, we have
\begin{align}
&\sum _{i=1}^m \frac{1}{\Gamma (c+i) \Gamma (a+i) \Gamma (d+m-i+1) \Gamma (b+m-i+1)}\nonumber\\
&=\frac{1}{\Gamma (d) \Gamma (a+m) \Gamma (a+b+m) \Gamma (c+d+m)}\sum _{i=1}^m \frac{\Gamma (c+d+i-1) \Gamma (a+b-i+2 m)}{\Gamma (c+i) \Gamma (b-i+m+1)}\nonumber\\
&~~~+\frac{1}{\Gamma (c) \Gamma (b+m) \Gamma (a+b+m) \Gamma (c+d+m)}\sum _{i=1}^m \frac{\Gamma (c+d+i-1) \Gamma (a+b-i+2 m)}{\Gamma (d+i) \Gamma (a-i+m+1)}.\nonumber\\
\label{eq:lemma4}
\end{align}
\end{lemma}
\begin{proof}
Proving the Lemma \ref{lemma4} also utilizes the dummy summation technique. Denote the left side of the identity~(\ref{eq:lemma4}) as
\begin{equation}\label{eq:lemma4f}
F=\sum _{i=1}^m f_i,
\end{equation}
with
\begin{equation}
f_i= \frac{1}{\Gamma (c+i) \Gamma (a+i) \Gamma (d+m-i+1) \Gamma (b+m-i+1)}.
\end{equation}
By interpreting respectively the gamma ratios $\Gamma (i)/\Gamma (c+i)$ and $\Gamma (m-i+1)/\Gamma (d+m-i+1)$  as the dummy summations~(\ref{eq:dumys}) and
\begin{equation}
\sum _{k=1}^d \frac{(-1)^{k+1}}{\Gamma (k)\Gamma (d-k+1)(m-i+k)},
\end{equation}
the summation $F$ can be written as
\begin{equation}\label{eq:lemma4d}
F=\sum _{i=1}^m\sum _{j=1}^c\sum _{k=1}^d g_{i,j,k},
\end{equation}
where
\begin{align}\label{eq:lemma4g}
g_{i,j,k}=~\!\!&\frac{1}{\Gamma (i)  \Gamma (a+i) \Gamma (m-i+1) \Gamma (b-i+m+1)}\nonumber\\
&\times \frac{(-1)^{j+1}}{\Gamma (j) \Gamma (c-j+1) (i+j-1)} \frac{(-1)^{k+1}}{\Gamma (k) \Gamma (d-k+1)(m-i+k)}.
\end{align}
After taking a partial fraction decomposition
\begin{equation}
\frac{1}{(m-i+k) (i+j-1)}=\frac{1}{(i+j-1) (j+k+m-1)}+\frac{1}{(k+m-i) (j+k+m-1)}
\end{equation}
in~(\ref{eq:lemma4g}), the summation~(\ref{eq:lemma4d}) is split into two sums
\begin{equation}
F=F_1+F_2,\label{eq:lemma4split}
\end{equation}
where
\begin{align}
F_1=~\!\!&\sum _{j=1}^c \frac{(-1)^{j+1}}{\Gamma (j) \Gamma (c-j+1)}\sum _{k=1}^d \frac{(-1)^{k+1}}{\Gamma (k) \Gamma (d-k+1) (j+k+m-1)}\nonumber\\
&\times\sum _{i=1}^m \frac{1}{ \Gamma (i) \Gamma (a+i) \Gamma (m-i+1) \Gamma (b-i+m+1)(i+j-1)}\label{eq:lemma4sc}\\
F_2=~\!\!&\sum _{k=1}^d \frac{(-1)^{k+1}}{\Gamma (k) \Gamma (d-k+1)}\sum _{j=1}^c \frac{(-1)^{j+1}}{\Gamma (j) \Gamma (c-j+1) (j+k+m-1)}\nonumber\\
&\times\sum _{i=1}^m \frac{1}{\Gamma (i) \Gamma (b+i) \Gamma (m-i+1) \Gamma (a-i+m+1) (i+k-1)}.
\end{align}
Specialize $c=d$ and $i=j+m$ in the identity~(\ref{eq:dumys}), one has
\begin{equation}
\sum _{k=1}^d \frac{(-1)^{k+1}}{\Gamma (k) \Gamma (d-k+1) (j+k+m-1)}=\frac{\Gamma (j+m)}{\Gamma (d+j+m)}.
\end{equation}
Insert the above result into~(\ref{eq:lemma4sc}), and evaluate the resulting summation over $i$ by the identity~(\ref{eq:lemma1}), one arrives at
\begin{align}
F_1=~\!\!&\frac{1}{\Gamma (b+m) \Gamma (a+b+m)}\sum _{j=1}^c \frac{(-1)^{j+1} }{\Gamma (j) \Gamma (c-j+1) \Gamma (d+j+m)}\nonumber\\
&\times\sum _{i=1}^m \frac{\Gamma (m+j-i)\Gamma (a+b-i+2 m)}{\Gamma (m-i+1) \Gamma (a-i+m+1)},
\end{align}
where the summation over $j$ is simplified, by using the identity~(\ref{eq:2f1}), to
\begin{align}
&\sum _{j=1}^c \frac{(-1)^{j+1} \Gamma (m-i+j)}{\Gamma (j) \Gamma (c-j+1) \Gamma (d+j+m)}\nonumber\\
&=\frac{\Gamma (m-i+1)}{\Gamma (c) \Gamma (d+m+1)} \, _2F_1(1-c,m-i+1;d+m+1;1)\\
&=\frac{\Gamma (m-i+1) \Gamma (c+d+i-1)}{\Gamma (c) \Gamma (d+i) \Gamma (c+d+m)}.
\end{align}
 Consequently, we have
\begin{align}
F_1=\!\!~&\frac{1}{\Gamma (c) \Gamma (b+m) \Gamma (a+b+m) \Gamma (c+d+m)}\nonumber\\
&\times\sum _{i=1}^m \frac{\Gamma (c+d+i-1) \Gamma (a+b-i+2 m)}{\Gamma (d+i) \Gamma (a-i+m+1)}.\label{eq:lemma4sc2}
\end{align}
In the same manner, we obtain
\begin{align}
F_2=\!\!~&\frac{1}{\Gamma (d) \Gamma (a+m) \Gamma (a+b+m) \Gamma (c+d+m)}\nonumber\\
&\times\sum _{i=1}^m \frac{\Gamma (c+d+i-1) \Gamma (a+b-i+2 m)}{\Gamma (c+i) \Gamma (b-i+m+1)}.\label{eq:lemma4sd2}
\end{align}
Insert the results~(\ref{eq:lemma4sc2})--(\ref{eq:lemma4sd2}) into~(\ref{eq:lemma4split}), we complete the proof of Lemma~\ref{lemma4}.
\end{proof}
\begin{lemma}\label{lemma5}
Denote
\begin{equation}\label{eq:lemma56n}
\Phi_{a,b,c,d}^{(x)}=\frac{\Gamma (x+c+1) \Gamma (x+d+1) }{\Gamma (x+a+1) \Gamma (x+b+1)},
\end{equation}
we have
\begin{align}
&\sum _{i=1}^m \frac{(-1)^i \Gamma (a-i+m+1)}{\Gamma (m-i+1)}\left(\frac{1}{ a+c-i+2 m+1}-\frac{1}{c+i}\right)\nonumber\\
&=\frac{ \Gamma (a+c+m+1) }{\Gamma (c+m+1)}\sum _{i=1}^m\left(\frac{\Phi _{0,a+c,a,c}^{(m-i)}}{a+c-2 i+2 m+1}+a\Phi _{1,a+c+1,a,c}^{(m-i)} \right.\nonumber\\
&~~~~\left.\!-\frac{\Phi _{1,a+c,a+1,c}^{(m-i)}}{a+c-2 i+2 m+2}\right),\label{eq:lemma5}
\end{align}
where it is sufficient to consider $\Re(a,c)\geq 0 $.
\end{lemma}

\begin{proof}
We prove the Lemma \ref{lemma5} by using the re-summation technique. We denote the left side of the identity~(\ref{eq:lemma5}) as
\begin{equation}\label{lemma5s}
\mathcal{G}=G_m-G^\prime_m,
\end{equation}
where
\begin{align}
G_m=&\sum _{i=1}^m \frac{(-1)^i \Gamma (a-i+m+1)}{\Gamma (m-i+1) (a+c-i+2 m+1)}\label{eq:lemma5rr1}\\
G^\prime_m=&\sum _{i=1}^m \frac{(-1)^i \Gamma (a-i+m+1)}{ \Gamma (m-i+1)(c+i)}\label{eq:lemma5rr2}.
\end{align}
Here, we choose the recurrence relations
\begin{align}
G_m=&c_{m-1}G_{m-1}+r_{m-1}\label{eq:lemma5r1}\\
G^\prime_m=&c_{m-1}G^\prime_{m-1}+r^{\prime}_{m-1},\label{eq:lemma5r2}
\end{align}
where
\begin{align}
c_{m-1}=~\!\!&\frac{a+c+m}{c+m}\\
r_{m-1}=~\!\!&\frac{a}{c+m} \sum _{i=0}^{m} \frac{(-1)^{i} \Gamma (a-i+m)}{\Gamma (m-i+1)}+\frac{(a+c+m) \Gamma (a+m)}{(c+m) \Gamma (m+1)}\nonumber\\
&\times\left(\frac{m}{a+c+2 m-1}-\frac{a+m}{a+c+2 m}\right)\\
r^{\prime}_{m-1}=~\!\!&\frac{a}{c+m}\sum _{i=1}^m \frac{(-1)^i \Gamma (a-i+m)}{\Gamma (m-i+1)}.
\end{align}
Iterating $m$ times the recurrence relations~(\ref{eq:lemma5r1}) and (\ref{eq:lemma5r2}), we obtain
\begin{align}
G_m=~\!\!&\frac{ \Gamma (a+c+m+1) }{\Gamma (c+m+1)}\left(a\sum _{i=1}^m\! \right. \sum _{j=0}^{m-i+1} (-1)^j \Phi _{a+c+1,1-j,c,a-j}^{(m-i)}\nonumber\\
&+\!\left.\sum _{i=1}^m \left(\frac{\Phi _{0,a+c,c,a}^{(m-i)}}{a+c-2 i+2 m+1}-\frac{\Phi _{1,a+c,c,a+1}^{(m-i)}}{a+c-2 i+2 m+2}\right)\right)\label{eq:lemma5s1}\\
G^\prime_m=~\!\!&\frac{a \Gamma (a+c+m+1) }{\Gamma (c+m+1)}\sum _{i=1}^m \sum _{j=1}^{m-i+1}(-1)^j \Phi _{a+c+1,1-j,c,a-j}^{(m-i)}\label{eq:lemma5s2}.
\end{align}
It is noticed that the double summation in~(\ref{eq:lemma5s1}) is  the same form as the one in~(\ref{eq:lemma5s2}).
By inserting the re-summations~(\ref{eq:lemma5s1}) and~(\ref{eq:lemma5s2}) into~(\ref{lemma5s}), we obtain the desired identity~(\ref{eq:lemma5}). This proves Lemma~\ref{lemma5}.
\end{proof}
\begin{lemma}\label{lemma6}
Using the notation~(\ref{eq:lemma56n}), we have
\begin{align}
& \sum _{i=1}^m\Phi _{0,a,b,a+b}^{(m-i)} \left(\frac{1}{a+b+c-i+2 m+1}-\frac{1}{c+i}\right)\nonumber\\
&=\Phi _{0,a,b,a+b}^{(m+c)}\sum _{i=1}^m \Phi _{0,a,b,a+b}^{(i-1)} \Phi _{b,a+b,0,a}^{(c+i-1)}\left(\frac{1}{a+b+c+2 i-1}-\frac{b (a+b)}{a i (a+b+c+i)}\right.\nonumber\\
&~~~~\!\left.-\frac{b (a-b)}{a (a+i) (b+c+i)}+\frac{(b+i) (a+b+i)}{i (a+i) (a+b+c+2 i)}-\frac{a+b+2 i-2}{(b+i-1) (a+b+i-1)}\right)\nonumber\\
&~~~~+\frac{1}{b}\Phi _{0,a,b,a+b}^{(m-1)},\label{eq:lemma6}
\end{align}
where it is sufficient to consider $a ,b,c, a+b, a+b+c \notin \mathbb{Z}^{-} $.
\end{lemma}
\begin{proof}
The proof of Lemma \ref{lemma6} uses the re-summation technique. The left side of the identity~(\ref{eq:lemma6}) can be written as
\begin{equation}\label{eq:lemma6s}
\mathcal{G}=G_m-G^\prime_m,
\end{equation}
where
\begin{align}
G_m&=\sum _{i=1}^m \Phi _{0,a,b,a+b}^{(m-i)}\frac{1}{a+b+c-i+2 m+1}\\
G^\prime_m&=\sum _{i=1}^m \Phi _{0,a,b,a+b}^{(m-i)}\frac{1}{c+i}.
\end{align}
In relating the two summations above, we choose the following recurrence relations
\begin{align}
G_m&=c_{m-1}G_{m-1}+r_{m-1}\label{eq:lemma6r1}\\
G^\prime_m&=c_{m-1}G^\prime_{m-1}+r^{\prime}_{m-1},\label{eq:lemma6r2}
\end{align}
where
\begin{align}
c_{m-1}=~\!\!&\frac{(b+c+m) (a+b+c+m)}{(c+m) (a+c+m)}\\
r_{m-1}=~\!\!&c_{m-1} \left(\frac{\Phi _{0,a,b,a+b}^{(m-1)}}{a+b+c+2 m-1}+\frac{\Phi _{0,a,b,a+b}^{(m)}}{a+b+c+2 m}\right)-\frac{b}{a}c_{m-1}\nonumber\\
&\times\left(\frac{(a-b) }{b+c+m}\sum _{i=1}^m \Phi _{0,a+1,b,a+b}^{(m-i)}+\frac{(a+b) }{a+b+c+m}\sum _{i=1}^{m+1} \Phi _{1,a,b,a+b}^{(m-i)}\right)\\
r^{\prime}_{m-1}=~\!\!&\frac{b (a-b) }{a (a+c+m)}\sum _{i=1}^{m-1} \Phi _{0,a+1,b,a+b}^{(m-i-1)}+\frac{b (a+b) }{a (c+m)}\sum _{i=1}^m \Phi _{1,a,b,a+b}^{(m-i-1)}.
\end{align}
Iterating $m$ times the above relations~(\ref{eq:lemma6r1}) and (\ref{eq:lemma6r2}), one obtains
\begin{align}
G_m=~\!\!&\Phi _{0,a,b,a+b}^{(m+c)}\left(\rule{0cm}{0.67cm}\!-\frac{b (a+b)}{a}\sum _{i=1}^m \Phi _{1,a,b,a+b}^{(i-2)}\sum _{j=1}^{m-i+1} \Phi _{a+b+1,b,0,a}^{(m-j+c)}-\frac{b (a-b)}{a}\right.\nonumber\\
&\times\sum _{i=1}^m \Phi _{0,a+1,b,a+b}^{(i-1)}\sum _{j=1}^{m-i+1} \Phi _{b+1,a+b,a,0}^{(m-j+c)}+\sum _{i=1}^m \Phi _{b,a+b,0,a}^{(m-i+c)}\nonumber\\
&\times\left(\frac{\Phi _{0,a,b,a+b}^{(m-i+1)}}{a+b+c-2 i+2 m+2}+\frac{\Phi _{0,a,b,a+b}^{(m-i)}}{a+b+c-2 i+2 m+1}\right.\nonumber\\
&\left.\!-\frac{b (a+b) \Phi _{1,a,b,a+b}^{(m-i)}}{a (a+b+c-i+m+1)}\right)\!\left.\rule{0cm}{0.67cm}\right)\label{eq:lemma6s1}\\
G^\prime_m=~\!\!&\Phi _{0,a,b,a+b}^{(m+c)}\left(\frac{b (a+b)}{a}\sum _{i=1}^m \Phi _{1,a,b,a+b}^{(i-2)}\sum _{j=1}^{m-i+1} \Phi _{a+b+1,b+1,0,a+1}^{(m-j+c)}\right.\nonumber\\
&+\!\left.\frac{b (a-b)}{a}\sum _{i=1}^m \Phi _{0,a+1,b,a+b}^{(i-2)}\sum _{j=1}^{m-i+1} \Phi _{b+1,a+b+1,a,1}^{(m-j+c)}\right)\label{eq:lemma6s2}.
\end{align}
Note that in order to reveal the potential cancellations between the two re-summations above, one will also need the identity
\begin{align}
&(c_1-b_1) \sum _{j=1}^m \Phi _{a_1,b_1+1,c_1,d_1+1}^{(m-j)}+(d_1-a_1+1) \sum _{j=1}^m \Phi _{a_1,b_1,c_1,d_1}^{(m-j)}\nonumber\\
&=\Phi _{a_1-1,b_1,c_1,d_1}^{(m)}-\Phi _{a_1-1,b_1,c_1,d_1}^{(0)}\label{eq:lemma6I}, \qquad a_1,b_1,c_1,d_1\notin\mathbb{Z}^{-}.
\end{align}
This identity is obtained by iterating $m$ times the following recurrence relation
\begin{equation}
H_{m,a_1,d_1}=H_{m-1,a_1+1,d_1+1}+s_{m-1},
\end{equation}
where
\begin{equation}
H_{m,a_1,d_1}=\sum _{j=1}^m \Phi _{a_1,b_1,c_1,d_1}^{(m-j)}
\end{equation}
\begin{equation}
s_{m-1}=H_{m,a_1,d_1}-H_{m-1,a_1+1,d_1+1}=(c_1-b_1) \sum _{j=1}^{m-1} \Phi _{a_1,b_1,c_1-1,d_1}^{(m-j)}+\Phi _{a_1,b_1,c_1,d_1}^{(0)}.
\end{equation}
We now use the identity~(\ref{eq:lemma6I}) with the specializations
\begin{align}
&a_1=a+b+c+i,~~~b_1=b+c+i-1,~~~c_1=c+i-1,\nonumber\\
&d_1=a+c+i-1, ~~~m\to m-i+1
\end{align}
and
\begin{align}
~~~~~~&a_1=b+c+i,~~~b_1=a+b+c+i-1,~~~c_1=a+c+i-1,\nonumber\\
&d_1=c+i-1,~~~m\to m-i+1,
\end{align}
the result (\ref{eq:lemma6s2}) becomes
\begin{align}
G^\prime_m=&~\Phi _{0,a,b,a+b}^{(c+m)}\Bigg(\!-\frac{b (a+b)}{a}\sum _{i=1}^m \Phi _{1,a,b,a+b}^{(i-2)} \sum _{j=1}^{m-i+1} \Phi _{a+b+1,b,0,a}^{(m-j+c)}-\frac{b (a-b)}{a}\nonumber\\
&\!~\times\sum _{i=1}^m \Phi _{0,a+1,b,a+b}^{(i-2)} \sum _{j=1}^{m-i+1} \Phi _{b+1,a+b,a,0}^{(m-j+c)}+\sum _{i=1}^m \left(\Phi _{a+b,b,0,a}^{(c+i-1)}-\Phi _{a+b,b,0,a}^{(c+m)}\right)\nonumber\\
&\!~\times\left(\frac{a-b }{a}\Phi _{0,a+1,b,a+b}^{(i-2)}+\frac{a+b}{a}\Phi _{1,a,b,a+b}^{(i-2)}\right)\!\Bigg),
\end{align}
where the double summations are of the same form as the ones in~(\ref{eq:lemma6s1}). Therefore,
all the double summations are cancelled completely in~(\ref{eq:lemma6s}). The remaining terms are the desired result~(\ref{eq:lemma6}). This completes the proof of Lemma~\ref{lemma6}.
\end{proof}

Note that the results in lemmas~\ref{lemma2}--\ref{lemma4} are analytically continued to any complex number $c$ or $d$ except for negative integers. This fact allows us to take derivatives of the formulas~(\ref{eq:lemma2}),~(\ref{eq:lemma3}), and~(\ref{eq:lemma4}) in obtaining the identities (\ref{eq:Bn0})--(\ref{eq:Bn6}) listed in Appendix \ref{secB2}. For example, the identity~(\ref{eq:Bn4}) is established by taking derivatives of $c$ and $d$ of the formula~(\ref{eq:lemma4}) before setting $c=d=0$. The lemmas~\ref{lemma1}--\ref{lemma4} along with the identities (\ref{eq:Bn0})--(\ref{eq:Bn6}) are useful in simplifying the summations in~(\ref{eq:A2S}) and~(\ref{eq:fA2S}). It is also worth mentioning that the identities of Lemma~\ref{lemma5} and Lemma~\ref{lemma6} admit closed-form representations for some special cases,  e.g., when the parameter $c$ is  fixed to non-negative integers, since all the gamma functions involved are reduced to rational functions. These identities as well as their derivatives of $c$ are the key tools in simplifying the summations in~(\ref{eq:B2S}) and~(\ref{eq:fB2S}).

Note also that in lemmas \ref{lemma1}--\ref{lemma4}, by rewriting the summations involved as hypergeometric functions, the corresponding identities induce the following new transformation formulas of unit arguments~(\ref{eq:tf1})--(\ref{eq:tf4}), which may be of independent interest.
\begin{align}
&\, _3F_2(c+1,1-m,1-b-m;a+1,c+2;1)\nonumber\\
&=\frac{(c+1) \Gamma (a+1) \Gamma (a+b+2 m-1)}{(c+m) \Gamma (a+m) \Gamma (a+b+m)}\nonumber\\
&~~~\times\, _3F_2(1,1-m,1-a-m;2-a-b-2 m,1-c-m;1)\label{eq:tf1}\\
\nonumber\\
&\, _3F_2(1,1-m,1-b-m;a+1,c+1;1)\nonumber\\
&=\frac{c }{c+m-1}\, _3F_2(1,1-m,a+b+m;a+1,2-c-m;1)\\
\nonumber\\
&\, _4F_3(1,1,1-m,1-b-m;2,a+1,c+1;1)\nonumber\\
&=\frac{a (b+c+m)}{(a+m-1) (b+m)}\, _4F_3(1,1,1-m,b+c+m+1;2,c+1,2-a-m;1)\nonumber\\
&~~~+\frac{a c (\psi_0(a)-\psi_0(a+m))}{m (b+m)}\\
\nonumber\\
&\frac{\, _3F_2(1,1-b-m,1-d-m;a+1,c+1;1)}{\Gamma (a+1) \Gamma (c+1) \Gamma (b+m) \Gamma (d+m)}\nonumber\\
&=\frac{\, _3F_2(1,1-b,1-d;a+m+1,c+m+1;1)}{\Gamma (b) \Gamma (d) \Gamma (a+m+1) \Gamma (c+m+1)}-\frac{1}{a+b+m-1}\nonumber\\
&~~~\times\left(\frac{\, _3F_2(1,1-a,c+d+m;2-a-b-m,d+m+1;1)}{\Gamma (a) \Gamma (c) \Gamma (b+m) \Gamma (d+m+1)}\right.\nonumber\\
&~~~+\left.\!\frac{\, _3F_2(1,1-b,c+d+m;2-a-b-m,c+m+1;1)}{\Gamma (b) \Gamma (d) \Gamma (a+m) \Gamma (c+m+1)}
\right)\nonumber\\
&~~~+\frac{\Gamma (c+d) \Gamma (a+b+2 m-1)}{\Gamma (c) \Gamma (d) \Gamma (a+m) \Gamma (b+m) \Gamma (a+b+m) \Gamma (c+d+m)}\nonumber\\
&~~~\times\left(\frac{1}{d}\, _3F_2(1,c+d,1-a-m;d+1,2-a-b-2 m;1)\right.\nonumber\\
&~~~+\!\left.\frac{1}{c}\, _3F_2(1,c+d,1-b-m;c+1,2-a-b-2 m;1)\right)\label{eq:tf4}.
\end{align}

\subsection{Simplification of summations}\label{subsec3.3}
In this section, we compute the integrals $\mathrm{I_A}$ and $\mathrm{I_B}$ in obtaining the variance formulas in Proposition~\ref{propav} and Proposition~\ref{propfv}. For convenience, we summarize in Table~\ref{tab_in} below the corresponding integrals of the variance~(\ref{def:vs}) for case A and case B.

\begin{table}[htbp]
\begin{center}
\begin{minipage}{240pt}
\caption{Integrals of the variance~(\ref{def:vs}) for case A and case B.}\label{tab_in}
\begin{tabular}{@{}lcc@{}}
\toprule
&~~~~$\mathrm{I_A}$~~~~&$~~~~\mathrm{I_B}$~~~~ \\
\midrule
Case A: Arbitrary number of particle~(\ref{eq:fgap}) &(\ref{eq:IAex})& (\ref{eq:IBex}) \\
\\
 Case B: Fixed number of particle~(\ref{eq:fgfp})    &(\ref{eq:IAexf})&(\ref{eq:IBexf})\\
\botrule
\end{tabular}
\end{minipage}
\end{center}
\end{table}
\noindent

\subsubsection*{Case A: Arbitrary number of particles}\label{subsubsec4}
In case A, the variance calculation boils down to simplifying the summation representations of $\mathrm{A_1}$, $\mathrm{A_2}$, $\mathrm{B_1}$, and $\mathrm{B_2}$ in~(\ref{eq:A1S})--(\ref{eq:B2S}) as listed in Appendix \ref{subsecA1}.
The corresponding simplification procedures are shown below.

We first discuss the simplifications of summation representations~(\ref{eq:A1S}) and~(\ref{eq:A2S}) in computing the integral $\mathrm{I_A}$ in (\ref{eq:IAex}). The summations in~(\ref{eq:A1S})  can be simplified by appropriately changing the summation orders and using the first type identities~(\ref{eq:B1})--(\ref{eq:B12c1}) of the existing framework, where the closed-form identities are~(\ref{eq:B1})--(\ref{eq:B5}) and the semi closed-form ones are~(\ref{eq:B6})--(\ref{eq:B12c1}). Semi closed-form identities that we referred to represent the relation between two single summations. These sums are what we refer to as unsimplifiable basis, which will cancel completely in obtaining the closed-form results. On the other hand, whether the considered unsimplifiable basis can be computed into closed-form formulas is not of primary importance.

The simplification of the summations in (\ref{eq:A1S}) is as follows. In~(\ref{eq:A1S}), the first double summation is
\begin{align}
&\sum _{k=0}^{m-1}\sum _{j=2 k-2}^{2 k}\frac{2(-1)^j(2 a+4 k+1) (j+1)_2 (a+j+1)_2}{\Gamma (2k-j+1) \Gamma (j-2 k+3) (2 a+j+2 k+1)_3}\bigg(\!(\psi_0 (j+3)\nonumber\\
&-\psi_0 (2 a+j+2 k+4)-\psi_0 (j-2 k+3)+\psi_0 (a+j+3))^2\nonumber\\
&-\psi _1(2 a+j+2 k+4)+\psi _1(a+j+3)-\psi _1(j-2 k+3)+\psi _1(j+3)\!\bigg),
\end{align}
which is directly reduced to a single sum after evaluating the sum of $j$. By using the identities~(\ref{eq:B1})--(\ref{eq:B12c1}) along with the results~\cite{AS72}
\begin{align}
\psi_{0}(mk)&=\ln m+\frac{1}{m}\sum_{i=0}^{m-1}\psi_{0}\!\left(k+\frac{i}{m}\right)\label{eq:m_poly0}\\
\psi_{1}(mk)&=\frac{1}{m^2}\sum_{i=0}^{m-1}\psi_{1}\!\left(k+\frac{i}{m}\right), \qquad m \in\mathbb{Z^{+}}, \label{eq:m_poly1}
\end{align}
the remaining single sum is computed to the unsimplifiable basis of the form
\begin{equation}\label{eq:s1basis}
\sum _{k=1}^{m}\frac{\psi_0(k+c)}{k+d}, \qquad c \neq d.
\end{equation}
The other double summation in~(\ref{eq:A1S}) is
\begin{align}\label{eq:A1S2}
\sum _{k=0}^{m-1}\sum _{j=0}^{2 k-3} \frac{4(2 a+4 k+1) (j+1)_2 (a+j+1)_2}{(2k-j-2)_3 (2 a+j+2 k+1)_3}\Psi(j),
\end{align}
where
\begin{equation}\label{eq:A1Ss}
\Psi(j)=\psi_0 (2 a+j+2 k+4)-\psi_0 (a+j+3)+\psi_0 (2k-j-2)-\psi_0 (j+3).
\end{equation}
To process the summation~(\ref{eq:A1S2}), we first perform the partial fraction decomposition of the term
\begin{align}\label{eq:A1S2pfd}
&\frac{4(2 a+4 k+1) (j+1)_2 (a+j+1)_2}{(2k-j-2)_3 (2 a+j+2 k+1)_3}\nonumber\\
&=-\frac{2 (a+k) (2 a+2 k+1)}{2 a+j+2 k+2}+\frac{2 (a+k) (a+2 k-1) (2 a+2 k-1)}{(2 a+4 k-1) (2 a+j+2 k+1)}\nonumber\\
&~~~\!~+\frac{2 (a+k+1) (a+2 k+2) (2 a+2 k+1)}{(2 a+4 k+3) (2 a+j+2 k+3)}-\frac{2 (k+1) (2 k+1) (a+2 k+2)}{(2 a+4 k+3) (j-2 k)}\nonumber\\
&~~~\!~-\frac{2 k (2 k-1) (a+2 k-1)}{(2 a+4 k-1) (j-2 k+2)}+\frac{2 k (2 k+1)}{j-2 k+1}.
\end{align}
The sum~(\ref{eq:A1S2}) now boils down to
\begin{align}
&\sum _{k=0}^{m-1}\mathrm{p}_{c}(k)\sum _{j=0}^{2 k-3}\frac{1}{j+2k+2a+c}\Psi(j)\label{eq:A1S3}\\
&\sum _{k=0}^{m-1}\mathrm{p}^{\prime}_{c}(k)\sum _{j=0}^{2 k-3}\frac{1}{j-2k-1+c}\Psi(j)\label{eq:A1S4},
\end{align}
where the parameter $c$ takes the values $c=1,2,3$, $\mathrm{p}_{c}(k)$ and $\mathrm{p}^{\prime}_{c}(k)$ denote the polynomials in $k$ from the partial fraction decomposition~(\ref{eq:A1S2pfd}).  When considering $c=2$ and the term $\psi_0(j+3)$ in~(\ref{eq:A1Ss}), the double summation
\begin{equation}\label{eq:A1S5}
-\sum _{k=0}^{m-1}2 (a+k) (2 a+2 k+1)\sum _{j=0}^{2 k-3}\frac{\psi_0(j+3)}{j+2k+2a+2}
\end{equation}
is simplified to a single sum as
\begin{align}\label{eq:A1S6}
&\sum _{j=0}^{m-3} \Bigg(\psi _0(2 j+3) ((j+1) (2 j+1) (\psi_0(1)-\psi_0(a+j+m+1))+(j-m+2)\nonumber\\
& \times(2 a-j+m)+(j+1) (2 j+1) (\psi_0(a+2 j+3)-\psi_0(1)))+\psi _0(2 j+4) \nonumber\\
&\times\left((j+1) (2 j+3) \left(\psi_0(1)-\psi_0\!\left(a+j+m+\frac{3}{2}\right)\right)+(j+1) (2 j+3)\right.\nonumber\\
&\times\!\left. \left(\psi_0\!\left(a+2 j+\frac{7}{2}\right)-\psi_0(1)\right)+(2 a-j+m-1)(j-m+2) \right)\!\Bigg).
\end{align}
The above result is obtained as follows. In~(\ref{eq:A1S5}), it is noticed that directly evaluating the inner summation does not permit further simplifications. Instead, we first have to separate the inner summation that depends on an even index $2k$ into two as
\begin{equation}\label{eq:simex11}
-2\sum _{k=0}^{m-1} (a+k) (2 a+2 k+1)\left(\sum _{j=0}^{k-2} \frac{\psi_0(2 j+3)}{2 a+2 k+2 j+2}+\sum _{j=0}^{k-2} \frac{\psi_0(2 j+4)}{2 a+2 k+2 j+3}\right),
\end{equation}
which allows us to change the summation order as
\begin{align}\label{eq:simex12}
&-2\sum _{j=0}^{m-3} \psi_0(2 j+3)\sum _{k=j+2}^{m-1} \frac{ (a+k) (2 a+2 k+1)}{2 a+2 k+2 j+2}\nonumber\\
&-2\sum _{j=0}^{m-3} \psi_0(2 j+4)\sum _{k=j+2}^{m-1} \frac{ (a+k) (2 a+2 k+1) }{2 a+2 k+2 j+3}.
\end{align}
Evaluating the inner summations over $k$ in~(\ref{eq:simex12}) directly gives the result~(\ref{eq:A1S6}). Other double summations in~(\ref{eq:A1S3}) and~(\ref{eq:A1S4}) are similarly simplified into single sums. The resulting single sums, cf.~(\ref{eq:A1S6}), are further manipulated into the unsimplifiable basis of the form~(\ref{eq:s1basis}) by using the first type identities~(\ref{eq:B1})--(\ref{eq:B12c1}) along with the result~(\ref{eq:m_poly0}).

Different from the summations in~(\ref{eq:A1S}), the simplification of the summations in~(\ref{eq:A2S}) will need new simplification techniques before the existing framework is applicable. The summations in~(\ref{eq:A2S}) include three double sums and one triple sum. We first simplify the inner summations over $j$ of the three double sums by using the identities~(\ref{eq:Bn1}) and~(\ref{eq:Bn5}) of the new simplification framework.
Specifically, the inner summation
\begin{equation}\label{eq:s_2}
\sum _{j=0}^{2 k} \frac{2 (j+1) (2k-j+1) \psi_0 (j+2) \psi_0 (2k-j+2)}{\Gamma (j+1) \Gamma (a+j+1) \Gamma (2k-j+1) \Gamma (a-j+2 k+1)}
\end{equation}
of the first double sum in~(\ref{eq:A2S}) is simplified to a semi-closed form representation that can be compactly written as
\begin{align}\label{eq:s2_result}
&\frac{1}{\Gamma (2 k) \Gamma (a+2 k+1) \Gamma (2 a+2 k+1)}\Bigg(\!\left(\frac{1-4 a^2}{4 a+8 k-2}+a-8 k^2-14 k-\frac{15}{2}\right)\nonumber\\
&\times \sum _{j=1}^{2 k} \frac{\Gamma (2 a-j+4 k)}{ \Gamma (a-j+2 k)j^2}+\left(\left(\frac{1-4 a^2}{4 a+8 k-2}+a-8 k^2-14 k-\frac{15}{2}\right) \psi _0(2 k)\right.\nonumber\\
&+\!\left.\frac{1-4 a^2}{2 (2 a+4 k-1)^2}+\frac{3-4 a^2}{4 a+8 k-2}+\frac{4 (a-1)}{a+2 k}+a-14 k-\frac{4}{k}-19\right) \nonumber\\
&\times\sum _{j=1}^{2 k} \frac{\Gamma (2 a-j+4 k)}{\Gamma (a-j+2 k)j }\!\Bigg)+\mathrm{CF},
\end{align}
where the shorthand notation $\mathrm{CF}$ denotes the closed-form terms omitted. The result~(\ref{eq:s2_result}) is obtained by using the identities~(\ref{eq:ofgi}), ~(\ref{eq:Bn1}), and~(\ref{eq:Bn5}) after one rewrites the summation~(\ref{eq:s_2}) as
\begin{align}
&8 \sum _{j=2}^{2 k+1} \frac{1}{\Gamma (j) \Gamma (a+j) \Gamma (2 k-j+2) \Gamma (a-j+2 k+2)}\nonumber\\
&+8 \sum _{j=1}^{2 k+1} \frac{\psi _0(j)}{\Gamma (j) \Gamma (a+j) \Gamma (2 k-j+2) \Gamma (a-j+2 k+2)}\nonumber\\
&+8 \sum _{j=1}^{2 k} \frac{\psi _0(j)}{\Gamma (j) \Gamma (a+j+1) \Gamma (2 k-j+1) \Gamma (a-j+2 k+1)}\nonumber\\
&+2 \sum _{j=1}^{2 k+1} \frac{\psi _0(j) \psi _0(2 k-j+2)}{\Gamma (j) \Gamma (a+j) \Gamma (2 k-j+2) \Gamma (a-j+2 k+2)}
\nonumber\\
&+4 \sum _{j=1}^{2 k} \frac{\psi _0(j) \psi _0(2 k-j+1)}{\Gamma (j) \Gamma (a+j) \Gamma (2 k-j+1) \Gamma (a-j+2 k+2)}
\nonumber\\
&+2 \sum _{j=1}^{2 k-1} \frac{\psi _0(j) \psi _0(2 k-j)}{\Gamma (j) \Gamma (a+j+1) \Gamma (2 k-j) \Gamma (a-j+2 k+1)}\nonumber\\
&-\frac{4 }{\Gamma (a+1) \Gamma (a+2 k+1)}\left(\frac{\psi_0 (2 k+1)}{\Gamma(2 k+1)}+\frac{\psi_0 (2 k)}{\Gamma (2 k)}\right).
\end{align}
Using the same approach, the other inner summations over $j$ of the first three double summations in~(\ref{eq:A2S}) are simplified to a semi closed-form representation, cf.~(\ref{eq:s2_result}), where the resulting unsimplifiable basis are
\begin{align}
&\sum _{j=1}^{2 k} \frac{\Gamma (2 a-j+4 k)}{ \Gamma (a-j+2 k)j}\label{basis1}\\
&\sum _{j=1}^{2 k} \frac{\Gamma (2 a-j+4 k)}{ \Gamma (a-j+2 k)j^2}\label{basis2}.
\end{align}

Now we simplify the inner summations over the indexes $i$ and $j$ of the triple sum in~(\ref{eq:A2S}). After the partial fraction decomposition
\begin{equation}\label{eq:A2Stippd}
\frac{1}{(j)_3}=-\frac{1}{j+1}+\frac{1}{2 (j+2)}+\frac{1}{2 j},
\end{equation}
the corresponding inner summation is written as
\begin{align}\label{eq:A2Sti}
&\sum _{j=1}^{2 k-1} \sum _{i=1}^{2 k-j} \frac{4 i (2 k-i+2) \Gamma (a+2 k+j+3)\Gamma (a+2 k-j+1)}{ \Gamma (j+i+1) \Gamma (2 k-j-i+1)\Gamma (a+i) \Gamma (a+2 k-i+2)}\nonumber\\
&\times(\psi _0(a+2 k+j+3)-\psi _0(2 a+4 k+4)+\psi _0(2 k-i+3)-\psi _0(j+3))\nonumber\\
&\times\left(-\frac{1}{j+1}+\frac{1}{2 (j+2)}+\frac{1}{2 j}\right).
\end{align}
We first consider the simplification of the summation
\begin{align}~\label{eq:A2Se}
&\sum _{j=1}^{2 k-1} \sum _{i=1}^{2 k-j} \frac{4 i (2 k-i+2) \Gamma (a+2 k+j+3)\Gamma (a+2 k-j+1)}{(j+1) \Gamma (j+i+1) \Gamma (2 k-j-i+1)\Gamma (a+i) \Gamma (a+2 k-i+2)}\nonumber\\
&\times(\psi _0(a+2 k+j+3)-\psi _0(2 a+4 k+4)+\psi _0(2 k-i+3)-\psi _0(j+3)),
\end{align}
which involves the term $1/(j+1)$ from the decomposition~(\ref{eq:A2Stippd}), and the remaining parts in~(\ref{eq:A2Sti}) can be simplified in the same manner. For convenience, we divide~(\ref{eq:A2Se}) into the three sums
\begin{align}
&\sum _{j=1}^{2 k-1} \sum _{i=1}^{2 k-j} \frac{4 i (2 k-i+2) \Gamma (a+2 k+j+3)\Gamma (a+2 k-j+1)}{(j+1) \Gamma (j+i+1) \Gamma (2 k-j-i+1)\Gamma (a+i) \Gamma (a+2 k-i+2)}\nonumber\\
&\times(-\psi _0(2 a+4 k+4)-\psi _0(j+3))\label{eq:A2Se1}\\
\nonumber\\
&\sum _{j=1}^{2 k-1} \sum _{i=1}^{2 k-j} \frac{4 i (2 k-i+2) \Gamma (a+2 k+j+3)\Gamma (a+2 k-j+1)}{(j+1) \Gamma (j+i+1) \Gamma (2 k-j-i+1)\Gamma (a+i) \Gamma (a+2 k-i+2)}\nonumber\\
&\times\psi _0(2 k-i+3)\label{eq:A2Se2}\\
\nonumber\\
&\sum _{j=1}^{2 k-1} \sum _{i=1}^{2 k-j} \frac{4 i (2 k-i+2) \Gamma (a+2 k+j+3)\Gamma (a+2 k-j+1)}{(j+1) \Gamma (j+i+1) \Gamma (2 k-j-i+1)\Gamma (a+i) \Gamma (a+2 k-i+2)}\nonumber\\
&\times\psi _0(a+2 k+j+3)\label{eq:A2Se3}.
\end{align}
The digamma functions in the summation~(\ref{eq:A2Se1}) is independent to the index $i$. This fact allows us to evaluate the summation over $i$ by using the identity~(\ref{eq:lemma2}) of Lemma \ref{lemma2}. Specifically, we first rewrite~(\ref{eq:A2Se1}) as
\begin{align}
&\sum _{j=1}^{2 k-1} \frac{\Gamma (a-j+2 k+1) \Gamma (a+j+2 k+3) (-\psi_0(2 a+4 k+4)-\psi_0(j+3))}{j+1}\nonumber\\
&\times\sum _{i=1}^{2 k-j}\frac{1}{\Gamma (a+i) \Gamma (i+j) \Gamma (a-i+2 k+2) \Gamma (2k-i-j)} \frac{4 i (2k-i+2)}{(i+j) (2k-j-i) }.\label{eq:A2Se12}
\end{align}
We now partial fraction decompose the rational function with respect to $i$ in~(\ref{eq:A2Se12}). The summation~(\ref{eq:A2Se1}) becomes
\begin{align}
&\sum _{j=1}^{2 k-1} \frac{\Gamma (a-j+2 k+1) \Gamma (a+j+2 k+3) (-\psi_0(2 a+4 k+4)-\psi_0(j+3))}{j+1}\nonumber\\
&\times\left(4 \sum _{i=1}^{2 k-j-1} \frac{1}{\Gamma (a+i) \Gamma (i+j) \Gamma (a-i+2 k+2) \Gamma (2k-j-i)}-\frac{j+2 k+2}{k}\right.\nonumber\\
&\times  2j\sum _{i=1}^{2k-j-1} \frac{1}{\Gamma (a+i) \Gamma (i+j+1) \Gamma (a-i+2 k+2) \Gamma (2k-j-i)}+\frac{2 k-j}{k}\nonumber\\
& \times\!\left.2(j+2) \sum _{i=1}^{2 k-j} \frac{1}{\Gamma (a+i) \Gamma (i+j) \Gamma (a-i+2 k+2) \Gamma (2k-j+1-i)}\right),\label{eq:A2Se13}
\end{align}
where the inner summations over $i$ can be evaluated by the identity~(\ref{eq:lemma2}) with the specializations
\begin{align}
&a= j,\qquad ~~~~~~\!b= a+j+2,\qquad c= a,\qquad m= 2 k-j-1\\
&a= j+1,\qquad b= a+j+2,\qquad c= a,\qquad m= 2 k-j-1\\
&a= j,\qquad ~~~~~~\!b= a+j+1,\qquad c= a,\qquad m= 2 k-j.
\end{align}
The summation~(\ref{eq:A2Se13}) becomes
\begin{align}
&\sum _{j=1}^{2 k-1} \left(j+1-\frac{(a+2 k+1)^2}{j+1}\right)\frac{\psi _0(2 a+4 k+4)+\psi _0(j+3)}{\Gamma (a) \Gamma (a+2 k+1)}\nonumber\\
&\times\left(4 (a-j+2 k-1) (a+j+2 k+1)\sum _{i=1}^{2 k-j-1} \frac{\Gamma (a+i-1) \Gamma (a-i+4 k)}{\Gamma (i) \Gamma (2 k-i)}\right.\nonumber\\
&-\frac{2 j (j+2 k+2) (a-j+2 k-1)}{k}\sum _{i=1}^{2k-j-1} \frac{\Gamma (a+i-1) \Gamma (a-i+4 k+1)}{\Gamma (i) \Gamma (2k-i+1)}\nonumber\\
&+\!\left.\frac{2 (j+2) (2 k-j) (a+j+2 k+1)}{k}\sum _{i=1}^{2 k-j} \frac{\Gamma (a+i-1) \Gamma (a-i+4 k+1)}{\Gamma (i) \Gamma (2k-i+1)}\right).\label{eq:A2Se1l}
\end{align}
In~(\ref{eq:A2Se1l}), we further change the orders of summations to simplify the sums over $j$. These summations admit closed-form expressions by using the identities~(\ref{eq:B3})--(\ref{eq:B32}) and (\ref{eq:B4}). The remaining summations only consist of single sums, which can be further simplified by using the identities~(\ref{eq:Bn6})--(\ref{eq:Bn7}) into the unsimplifiable basis~(\ref{basis1}) and
\begin{align}
&\sum _{j=1}^{2 k} \frac{\Gamma (2 a-j+4 k-1)\psi _0(j) }{\Gamma (a-j+2 k+1)}\label{basis3}\\
&\sum _{j=1}^{2 k} \frac{ \Gamma (2 a-j+4 k-1)\psi _0(j)}{\Gamma (a-j+2 k)j }\label{basis4}\\
&\sum _{j=1}^{2 k} \frac{\Gamma (2 a-j+4 k) \psi _0(2 a-j+4 k)}{\Gamma (a-i+2 k)j }\label{basis5}.
\end{align}

We now move on to simplifying the summation~(\ref{eq:A2Se2}). In~(\ref{eq:A2Se2}), we change the summation order to evaluate the sum over $j$ by the identity~(\ref{eq:lemma3}) in Lemma~\ref{lemma3}. The summation then becomes
\begin{align}
&\Gamma^2 (a+2 k+2)\sum _{i=1}^{2 k-1} \sum _{j=1}^{2 k-i}\frac{4 (i+j+1) (a+i+j+1) (2k-i-j+1)}{j \Gamma (i+1) \Gamma (a+i+2) \Gamma (2k-i+1) \Gamma (a-i+2 k+2)}\nonumber\\
&\times (a-i-j+2 k+1) \psi_0 (i+j+2)-\Gamma^2 (a+2 k+2)\sum _{i=1}^{2 k-1}4 i (2k-i+2)\nonumber\\
&\times\frac{\psi _0(2k-i+3) \left(\psi _0(a+i+1)-\psi _0(a+2 k+2)\right)}{\Gamma (i) \Gamma (a+i) \Gamma (2k-i+2) \Gamma (a-i+2 k+2)}-\Gamma (a+2 k+1) \nonumber\\
&\times\Gamma (a+2 k+3)\sum _{i=1}^{2 k-1} \frac{4 (2k-i+2) (2k-i+1) \psi_0 (2k-i+3)}{\Gamma (i) \Gamma (a+i) \Gamma (2k-i+2) \Gamma (a-i+2 k+2)}+\frac{1}{\Gamma (a+2)}\nonumber\\
&\times\frac{\Gamma (a+2 k+2)}{\Gamma (2 k+1) } \sum _{i=1}^{2 k-1} \frac{4 i (a+i) (2k-i+2) (a-i+2 k+2) \psi _0(2k-i+3)}{2k-i+1}.\label{eq:A2Se21}
\end{align}
We now shift the index $i\to 2k-i$ of the three single sums in the above result~(\ref{eq:A2Se21}). The first two single sums are simplified into the unsimplifiable basis~(\ref{basis1}) by using the identities~(\ref{eq:Bn1})--(\ref{eq:Bn2}). The last single sum in~~(\ref{eq:A2Se21}) is simplified directly into a closed-form representation by using the identities~(\ref{eq:B3})--(\ref{eq:B32}) and (\ref{eq:B4}). For the double summation in~(\ref{eq:A2Se21}), we take the partial fraction decomposition
\begin{align}
&\frac{4 (i+j+1) (a+i+j+1) (2k-i-j+1)(a-i-j+2 k+1) }{j}\nonumber\\
&=\frac{4 (i+1) (a+i+1) (2k-i+1) (a-i+2 k+1)}{j}+4 j^3+16 (i-k) j^2\nonumber\\
&~~~\!~-4 \left(a^2+2 a (k+1)-6 i^2+12 i k-4 k^2+4 k+2\right)j \nonumber\\
&~~~\!~+8 (i-k) \left(-a^2-2 a (k+1)+2 (i+1) (i-2 k-1)\right).\label{eq:A2Se21p}
\end{align}
In~(\ref{eq:A2Se21p}), the polynomial part can be simplified similarly approach as in~(\ref{eq:A2Se1l}), while the rational part is simplified as follows. The corresponding summation is
\begin{align}
\sum _{i=1}^{2 k-1} \frac{4 (i+1) (2k-i+1) \Gamma^2 (a+2 k+2)}{\Gamma (i+1) \Gamma (a+i+1) \Gamma (2k-i+1) \Gamma (a-i+2 k+1)}\sum _{j=1}^{2 k-i} \frac{\psi_0 (i+j+2)}{j}.\label{eq:A2Se22}
\end{align}
To evaluate~(\ref{eq:A2Se22}), the sum over $j$ is computed by using the identity~(\ref{eq:B11ic}) with the specialization
\begin{equation}
i= j,\qquad a= i+1,\qquad b= 1,\qquad m= 2 k-i,
\end{equation}
then we shift the summation index $i\to2k-i$ of the outer sum. Consequently,~(\ref{eq:A2Se22}) is simplified to
\begin{align}
&-\sum _{i=1}^{2 k-1} \frac{4 (i+1) (2k-i+1) \Gamma^2 (a+2 k+2)}{\Gamma (i+1) \Gamma (a+i+1) \Gamma (2k-i+1) \Gamma (a-i+2 k+1)}\nonumber\\
&\times\Bigg(\sum _{j=1}^{2k-i+1} \frac{\psi_0(i+j+1)}{j}-\frac{1}{2}(\left(\psi _0(2k-i+2)+\psi _0(i+1)-2 \psi _0(1)\right)\nonumber\\
&\times\left(\psi _0(2k-i+2)+\psi _0(i+1)\right)\left.-\psi _1(2k-i+2)-\psi _1(i+1)+2 \psi _1(1)\right)\!\Bigg).\label{eq:A2Se23}
\end{align}
In~(\ref{eq:A2Se23}), the double sum is the same form as in~(\ref{eq:A2Se22}) but with a negative sign. Therefore, by adding up~(\ref{eq:A2Se22}) and (\ref{eq:A2Se23}), and dividing the result by two, we reduce the double summation in~(\ref{eq:A2Se22}) to a single sum as
\begin{align}
&\sum _{i=1}^{2 k-1} \frac{4 (i+1) (2k-i+1) \Gamma^2 (a+2 k+2)}{\Gamma (i+1) \Gamma (a+i+1) \Gamma (2k-i+1) \Gamma (a-i+2 k+1)}\sum _{j=1}^{2 k-i} \frac{\psi_0 (i+j+2)}{j}\nonumber\\
&=\sum _{i=1}^{2 k-1} \frac{2 \Gamma^2 (a+2 k+2)}{\Gamma (i+1) \Gamma (a+i+1) \Gamma (2k-i+1) \Gamma (a-i+2 k+1)}\Big(\frac{1}{2}(2k-i+1)\nonumber\\
&~~~\!~\times  (i+1)((\psi _0(2k-i+2)+\psi _0(i+1))(\psi _0(2k-i+2)+\psi _0(i+1)\nonumber\\
&~~~\!~-2 \psi _0(1))-\psi _1(2k-i+2)-\psi _1(i+1)+2 \psi _1(1))+(2k-i+1) \nonumber\\
&~~~\!~\times\left(\psi _0(2k-i+2)+\psi _0(i+2)-\psi _0(2 k+3)-\psi _0(1)\right)-(i+1) \psi _0(2 k+3)\Big),
\end{align}
which is further simplified into a semi closed-form expression involving the unsimplifiable basis~(\ref{basis1})--(\ref{basis2})
by using the identities~(\ref{eq:Bn1}), and~(\ref{eq:Bn4}).

We now simplify the summation (\ref{eq:A2Se3}) as a last piece in~(\ref{eq:A2S}). We first change the summation order in~(\ref{eq:A2Se3}) to evaluate the sum over $j$ by the identity~(\ref{eq:Bn8}). As a result, (\ref{eq:A2Se3}) is simplified to
\begin{align}
&\Gamma^2 (a+2 k+2) \sum _{i=1}^{2 k} 4 i (a+i) (2k-i+2) (a-i+2 k+2)\nonumber\\
&\times\sum _{j=1}^{2k-i+1} \frac{\psi _0(a-i+2 k+3)-\psi _0(a-i-j+2 k+3)+\psi _0(a+2 k+2)}{j \Gamma (i+j) \Gamma (a+i+j+1) \Gamma (2k-i-j+2) \Gamma (a-i-j+2 k+3)}\nonumber\\
&-\Gamma^2 (a+2 k+2) \sum _{i=1}^{2 k-1} \frac{4 i (2k-i+2) \psi _0(a+2 k+2)}{\Gamma (i) \Gamma (a+i) \Gamma (2k-i+2) \Gamma (a-i+2 k+2)}\nonumber\\
&\times \left(\psi _0(a+i+1)-\psi _0(a+2 k+2)\right)-\Gamma (a+2 k+1) \Gamma (a+2 k+3)\nonumber\\
&\times \sum _{i=1}^{2 k-1} \frac{4 (-i+2 k+2) \psi _0(a+2 k+3)}{\Gamma (i) \Gamma (a+i) \Gamma (-i+2 k+1) \Gamma (a-i+2 k+2)}-16 (a+2) k (a+2 k)\nonumber\\
&\times\frac{\Gamma (a+2 k+2) }{\Gamma (a+2) \Gamma (2 k+1)}\left(\psi _0(a+2 k+2)-\psi _0(a+2)+\psi _0(a+3)\right),\label{eq:A2Se31}
\end{align}
where the single summations are simplified similarly to the ones in~(\ref{eq:A2Se21}). To simplify the double summation in~(\ref{eq:A2Se31}), we shift the summation index $i\to 2k-i-j+2$ as
\begin{align}
&\sum _{i=1}^{2 k} \frac{\Gamma^2 (a+2 k+2)}{\Gamma (i) \Gamma (a+i+1) \Gamma (2k-i+2) \Gamma (a-i+2 k+3)}\nonumber\\
&\times\sum _{j=1}^{2k-i+1} \frac{4 (i+j) (a+i+j) (2k-i-j+2) (a-i-j+2 k+2)}{j}\nonumber\\
&\times\left(\psi _0(a+i+j+1)-\psi _0(a+i+1)+\psi _0(a+2 k+2)\right).\label{eq:A2Se32}
\end{align}
In~(\ref{eq:A2Se32}), the inner summations over $j$ are simplified into a closed-form representation except for the sum
\begin{equation}
\sum _{j=1}^{2k-i+1}\frac{\psi _0(a+i+j+1)}{j}.\label{eq:A2Se32u}
\end{equation}
Therefore, the resulting summations consist of several single sums as well as a double sum. These single sums are simplified into closed-form expressions by using the identity~(\ref{eq:ofgi}) and its derivative of $a$.
The double summation involves the sum in~(\ref{eq:A2Se32u}), and is written as
\begin{align}
\sum _{i=1}^{2 k} \frac{4 i (2k-i+2) \Gamma^2 (a+2 k+2)}{\Gamma (i) \Gamma (a+i) \Gamma (2k-i+2) \Gamma (a-i+2 k+2)}\sum _{j=1}^{2k-i+1}\frac{\psi _0(a+i+j+1)}{j}.\label{eq:A2Se33}
\end{align}
To process~(\ref{eq:A2Se33}), we first evaluate the inner summation by the identity~(\ref{eq:B11ic}) with the specialization
\begin{equation}
i=j, \qquad a=a+1, \qquad  b=i,\qquad m=2k-i+1,
\end{equation}
and the sum~(\ref{eq:A2Se33}) becomes
\begin{align}
&\sum _{i=1}^{2 k} \frac{4 i (2k-i+2) \Gamma^2 (a+2 k+2)}{\Gamma (i) \Gamma (a+i) \Gamma (2k-i+2) \Gamma (a-i+2 k+2)}\Bigg(\sum _{j=1}^{2k-i+1} \frac{\psi _0(i+j)}{j}\nonumber\\
&-\sum _{j=1}^{a+1} \frac{\psi _0(j+2 k+1)}{i+j-1}+\frac{1}{2}(-\psi _1(a+i+1)+\psi _1(i)+\left(\psi _0(a+i+1)-\psi _0(i)\right)\nonumber\\
&\times \left(\psi _0(a+i+1)+2 \left(\psi _0(2k-i+2)-\psi_0(1)\right)+\psi _0(i)\right))\!\Bigg).\label{eq:A2Se34}
\end{align}
In comparison with the summations in~(\ref{eq:A2Se23}), the summations in~(\ref{eq:A2Se34}) can be simplified similarly except for a double sum
\begin{align}
\sum _{i=1}^{2 k} \frac{4 i (2k-i+2) \Gamma^2 (a+2 k+2)}{\Gamma (i) \Gamma (a+i) \Gamma (2k-i+2) \Gamma (a-i+2 k+2)}\sum _{j=1}^{a+1} \frac{\psi _0(j+2 k+1)}{i+j-1},\label{eq:A2Se35}
\end{align}
where it does not permit further simplifications by directly evaluating the inner sum. To proceed further~(\ref{eq:A2Se35}), we change the summation order and take the partial fraction decomposition
\begin{equation}
\frac{4 i (2k-i+2)}{i+j-1}=-\frac{4 (j-1) (j+2 k+1)}{i+j-1}-4 (i-1)+4 (j+2 k).
\end{equation}
The summation (\ref{eq:A2Se35}) is now written as
\begin{align}
&-\Gamma^2 (a+2 k+2)\sum _{j=1}^{a+1} 4 (j-1) (j+2 k+1) \psi _0(j+2 k+1)\nonumber\\
&\times\sum _{i=1}^{2 k} \frac{1}{\Gamma (i) \Gamma (a+i) \Gamma (2k-i+2) \Gamma (a-i+2 k+2)}\frac{1}{i+j-1}\nonumber\\
&-4\sum _{j=1}^{a+1} \psi _0(j+2 k+1)\sum _{i=2}^{2 k} \frac{ \Gamma^2 (a+2 k+2) }{\Gamma (i-1) \Gamma (a+i) \Gamma (2k-i+2) \Gamma (a-i+2 k+2)}\nonumber\\
&+4 \sum _{j=1}^{a+1} \psi _0(j+2 k+1) \sum _{i=1}^{2 k} \frac{(j+2 k) \Gamma^2 (a+2 k+2)}{\Gamma (i) \Gamma (a+i) \Gamma (2k-i+2) \Gamma (a-i+2 k+2)}.\label{eq:A2Se351}
\end{align}
In~(\ref{eq:A2Se351}), the second and third double summations are simplified directly into closed-form expressions by using the identities~(\ref{eq:ofgi}), and (\ref{eq:B1})--(\ref{eq:B3}). The first double summation in~(\ref{eq:A2Se351}), after applying the identity (\ref{eq:lemma1}) with the specialization
\begin{equation}
 b= a, \qquad c=j-1, \qquad m=2 k+1,
\end{equation}
 is simplified to
\begin{align}
&\frac{(a+2 k+1) \Gamma (a+2 k+2)}{\Gamma (a+1)  \Gamma (2 k+1)}\sum _{j=1}^{a+1}\frac{ 4 (j-1) (j+2 k+1)}{(j+2 k)} \psi _0(j+2 k+1)-\nonumber\\
&\frac{(a+2 k+1) \Gamma (a+2 k+2)}{\Gamma (2 a+2 k+1) }\sum _{j=1}^{a+1}\sum _{i=1}^{2 k+1} 4 (j-1) (j+2 k+1) \psi _0(j+2 k+1)\nonumber\\
&\times \frac{\Gamma (2 a-i+4 k+2) }{\Gamma (2 k-i+2) \Gamma (a-i+2 k+2)}\frac{\Gamma (2 k-i+j+1)}{\Gamma (j+2 k+1)}.\label{eq:A2Se352}
\end{align}
The single sum over $j$ in~(\ref{eq:A2Se352}) admits a closed-form representation by applying the identity~(\ref{eq:B1}), (\ref{eq:B3}), and~(\ref{eq:B4}). The double sum in~(\ref{eq:A2Se352}) is simplified by first using the identities~(\ref{eq:B201})--(\ref{eq:B202}) to compute the sum over $j$, before using the identities~(\ref{eq:B20})--(\ref{eq:B22}) to evaluate the sum over $i$. Consequently, we obtain a semi closed-form result of (\ref{eq:A2Se352}), and the corresponding unsimplifiable basis are~(\ref{basis1}) and
\begin{equation}
\sum _{i=1}^{2 k} \frac{\psi_0(2 a+i+2 k)}{i} \label{A2basis_inner}.
\end{equation}

We have so far completed the simplification of the inner summations over indexes $i$ and $j$ in~(\ref{eq:A2S}). Summing up these results, we observe the complete cancellation among the unsimplifiable basis (\ref{basis1})-(\ref{basis2}) and (\ref{basis3})-(\ref{basis5}). The only survived term is~(\ref{A2basis_inner}). In the resulting outer sum over index $k$ in~(\ref{eq:A2S}), all the gamma functions are reduced to rational ones. Therefore, (\ref{eq:A2S}) can now be simplified into a similar form as~(\ref{eq:A1S}) in terms of the unsimplifiable basis~(\ref{eq:s1basis}).

Inserting the resulting summations of~(\ref{eq:A1S}) and~(\ref{eq:A2S}) into~(\ref{eq:IAex}), we obtain
\begin{align}
\mathrm{I_A}=&\sum _{k=1}^{m}\left(\left(\frac{2 a-1}{2 k}-\frac{2 a+1}{2 (a+k)}+\frac{2 a+1}{2 k+1}-\frac{2a-1}{2 a+2 k+1}\right) \psi _0(2 a+4 k)\right.\nonumber\\
&+\left(\frac{2 (a+m) (2 a+3 m-1)}{2 a+4 m-1}\left(\frac{1}{a+2 k+1}+\frac{1}{a+2 k}\right)-\frac{2 a}{2 k+1}\right.\nonumber\\
&+\!\left.\frac{1-2 a}{2 k}+\frac{1}{2 (a+k)}\right) \psi _0(2 a+2 k)+\left(\frac{2 a m-2 a+6 m^2-6 m+1}{(2 k+1) (2 a+4 m-1)}\right.\nonumber\\
&+\!\left.\frac{4 a m+2 a+12 m^2-1}{4 k (2 a+4 m-1)}+\frac{1}{4 (a+k)}\right) \psi _0(a+2 k)-\left(\frac{2 a+2 m-1}{2 (2 k+1)}\right.\nonumber\\
&+\!\left.\frac{a+m}{2 k}\right) \psi _0(a+k+m)-\left(\frac{2 a+2 m-1}{4 k}+\frac{a+m}{2 k+1}\right) \nonumber\\
&\times\!\left.\psi _0\!\left(a+k+m+\frac{1}{2}\right)-\left(\frac{1}{2 (2 k+1)}+\frac{1}{4 k}\right) \psi _0(a+k)\right)+\mathrm{CF}.\label{I_Aas}
\end{align}
We remind the readers that the omitted closed-form terms are denoted by the abbreviation $\mathrm{CF}$, which in general is different in each use.
Now we discuss the simplification of the summations in~(\ref{eq:B1S}) and~(\ref{eq:B2S}) in computing the integral $\mathrm{I_B}$ in~(\ref{eq:IBex}). The single summation in~(\ref{eq:B1S}) are computed into the unsimplifiable basis of the form~(\ref{eq:s1basis}) by using the identities~(\ref{eq:B1})--(\ref{eq:B12c1}) along with the result~(\ref{eq:m_poly0}).

The simplification of the double summation in~(\ref{eq:B2S}) will utilize Lemma~\ref{lemma6}. By partial fraction decomposing the rational functions in $j$ and shifting the summation index $j\to m-k-j$, the simplification of~(\ref{eq:B2S}) boils down to computing the summations
\begin{align}
&\sum _{k=1}^{m-1}\mathrm{p}_{c,\lambda}(k) \frac{\Gamma (2m-2k+1)}{\Gamma (2 a-2 k+2 m+1)}\sum _{j=1}^{m-k}\frac{\Gamma (2 m+2 a-2 j-2 k+1)}{\Gamma (2 m-2 j-2 k+1)}\nonumber\\
&\times \left(\frac{1}{\left(a-j-2 k+2 m+\frac{1}{2}+c\right)^{\lambda}}-\frac{1}{(j+c)^{\lambda}}\right),\label{eq:b2spfd}
\end{align}
where the parameter $c$ and $\lambda$ take the values $c=-1/2, 0, 1/2$, $\lambda=1,2$, and $\mathrm{p}_{c,\lambda}(k)$ denotes the rational functions in $k$. These summations are simplified by using Lemma~\ref{lemma6} to evaluate the inner sums over $j$, which reduces the gamma ratio
\begin{equation}\label{eq:gratio}
\frac{\Gamma (2m-2k+1)}{\Gamma (2 a-2 k+2 m+1)}
\end{equation}
into a rational function. 
Specifically, in the case when $c=1/2$ and $\lambda=1$ in~(\ref{eq:b2spfd}), the corresponding inner summation is
\begin{equation}
\sum _{j=1}^{m-k}\frac{\Gamma (2 m+2 a-2 j-2 k+1)}{\Gamma (2 m-2 j-2 k+1)} \left(\frac{1}{1+a-j-2 k+2 m}-\frac{1}{j+\frac{1}{2}}\right).\label{eq:s_4}
\end{equation}
By using the relation~\cite{Brychkov08}
\begin{equation}
\Gamma (2 k)=\frac{2^{2 k-1} }{\sqrt{\pi }}\Gamma (k) \Gamma\! \left(k+\frac{1}{2}\right),\label{eq:gamma2}
\end{equation}
the summation~(\ref{eq:s_4}) is written as
\begin{equation}
2^{2a}\sum _{j=1}^{m-k}\Phi _{0,-\frac{1}{2},a,a-\frac{1}{2}}^{(m-k-j)}\left(\frac{1}{1+a-j-2 k+2 m}-\frac{1}{j+\frac{1}{2}}\right).\label{eq:s_4next}
\end{equation}
Here, we recall the notation~(\ref{eq:lemma56n})
\begin{equation}
\Phi_{a,b,c,d}^{(x)}=\frac{\Gamma (x+c+1) \Gamma (x+d+1) }{\Gamma (x+a+1) \Gamma (x+b+1)}.
\end{equation}
By using Lemma~\ref{lemma6} with the specialization
\begin{equation}
a=-\frac{1}{2},~~~b=a,~~~c=\frac{1}{2},~~~m\to m-k,
\end{equation}
(\ref{eq:s_4}) is simplified to
\begin{align}
&\sum _{j=1}^{m-k}\frac{\Gamma (2 m+2 a-2 j-2 k+1)}{\Gamma (2 m-2 j-2 k+1)} \left(\frac{1}{1+a-j-2 k+2 m}-\frac{1}{j+\frac{1}{2}}\right)\nonumber\\
&=\frac{\Gamma (2 a-2 k+2 m+2)}{\Gamma (2m-2k+2)}\Bigg(\psi _0(a+1)-\psi _0(a-2 k+2 m+1)\nonumber\\
&~~~\!~-\frac{2 (2 a+1)}{2 a-2 k+2 m+1}+2\Bigg).\label{eq:s_4r}
\end{align}

To simplify the summation~(\ref{eq:b2spfd}) for $c=1/2$ and $\lambda=2$, we again use the relation~(\ref{eq:gamma2}), and the corresponding inner summation
\begin{equation}\label{eq:s_50}
\sum _{j=1}^{m-k}\frac{\Gamma (2 m+2 a-2 j-2 k+1)}{\Gamma (2 m-2 j-2 k+1)} \left(\frac{1}{(1+a-j-2 k+2 m)^2}-\frac{1}{\left(j+\frac{1}{2}\right)^2}\right)
\end{equation}
 becomes
\begin{equation}
2^{2a}\sum _{j=1}^{m-k}\Phi _{0,-\frac{1}{2},a,a-\frac{1}{2}}^{(m-k-j)} \left(\frac{1}{(1+a-j-2 k+2 m)^2}-\frac{1}{(j+\frac{1}{2})^2}\right)\label{eq:s_5},
\end{equation}
which is evaluated by taking derivative of the parameter $c$ of the identity~(\ref{eq:lemma6}) with the specialization
\begin{equation}
a=-\frac{1}{2},~~~b=a,~~~m\to m-k,
\end{equation}
before setting $c=1/2$. As a result, (\ref{eq:s_50}) is simplified to
\begin{align}
&\sum _{j=1}^{m-k}\frac{\Gamma (2 m+2 a-2 j-2 k+1)}{\Gamma (2 m-2 j-2 k+1)} \left(\frac{1}{(1+a-j-2 k+2 m)^2}-\frac{1}{\left(j+\frac{1}{2}\right)^2}\right)\nonumber\\
&=\frac{\Gamma (2 m+2 a-2 k+2)}{\Gamma (2m-2k+2)}\sum _{j=1}^{m-k}\Bigg(\!\left(-\frac{4 a^2-1}{a+j-1}-\frac{1-4 a^2}{a+j}-\frac{4 a^2-6 a+2}{-2 a-2 j+3}\right.\nonumber\\
&~~~\!~-\!\left.\frac{2 a (2 a+1)}{2 a+2 j+1}+\frac{1}{a+2 j}-\frac{2-8 a}{2 a+2 j-1}-\frac{1}{-a-2 j+1}\right) \left(\psi _0\!\left(j+\frac{1}{2}\right)\right.\nonumber\\
&~~~\!~-\psi _0\!\left(a+j+\frac{1}{2}\right)+\psi _0(a-k+m+1)+\psi _0\!\left(a-k+m+\frac{3}{2}\right)+\psi _0(j)\nonumber\\
&~~~\!~-\!\left.\psi _0(a+j)-\psi _0(m-k+1)-\psi _0\!\left(m-k+\frac{3}{2}\right)\right)+\frac{a+j}{j (a+2 j)^2}\nonumber\\
&~~~\!~+\frac{1}{2 a+2 j-1}\left(\frac{a (2 a-1) (2 j-1)}{j (a+j)^2}-\frac{2 a (2 a+1)}{\left(a+j+\frac{1}{2}\right)^2}+\frac{2 j-1}{(a+2 j-1)^2}\right)\!\Bigg).
\label{eq:s_5r}
\end{align}

The other cases of $c$ and $\lambda$ combinations are obtained similarly. Inserting the results into~(\ref{eq:b2spfd}), the gamma ratio~(\ref{eq:gratio}) of the outer sum over $k$ is reduced to a rational function. The resulting summations are further computed into unsimplifiable basis of the form~(\ref{eq:s1basis}) by using the identities~(\ref{eq:B1})--(\ref{eq:B12c1}).

Inserting the resulting summations of~(\ref{eq:B1S}) and~(\ref{eq:B2S}) into~(\ref{eq:IBex}), we obtain
\begin{align}
\mathrm{I_B}=&\sum _{k=1}^{m}\left(\left(\frac{ 2 (a+m) (2 a+3 m-1)}{2 a+4 m-1}\left(\frac{1}{a+2 k+1}+\frac{1}{a+2 k}\right)+\frac{m}{k}+\frac{2 m-1}{2 k+1}\right)\right.\nonumber\\
&\times \psi_0(2 a+2 k)+\left(\frac{m}{2 a+2 k+1}+\frac{2 m-1}{4 (a+k)}-\frac{2 m-1}{2 (2 k+1)}-\frac{m}{2 k}\right)\times\nonumber\\
& \psi _0\!\left(a+2 k+\frac{1}{2}\right)+\left(\frac{m}{2 (a+k)}+\frac{m}{2 a+2 k+1}+\frac{ m (2 m-1)}{2 a+4 m-1}\right.\nonumber\\
&\times\left.\left.\left(\frac{1}{2 k+1}+\frac{1}{2 k}\right)\right) \psi _0(a+2 k)\right)+\mathrm{CF}.\label{I_Bas}
\end{align}

Now inserting the $\mathrm{I_A}$ expression (\ref{I_Aas}) and $\mathrm{I_B}$ expression~(\ref{I_Bas}) into~(\ref{def:vs}), we obtain
\begin{align}
\mathbb{V}\!\left[S\right]=&\frac{-2 a-2 m+1}{4} \Omega _1^{\left(a,a+\frac{1}{2}\right)}+\frac{-2 a-2 m+1}{4} \left(\Omega _2^{\left(a+\frac{1}{2},0\right)}+\Omega _2^{\left(a+\frac{1}{2},\frac{1}{2}\right)}\right)\nonumber\\
&-\frac{a+m}{2} \left(\Omega _2^{\left(a,0\right)}+\Omega _2^{\left(a,\frac{1}{2}\right)}\right)+\mathrm{CF},\label{eq:apvascr}
\end{align}
where
\begin{align}
\Omega _1^{\left(a,b\right)}=&\sum _{k=1}^m \left(\frac{\psi _0(a+k)}{b+k}+\frac{\psi _0(b+k)}{a+k}\right)\\
\Omega _2^{\left(a,b\right)}=&\sum _{k=1}^m \left(\frac{\psi _0(a+b+k+m)}{b+k}+\frac{\psi _0(a+k)}{b+k}+\frac{\psi _0(a+b+2 k)}{a+k}\right.\nonumber\\
&-\!\left.\frac{\psi _0(a+b+k)}{a+k}-\frac{\psi _0(a+b+2 k)}{b+k}\right)\label{eq:omega2}.
\end{align}
Simplifying the single sums in~(\ref{eq:apvascr}) by using the closed-form identities~(\ref{eq:B12c1}) and~(\ref{eq:B12c2}) directly leads to the desired result~(\ref{f:apv}). This completes the proof of Proposition~\ref{propav}.
\subsubsection*{Case B: Fixed number of particles}\label{subsubsec5}
In case B, the variance calculation boils down to simplifying the summation representations of $\mathcal{A}_1$, $\mathcal{A}_2$, $\mathcal{B}_1$, $\mathcal{B}_2$ as summarized  in~(\ref{eq:fA1S})--(\ref{eq:fB2S}) in Appendix \ref{subsecA2}.

Note that the simplification procedure in case A also works for the majority of the summations in case B. The only new summation in case B is the double summation
\begin{align}
&\sum _{k=1}^{m-2}4 (a+b-2 k+2 m+1) \frac{ \Gamma (m-k+1)}{\Gamma (a+b-k+m+1)}\sum _{j=1}^{m-k-1} \frac{\Gamma (a+b-j-k+m)}{\Gamma (m-k-j)}\nonumber\\
&\times(-1)^j(2 j+2 k-2 m-a-b+1) \left(\frac{1}{(j)_3 (a+b-j-2 k+2 m-1)_3}\right)^2\nonumber\\
&\times \left((1-a) j (a+b-j-2 k+2 m-1)-2 (a-k+m) (a+b-k+m)\right)\nonumber\\
&\times\left((1-b) j (a+b-j-2 k+2 m-1)-2 (b-k+m) (a+b-k+m)\right),\label{eq:fb2c}
\end{align}
which is obtained after opening the bracket of the double summation in~(\ref{eq:fB2S}) and shifting the index $k\to m-j-k$.
To process~(\ref{eq:fb2c}), we take the partial fraction decomposition of the rational functions (starting from the second to the last line) in $j$, the summation (\ref{eq:fb2c}) now boils down to
\begin{align}
&\sum _{k=1}^{m-2}\mathrm{p}_{c,\lambda}(k)\frac{\Gamma (m-k+1)}{\Gamma (a+b-k+m+1)}\sum _{j=1}^{m-k-1} \frac{(-1)^j \Gamma (m-k+a+b-j)}{\Gamma (m-k-j)}\nonumber\\
&\times\left(\frac{1}{(2m-2k-1+a+b-j+c)^\lambda}-\frac{1}{(j+c)^\lambda}\right),\label{eq:fb2c1}
\end{align}
where the parameters $c$ and $\lambda$ take the values $c=0,1,2$, $\lambda=1,2$, and $\mathrm{p}_{c,\lambda}(k)$ denotes the rational polynomials in $k$. The simplification of these summations will utilize the Lemma \ref{lemma5}. For the case when $c=0$ and $\lambda=1$ in~(\ref{eq:fb2c1}), the corresponding inner summation  can be simplified into a closed-form expression by directly using the result~(\ref{eq:lemma5c=0}), which is a special case of the identity~(\ref{eq:lemma5}) in Lemma~\ref{lemma5}. When $c=0$ and $\lambda=2$ in~(\ref{eq:fb2c1}), the inner summation
\begin{equation}
\sum _{j=1}^{m-k-1} \frac{(-1)^j  \Gamma (a+b-j-k+m)}{\Gamma (m-k-j)}\left(\frac{1}{(a+b-j-2 k+2 m-1)^2}-\frac{1}{j^2}\right)
\end{equation}
is simplified to
\begin{align}
&\frac{\Gamma (a+b-k+m)}{\Gamma (m-k)}\sum _{j=1}^{m-k-1}\Bigg(\!\left(\frac{1}{a+b-j-k+m}-\frac{1}{a+b-2 j-2 k+2 m}\right.\nonumber\\
&-\!\left.\frac{1}{a+b-2 j-2 k+2 m-1}\right)(\psi _0(m-k)-\psi _0(a+b-j-k+m)\nonumber\\
&-\psi _0(a+b-k+m)-\psi _0(m-k-j))+\frac{1}{(a+b-2 j-2 k+2 m-1)^2}\Bigg).\label{eq:s6r}
\end{align}
The result~(\ref{eq:s6r}) is obtained by taking derivative of $c$ of the identity~(\ref{eq:lemma5}) with the specialization
\begin{equation}
 a\to a+b,~~~m\to m-k-1,
\end{equation}
before setting $c=0$. For other combinations of $c$ and $\lambda$, the corresponding inner summations in (\ref{eq:fb2c1}) can be simplified similarly as the two cases above.

Inserting the simplification results of the inner summations into~(\ref{eq:fb2c}), the resulting sums only consist of polygamma and rational functions, which are further computed into unsimplifiable basis of the form~(\ref{eq:s1basis}) by using the identities~(\ref{eq:B1})--(\ref{eq:B12c1}).

For the integrals $\mathrm{I_A}$ and $\mathrm{I_B}$ in case B, the corresponding summations~(\ref{eq:fA1S})--(\ref{eq:fB2S}) are now simplified to the results shown below. For $\mathrm{I_A}$, one has
\begin{align}
\mathrm{I_A}=~\!\!&d_1 \sum _{k=1}^m \frac{\psi _0(a+k)}{k}+d_2 \sum _{k=1}^m \frac{\psi _0(b+k)}{k}-2 m \sum _{k=1}^m \frac{\psi _0(a+b+k+m)}{k}\nonumber\\
&+d_3 \sum _{k=1}^m \frac{\psi _0(a+b+k+m)}{a+k}+\left(d_3+d_4\right) \sum _{k=1}^m \frac{\psi _0(a+b+k+m)}{b+k}+\mathrm{CF},\label{eq:IAf}
\end{align}
where the coefficients $d_i$ are
\begin{align}
d_1=&\frac{2 m (a+m) \left(a^2+a (b+3 m)+2 b m+3 m^2-1\right)}{(a+b+2 m-1)_3}\\
d_2=&\frac{2 m (b+m) \left(a (b+2 m)+b^2+3 b m+3 m^2-1\right)}{(a+b+2 m-1)_3}\\
d_3=&\frac{2 (b+m) (a+b+m) \left(m (3 a+4 b)+(a+b)^2+3 m^2-1\right)}{(a+b+2 m-1)_3}\\
d_4=&\frac{2 (a-b) (a+b+m)}{a+b+2 m}.
\end{align}
For $\mathrm{I_B}$, one has
\begin{align}
\mathrm{I_B}=~\!\!&d_1 \sum _{k=1}^m \frac{\psi _0(a+k)}{k}+d_2 \sum _{k=1}^m \frac{\psi _0(b+k)}{k}+2 m \sum _{k=1}^m \frac{\psi _0(a+b+k)}{k} \nonumber\\
&-2 m\sum _{k=1}^m \frac{\psi _0(a+b+2 k)}{k}+d_3 \sum _{k=1}^m \frac{\psi _0(a+b+k)}{b+k}+d_4\sum _{k=1}^m \frac{\psi _0(b+k)}{a+k}\nonumber\\
&-d_4 \sum _{k=1}^m \frac{\psi _0(a+b+2 k)}{a+k}+d_4 \sum _{k=1}^m \frac{\psi _0(a+b+2 k)}{b+k}\nonumber\\
&+\left(d_3+d_4\right) \sum _{k=1}^m \frac{\psi _0(a+b+k)}{a+k}+\mathrm{CF}.\label{eq:IBf}
\end{align}
Inserting the $\mathrm{I_A}$ expression~(\ref{eq:IAf}) and $\mathrm{I_B}$ expression~(\ref{eq:IBf}) into~(\ref{def:vs}), we arrive at
\begin{align}
\mathbb{V}\!\left[S\right]=&-2 m \Omega _2^{(a+b,0)}-d_4 \Omega _2^{(b,a)}+\left(d_3+d_4\right) \Omega _3^{(a,b)}+\mathrm{CF},\label{eq:vsfsc}
\end{align}
where the summation $\Omega _2^{(a,b)}$ is defined in~(\ref{eq:omega2}), and
\begin{align}
\Omega _3^{(a,b)}=~\!\!&\sum _{k=1}^m \left(\frac{\psi _0(a+b+k+m)}{a+k}+\frac{\psi _0(a+b+k+m)}{b+k}-\frac{\psi _0(a+b+k)}{a+k}\right.\nonumber\\
&-\!\left.\frac{\psi _0(a+b+k)}{b+k}\right).
\end{align}
By using the identities~(\ref{eq:B12c2}) and~(\ref{eq:B12c3}), the result~(\ref{eq:vsfsc}) is simplified to the variance formula~(\ref{f:fpv}), which completes the proof of Proposition~{\ref{propfv}.

\section{Conclusions}\label{sec13}
In this work, we compute the exact yet explicit variance formulas of von Neumann entanglement entropy over fermionic Gaussian states with and without particle number constrains. The obtained formulas provide insights into the fluctuations of von Neumann entropy. An essential ingredient in obtaining the results is a new simplification framework of dummy summation and re-summation techniques. The new framework may also be useful in computing higher order moments of von Neumann entropy as well as other entanglement indicators over the fermionic Gaussian ensemble.

\bmhead{Acknowledgments}
The work of Lu Wei is supported in part by the U.S. National Science Foundation ($\#$2150486).




\begin{appendices}
\section{Summation representations}\label{secA}
In this appendix, we list the summation representations of the integrals in $\mathrm{I_A}$ and $\mathrm{I_B}$ as summarized in Table~\ref{tab_in}. The summation representations for case A are listed in Appendix~(\ref{subsecA1}), the ones for case B are listed in Appendix~(\ref{subsecA2}).

\subsection{Summation representations of case A}\label{subsecA1}
\mathleft
\begin{align}\label{eq:A1S}
\mathrm{A_1}=&\sum _{k=0}^{m-1}\sum _{j=2 k-2}^{2 k}\frac{2(-1)^j(2 a+4 k+1) (j+1)_2 (a+j+1)_2}{\Gamma (2k-j+1) \Gamma (j-2 k+3) (2 a+j+2 k+1)_3}\bigg((\psi_0 (j+3)\nonumber\\
&-\psi_0 (2 a+j+2 k+4)-\psi_0 (j-2 k+3)+\psi_0 (a+j+3))^2\nonumber\\
&-\psi _1(2 a+j+2 k+4)+\psi _1(a+j+3)-\psi _1(j-2 k+3)+\psi _1(j+3)\bigg)\nonumber\\
& +\sum _{k=0}^{m-1}\sum _{j=0}^{2 k-3} \frac{4(2 a+4 k+1) (j+1)_2 (a+j+1)_2}{(2k-j-2)_3 (2 a+j+2 k+1)_3}(\psi_0 (2 a+j+2 k+4)\nonumber\\
&-\psi_0 (a+j+3)+\psi_0 (2k-j-2)-\psi_0 (j+3))
\end{align}

\begin{align}\label{eq:A2S}
\mathrm{A_2}=&\sum _{k=0}^{m-1} \frac{(2 a+4 k+1) \Gamma (2 k+1) \Gamma (2 a+2 k+1)}{\Gamma (2 a+4 k+4)}\left(\rule{0cm}{0.67cm}\sum _{j=0}^{2 k}\frac{2 (j+1) (2k-j+1) }{\Gamma (j+1)\Gamma (a+j+1)}\right.\nonumber\\
 &\times\frac{\Gamma^2 (a+2 k+2)}{  \Gamma (2k-j+1) \Gamma (a+2k-j+1)}((\psi_0 (a+2 k+2)-\psi_0 (2 a+4 k+4)\nonumber\\
 &-\psi_0 (2)+\psi_0 (2k-j+2)) (\psi_0 (a+2 k+2)-\psi_0 (2 a+4 k+4)+\psi_0 (j+2)\nonumber\\
 &-\psi_0 (2))-\psi _1(2 a+4 k+4))-\sum _{j=0}^{2 k} \frac{(j+1) \Gamma (a+2 k+1) \Gamma (a+2 k+3)}{\Gamma (j) \Gamma (a+j+1) \Gamma (a-j+2 k+1)}\nonumber\\
& \times\frac{1}{\Gamma (2k-j+1)}((\psi_0 (a+2 k+1)-\psi_0 (2 a+4 k+4)+\psi_0 (2k-j+2)\nonumber\\
 &-\psi_0 (1)) (\psi_0 (a+2 k+3)-\psi_0 (2 a+4 k+4)+\psi_0 (j+2)-\psi_0 (3))\nonumber\\
 &-\psi _1(2 a+4 k+4))-\sum _{j=0}^{2 k} \frac{ (2k-j+1)\Gamma (a+2 k+1) \Gamma (a+2 k+3)}{ \Gamma (a+j+1) \Gamma (2 k-j) \Gamma (2 k-j+a+1)}\nonumber\\
&\times\frac{1}{\Gamma (j+1)}((\psi_0 (a+2 k+3)-\psi_0 (2 a+4 k+4) +\psi_0 (2k-j+2)\nonumber\\
&-\psi_0 (3)) (\psi_0 (a+2 k+1)-\psi_0 (2 a+4 k+4)+\psi_0 (j+2)-\psi_0 (1))\nonumber\\
& -\psi _1(2 a+4 k+4))+\sum _{j=1}^{2 k-1} \sum _{i=1}^{2 k-j} \frac{4 i (2 k-i+2) \Gamma (a+2 k-j+1)}{\Gamma (a+i) \Gamma (a+2 k-i+2)}\nonumber\\
&\times\frac{\Gamma (a+2 k+j+3)}{(j)_3 \Gamma (j+i+1) \Gamma (2 k-j-i+1)}(\psi _0(a+2 k+j+3)\nonumber\\
&-\psi _0(2 a+4 k+4)+\psi _0(2 k-i+3)-\psi _0(j+3))\left.\!\! \rule{0cm}{0.67cm}\right)
\end{align}

\begin{align}\label{eq:B1S}
 \mathrm{B_1}=&\sum _{k=0}^{m-1} \left(\psi_0 (a+2 k)+\psi_0 (2 a+2 k)-2 \psi_0 (2 a+4 k)-\frac{1}{2} \left(\frac{a}{a+2 k+1}\right.\right.~~~~~~~\nonumber\\
&\!\left.\left.+\frac{a}{a+2 k}+\frac{2}{2 a+4 k+1}\right)+1\right)^2
\end{align}

\begin{align}\label{eq:B2S}
 \mathrm{B_2}=&\sum _{k=1}^{m-1} \sum _{j=1}^{m-k} \frac{\Gamma (2 a+2 k-1) \Gamma (2 j+2 k-1)}{2 (2j-1)^2 j^2 (2 j+1)^2 \Gamma (2 k-1) \Gamma (2 a+2 j+2 k-1)}\nonumber\\
&\times \frac{(2 a+4 k-3) (2 a+4 j+4 k-3)}{(a+j+2 k-2)^2 (a+j+2 k-1)^2 (2 a+2 j+4 k-3)^2} \left(a^2 (2 j+1)\right.\nonumber\\
&+\!\left.a (j+1) (2 j+4 k-3)+2 j^2+j (4 k-3)+4 k^2-6 k+2\right)^2
\end{align}

\subsection{Summation representations of case B}\label{subsecA2}
\mathcenter
\begin{equation}\label{eq:fA1S}
\mathcal{A}_1=\mathcal{A}_1^{(a,b)}+\mathcal{A}_1^{(b,a)},
\end{equation}
\mathleft
\begin{align}\label{eq:fA1_abS}
\mathcal{A}_1^{(a,b)}=&-\frac{2 (m (b+m))}{a+b+2 m}\sum _{i=1}^{m-3} \frac{(b+i+1) (i)_2}{(m-i-2)_3 (a+b+i+m+1)}(\psi_0(b+i+2)\nonumber\\
&-\psi_0(a+b+i+m+2)-\psi_0(m-i-2)+\psi_0(i+2))+\frac{ a+b+m}{a+b+2 m}\nonumber\\
&\times2 (a+m)\sum _{i=1}^{m-2} \frac{(b+i+1) (i)_2}{(m-i-1) (a+b+i+m)_3}(-\psi_0(a+b+i+m+3)\nonumber\\
&+\psi_0(b+i+2)-\psi_0(m-i-1)+\psi_0(i+2))-\frac{m (b+m)}{a+b+2 m}\nonumber\\
&\sum _{i=m-3}^{m-1} \frac{(b+i+2) (-1)^{i+m} (i+1)_2}{\Gamma (m-i) \Gamma (i-m+4) (a+b+i+m+2)}\big(\psi _1(b+i+3)\nonumber\\
&+\psi _1(i+3)-\psi _1(i-m+4)-\psi _1(a+b+i+m+3)+(\psi_0(i+3)\nonumber\\
&-\psi_0(a+b+i+m+3)-\psi_0(i-m+4)+\psi_0(b+i+3))^2\big)\nonumber\\
&-\frac{(a + m)(a+b+m) (b+m) (m-1)_2}{(a+b+2 m) (a+b+2 m-1)_3}\big(\!-\psi _1(a+b+2 m+2)\nonumber\\
&+\psi _1(b+m+1)+\psi _1(m+1)-\psi _1(1)+\psi_0^2(1)+(\psi _0(b+m+1)\nonumber\\
&-\psi _0(a+b+2 m+2)+\psi _0(m+1))(\psi _0(b+m+1)+\psi _0(m+1)\nonumber\\
&-\psi _0(a+b+2 m+2)-2 \psi _0(1))\big)
\end{align}

\begin{align}\label{eq:fA2S}
\mathcal{A}_2=&-\frac{2 (a+m) (b+m) \Gamma (m+1)}{(a+b+2 m) (a+b+m+2)_m}\sum _{i=0}^{m-1}(-1)^i \frac{(i+1) (m-i)}{\Gamma (a+i+2)}\nonumber\\
&\times\sum _{j=i-1}^{i+1} (-1)^j \frac{\Gamma (a+i-j+m+1) (b-i+m+1)_j}{\Gamma (j+1) \Gamma (i-j+2) \Gamma (j-i+2) \Gamma (m-j)}\nonumber\\
&\times (\psi _1(a+b+2 m+2)+(\psi _0(a+i-j+m+1)-\psi _0(a+b+2 m+2)\nonumber\\
&-\psi _0(i-j+2)+\psi _0(i+2)) (\psi _0(a+b+2 m+2)+\psi _0(j-i+2)\nonumber\\
&-\psi _0(b-i+j+m+1)-\psi _0(m-i+1)))-2 \Gamma (m+1) \Gamma (a+m+1)\nonumber\\
&\times\frac{(a+b+m) \Gamma (b+m+1)}{(a+b+2 m) (a+b+m+1)_{m+1}}\sum _{i=0}^{m-2} \frac{1}{\Gamma (i+1) \Gamma (a+i+2)}\nonumber\\
&\times\frac{1}{\Gamma (m-i-1) \Gamma (b-i+m)}(\psi _1(a+b+2 m+2)+(\psi _0(a+m+1)\nonumber\\
&-\psi _0(a+b+2 m+2)+\psi _0(i+2)-\psi _0(1)) (\psi _0(a+b+2 m+2)\nonumber\\
&-\psi _0(b+m+1)-\psi _0(m-i)+\psi _0(1)))+\mathcal{A}_2^{(a,b)}+\mathcal{A}_2^{(b,a)},
\end{align}
\begin{align}\label{eq:fA2_abS}
\mathcal{A}_2^{(a,b)}=&\frac{2 \Gamma (m+1) \Gamma (a+b+m+1)}{\Gamma (a+b+2 m+2)}\Bigg(\frac{(a+m) (b+m) (a+b+m+1) }{a+b+2 m}\nonumber\\
&\times\sum _{i=1}^{m-2} \frac{i (m-i+1)}{\Gamma (b+i+1) \Gamma (a-i+m+2)}\sum _{j=1}^{m-i-1} \frac{\Gamma (a+j+m+2)}{(j)_3 \Gamma (i+j+1)}\nonumber\\
&\times\frac{\Gamma (b-j+m) }{\Gamma (m-i-j)}(\psi _0(m-i+2)-\psi _0(a+b+2 m+2)-\psi _0(j+3)\nonumber\\
&+\psi _0(a+j+m+2))-\frac{a+b+m}{a+b+2 m} \sum _{i=1}^{m-1} \frac{i (m-i)}{\Gamma (b+i+1)}\nonumber\\
&\times\frac{1}{ \Gamma (a-i+m+1)} (\psi _0(a+j+m+1)-\psi _0(a+b+2 m+2)\nonumber\\
&+\psi _0(m-i+1)-\psi _0(j+1))\Bigg)
\end{align}
\begin{align}\label{eq:fB1S}
\mathcal{B}_1=&\sum _{k=0}^{m-1}\left(\left(\frac{a^2-b^2}{4 (a+b+2 k)}+\frac{b^2-a^2}{4 (a+b+2 k+2)}+\frac{1}{2}\right) \psi _0(a+k+1)\right.\nonumber\\
&+\left(\frac{a^2-b^2}{4 (a+b+2 k+2)}+\frac{b^2-a^2}{4 (a+b+2 k)}+\frac{1}{2}\right) \psi _0(b+k+1)\nonumber\\
&+\psi _0(a+b+k+1)-2 \psi _0(a+b+2 k+2)-\frac{a+b}{2 (a+b+2 k)}\nonumber\\
&-\!\left.\frac{a+b}{2 (a+b+2 k+2)}+\frac{1}{a+b+2 k+1}+1\right)^2
\end{align}
\begin{align}\label{eq:fB2S}
\mathcal{B}_2=&\sum _{k=1}^{m-1} \frac{k (a+b+k)}{2 (a+k) (b+k) (a+b+2 k) (a+b+2 k-1)_3}\Bigg(2 (a+k) (b+k)\nonumber\\
&\times (\psi_0 (b+k+1)-\psi_0 (a+k+1))+\frac{(k-1) (a-b) (a+b+2 k+1)}{a+b+k}\!\Bigg)^2\nonumber\\
&+\sum _{k=1}^{m-2} \sum _{j=1}^{m-k-1}\frac{2 (a+b+2 k-1) \Gamma (j+k+1) (a+b+2 j+2 k+1)}{\Gamma (k) \Gamma (a+k) \Gamma (b+k) \Gamma (a+j+k+1) \Gamma (b+j+k+1)}\nonumber\\
&\times\frac{\Gamma (a+b+k)}{\Gamma (a+b+j+k+1)}\Bigg(\frac{\Gamma (a+k) \Gamma (b+j+k+1)}{(j)_3 (a+b+j+2 k-1)_3}\left(a^2 (j+2)+a (j+2)\right.\nonumber\\
&\times\left.\! (b+j+2 k)+j (b+2 k+1)+2 k (b+k)+j^2\right)\nonumber\\
&-\frac{(-1)^j \Gamma (b+k) \Gamma (a+j+k+1)}{(j)_3 (a+b+j+2 k-1)_3}(2 (k-1) (a+(b+1) (j+2)+2k-2)\nonumber\\
&+(b+1) (j+2) (a+b+j+1))\!\Bigg)^2
\end{align}

\section{Summation identities}\label{secB}
In this appendix, we list the finite sum identities useful in simplifying the summations in Appendix \ref{secA} that are not listed in the main part of the paper. The identities of the existing simplification framework are listed in Appendix \ref{secB1} and the identities of the new simplification framework are listed in Appendix \ref{secB2}.

\subsection{Summation identities of the existing simplification framework}\label{secB1}
We list below the identities of the existing simplification framework. Here, it is sufficient to assume that $a,b\ge0, a\neq b$ in identities (\ref{eq:B1})--(\ref{eq:B4}),~(\ref{eq:B7})--(\ref{eq:B72}),~(\ref{eq:B12c1}),~(\ref{eq:B12c2})--(\ref{eq:B12c3}), $a>m$ in~(\ref{eq:B9}), and $a,b\geq 1, n\geq m$ in~(\ref{eq:B11ic}),~(\ref{eq:B20})--(\ref{eq:B202}).
\mathleft
\begin{align}\label{eq:B1}
\sum_{i=1}^{m}\psi_{0}(i+a)=(m+a)\psi_{0}(m+a+1)-a\psi_{0}(a+1)-m
\end{align}

\begin{align}\label{eq:B3}
\sum_{i=1}^m i \psi_0(i+a) =&-\frac{1}{2} (a-m-1) (a+m) \psi _0(a+m+1)\nonumber\\
&+\frac{1}{2} (a-1) a \psi _0(a+1)-\frac{1}{4} m (-2 a+m+3)
\end{align}

\begin{align}\label{eq:B31}
\sum_{i=1}^m i^2 \psi_0(i+a) =&\frac{1}{6}\left(2a^3-3a^2+a+2m^3+3m^2+m\right)\psi_{0}(a+m+1)\nonumber\\
&-\frac{1}{6}a\left(2a^2-3a+1\right)\psi_{0}(a+1)\nonumber\\
&-\frac{1}{36}m\left(12a^2-6am-24a+4m^2+15m+17\right)
\end{align}

\begin{align}\label{eq:B32}
\sum_{i=1}^m i^3 \psi_0(i+a) =&-\frac{1}{4} \left(a^4-2a^3+a^2-m^4-2m^3-m^2\right)\psi_{0}(a+m+1)\nonumber\\
&+\frac{1}{4}(a-1)^{2}a^{2}\psi_{0}(a+1)-\frac{1}{48}m\left(-12a^3+6a^{2}m+30a^2\right.\nonumber\\
&-\!\left.4am^2-18am-26a+3m^3+14m^2+21m+10\right)
\end{align}

\begin{align}\label{eq:B2}
\sum_{i=1}^{m}\psi_{1}(i+a)=&(m+a)\psi_{1}(m+a+1)-a\psi_{1}(a+1)\nonumber\\
&+\psi_{0}(m+a+1)-\psi_{0}(a+1)
\end{align}

\begin{align}\label{eq:B30}
\sum_{i=1}^{m} \psi_{0}^2(i+a)
=&(a+m) \psi _0^2(a+m){}-(2 a+2 m-1) \psi _0(a+m)-a \psi _0^2(a){}\nonumber\\
&+(2 a-1) \psi _0(a)+2 m.
\end{align}

\begin{align}\label{eq:B4}
\sum_{i=1}^{m}\frac{\psi_{0}(i+a)}{i+a}=\frac{1}{2}\left(\psi_{1}(m+a+1)-\psi_{1}(a+1)+\psi_{0}^{2}(m+a+1)-\psi_{0}^{2}(a+1)\right)
\end{align}

\begin{align}\label{eq:B5}
\sum_{i=1}^{m}\frac{\psi_{0}(m+1-i)}{i}=\psi_{0}^{2}(m+1)-\psi_{0}(1)\psi_{0}(m+1)+\psi_{1}(m+1) -\psi_{1}(1)
\end{align}

\begin{align}\label{eq:B6}
\sum _{i=1}^m \frac{\psi_0 (m+1+i)}{i}=\psi _0^2(m+1)-\psi _0(1) \psi _0(m+1)-\frac{1}{2} \psi _1(m+1)+\frac{1}{2}\psi _1(1)
\end{align}

\mathleft
\begin{align}\label{eq:B7}
&\sum_{i=1}^{m}\psi_{0}(i+a)\psi_{0}(i+b)\nonumber\\
&=(b-a) \sum _{i=1}^{m-1}\frac{\psi_0(a+i)}{b+i}-a \psi_0(a+1)\psi_0(b+1)+(m+a) \psi_0(m+a)  \nonumber\\
&~~~\!~\times\psi_0(m+b)+a \psi_0(a+1)-(m+a-1) \psi_0(m+a)-(m+b) \psi_0(m+b)\nonumber\\
&~~~\!~+(b+1) \psi_0(b+1)+2 m-2
\end{align}

\begin{align}\label{eq:B71}
&\sum_{i=1}^{m}i\psi_{0}(i+a)\psi_{0}(i+b)\nonumber\\
&=\frac{1}{2} (b-a+1) (a-b) \sum _{i=1}^{m-1} \frac{\psi _0(a+i)}{b+i}-\frac{1}{4} a (a+2 b-3) \psi _0(a+1)-\frac{1}{4} (b+1)\nonumber\\
&~~~\!~\times (2 a+b-2) \psi _0(b+1)+\frac{1}{2} (a-1) a \psi _0(a+1) \psi _0(b+1)+\frac{1}{4} (a+m-1)\nonumber\\
&~~~\!~ \times(a+2 b-m-2) \psi _0(a+m)+\frac{1}{4} (b+m) (2 a+b-m-1) \psi _0(b+m)-\frac{1}{2}\nonumber\\
&~~~\!~ \times\left(a^2-a-m (m+1)\right) \psi _0(a+m) \psi (b+m)-\frac{1}{4} (m-1) (3 a+3 b-m-4)
\end{align}

\begin{align}\label{eq:B72}
&\sum_{i=1}^{m}i^2\psi_{0}(i+a)\psi_{0}(i+b)\nonumber\\
&=\frac{1}{6} (a-b) \left(3 a^2+2 a b-4 a-2 b^2-b+1\right) \sum _{i=1}^{m-1} \frac{\psi _0(a+i)}{b+i}-\frac{1}{6} \left(-a^2 (5 b+2)\right.\nonumber\\
&~~~\!~+\!\left.3 a^3+a (5 b-1)-m \left(2 m^2+3 m+1\right)\right) \psi _0(a+m) \psi _0(b+m)+\frac{1}{6} (a-1) a\nonumber\\
&~~~\!~\times (3 a-5 b+1) \psi _0(a+1) \psi_0 (b+1)-\Bigg(\!-\frac{1}{12} (2 b-1) m^2+\frac{1}{36} m \left(24 a^2\right.\nonumber\\
&~~~\!~-\left.\!24 a b-24 a+12 b^2+12 b-1\right)+\frac{1}{36} (a-1) \left(28 a^2-18 a b-5 a+6 b\right.\nonumber\\
&~~~\!~+\!\left.12 b^2+6\right)+\frac{(a-1) (a-b)}{3 (b+m-1)}+\frac{m^3}{9}\Bigg)\psi _0(a+m)-\frac{1}{36} \left(4 m^312 a^2 b\right.\nonumber\\
&~~~\!~-3 (2 a-1) m^2+\left(12 a^2-12 a-1\right) m-30 a^2+6 a b^2-12 a b+30 a+4 b^3\nonumber\\
&~~~\!~-\!\left.3 b^2-b\right) \psi _0(b+m)+\frac{1}{36} a \left(28 a^2-9 a (2 b+1)+12 b^2-13\right) \psi _0(a+1)\nonumber\\
&~~~\!~+\frac{1}{36} \left(6 a^2 (2 b-3)+6 a \left(b^2-2 b+2\right)+4 b^3-3 b^2-b+6\right) \psi _0(b+1)\nonumber\\
&~~~\!~+\frac{2 m^3}{27}-\frac{5}{36} m^2 (a+b-1)+\frac{1}{36} \left(-40 a^2+12 a b+51 a-16 b^2+3 b-16\right)\nonumber\\
&~~~\!~+\frac{1}{108} m \left(120 a^2-36 a b-138 a+48 b^2+6 b+25\right)+\frac{a-1}{3 (b+m-1)}\nonumber\\
&~~~\!~-\frac{a-1}{3 (a+m-1)}
\end{align}

\begin{align}\label{eq:B9}
\sum_{i=1}^{m}\frac{\psi_{0}(a+1-i)}{i}=&-\sum_{i=1}^{m}\frac{\psi_{0}(i+a-m)}{i}+\frac{1}{2}\left(\psi_{1}(a+1)-\psi_{1}(a-m)\right)\nonumber\\
&+(\psi_{0}(a-m)+\psi_0(a+1))(\psi_{0}(m+1)-\psi_{0}(1))\nonumber\\
&+\frac{1}{2}(\psi_{0}(a-m)-\psi_{0}(a+1))^2
\end{align}

\begin{align}\label{eq:B12c1}
&\sum_{i=1}^{m}\left(\frac{\psi_{0}(i+b)}{i+a}+\frac{\psi_{0}(i+a)}{i+b}\right)\nonumber\\
&=\psi_{0}(m+a+1)\psi_{0}(m+b+1)-\psi_{0}(a+1)\psi_{0}(b+1)+\frac{1}{a-b}(\psi_{0}(m+a+1)\nonumber\\
&~~~-\psi_{0}(m+b+1)-\psi_{0}(a+1)+\psi_{0}(b+1))
\end{align}

\begin{align}\label{eq:B11ic}
\sum _{i=1}^m \frac{\psi_0(a+b+i)}{i}=&\sum _{i=1}^m \frac{\psi_0(b+i)}{i}-\sum _{i=1}^a \frac{\psi_0(b+i+m)}{b+i-1}+\frac{1}{2} \big((\psi_0(a+b)\nonumber\\
&-\psi_0(b))\times (\psi_0(a+b)+\psi_0(b)+2 (\psi_0(m+1)-\psi_0(1)))\nonumber\\
&-\psi _1(a+b)+\psi _1(b)\big)
\end{align}

\begin{align}\label{eq:B12c2}
&\sum _{i=1}^m \left(\frac{\psi _0(a+b+i+m)}{b+i}+\frac{\psi _0(a+i)}{b+i}+\frac{\psi _0(a+b+2i)}{a+i}-\frac{\psi _0(a+b+i)}{a+i}\right.\nonumber\\
&-\!\left.\frac{\psi _0(a+b+2 i)}{b+i}\right)\nonumber\\
&=\frac{\psi_0\left(\frac{a}{2}+\frac{b}{2}+m\right)}{b-a}-\frac{(a+b+m) \psi_0(a+b+m)}{b (a+m)}+\frac{\psi_0(a+b+2 m)}{a+m}\nonumber\\
&~~~\!~+\psi _0(a+m) \left(\psi _0(b+m+1)-\psi _0(b)-\frac{1}{b-a}\right)+\frac{a \psi _0(a)}{b (b-a)}+\frac{\psi _0(a+b)}{b}\nonumber\\
&~~~~\!-\frac{\psi _0\left(\frac{a}{2}+\frac{b}{2}\right)}{b-a}
\end{align}

\begin{align}\label{eq:B12c3}
&\sum _{i=1}^m \left(\frac{\psi _0(a+b+i+m)}{a+i}+\frac{\psi _0(a+b+i+m)}{b+i}-\frac{\psi _0(a+b+i)}{a+i}\right.\nonumber\\
&-\!\left.\frac{\psi _0(a+b+i)}{b+i}\right)\nonumber\\
&=\frac{(a+b+2 m) \psi _0(a+b+2 m)}{(a+m) (b+m)}-\left(\frac{1}{a+m}+\frac{1}{a}+\frac{1}{b+m}+\frac{1}{b}\right) \psi _0(a+b+m)\nonumber\\
&~~~~+\left(\frac{1}{a}+\frac{1}{b}\right) \psi _0(a+b)+\left(\psi _0(a+m)-\psi _0(a)\right) \left(\psi _0(b+m)-\psi _0(b)\right)
\end{align}

\begin{align}\label{eq:B20}
\sum _{i=1}^{m} \frac{(n-i)!}{ (m-i)!}=\frac{n!}{(m-1)! (n-m+1)}
\end{align}

\begin{align}\label{eq:B21}
\sum _{i=1}^{m} \frac{(n-i)!}{ (m-i)!i}=\frac{n! }{m!}(\psi _{0}(n+1)-\psi _{0}(n-m+1))
\end{align}

\begin{align}\label{eq:B22}
\sum _{i=1}^{m} \frac{(n-i)!}{ (m-i)!i^2}=&\frac{n! }{m!} \Bigg(\sum _{i=1}^m\frac{\psi_0(i+n-m)}{i}+\frac{1}{2}\big (\psi _1(n-m+1)-\psi _1(n+1)\nonumber\\
&+\psi _0(n-m+1){}^2-\psi _0(n+1){}^2\big)+\psi _0(n-m)\nonumber\\
&\times\left(-\psi _0(n-m+1)+\psi _0(n+1)-\psi _0(m+1)+\psi _0(1)\right)\!\Bigg)
\end{align}

\begin{align}\label{eq:B201}
\sum _{i=1}^m \frac{(n-i)!}{(m+a-i)!}=\frac{1}{n-m-a+1}\left(\frac{n!}{(a+m-1)!}-\frac{(n-m)!}{(a-1)!}\right)
\end{align}

\begin{align}\label{eq:B202}
&\sum _{i=1}^m \frac{(n-i)! }{(m+a-i)!}\psi _0(m+a-i+1)\nonumber\\
&=\frac{1}{1-a-m+n}\left(\frac{n! }{(a+m-1)!}\left(\psi _0(a+m)-\frac{1}{1-a-m+n}\right)\right.\nonumber\\
&~~~\!~-\!\left.\frac{(n-m)! }{(a-1)!}\left(\psi _0(a)-\frac{1}{1-a-m+n}\right)\right)
\end{align}


The derivation of the identities~(\ref{eq:B1})--(\ref{eq:B12c1}) and~(\ref{eq:B20})--(\ref{eq:B22}) can be found in~\cite{Wei17,Wei20,Wei20BH,HWC21,HW22, Milgram}. In the same manner as obtaining~(\ref{eq:B30}) in~Section~\ref{subsubsec 3.2.1}, the three identities~(\ref{eq:B11ic})--(\ref{eq:B12c3}) are derived by first rewriting respectively the summations
\begin{align}
&\sum _{i=1}^m \frac{\psi_0(a+b+i)}{i}=\sum _{i=1}^m \frac{\psi_0(b+i)}{i}+\sum _{j=1}^a\sum _{i=1}^m \frac{1}{i} \frac{1}{b+i+j-1}\\
&\sum _{i=1}^m \frac{\psi _0(a+b+2 i)}{a+i}=\sum _{i=1}^m \frac{\psi _0(a+b+i)}{a+i}+\sum _{j=1}^m \sum _{i=j}^m \frac{1}{(a+i) (a+b+i+j-1)}\\
&\sum _{i=1}^m \frac{\psi _0(a+b+i+m)}{b+i}=\sum _{i=1}^m \frac{\psi _0(a+b+i)}{b+i}+\sum _{j=1}^m \sum _{i=1}^m \frac{1}{(b+i) (a+b+i+j-1)}.
\end{align}
For the identity~(\ref{eq:B201}), it is obtained by first considering
\begin{align}
\sum _{i=1}^m \frac{(n-i)!}{(m+a-i)!}=\sum _{i=1}^{a+m} \frac{(n-i)!}{(m+a-i)!}-\sum _{i=1}^a \frac{(n-m-i)!}{(a-i)!}
\end{align}
before applying~(\ref{eq:B20}). Note that the identity~(\ref{eq:B201}) is analytically continued to any complex number $a$, where, by taking an derivative of $a$, the identity~(\ref{eq:B202}) is established.

\subsection{Additional summation identities of the new simplification framework}\label{secB2}
Here, we list the additional summation identities of the new simplification framework useful in the simplification process in Section~\ref{subsec3.3}. These identities are obtained by taking appropriate derivatives of the formulas in lemmas~\ref{lemma1}--\ref{lemma4}. It is sufficient to assume that $a,b,c\geq 0$ in~(\ref{eq:Bn0})--(\ref{eq:Bn8}).
\begin{align} \label{eq:Bn0}
&\sum_{i=1}^m\frac{\psi_0^2(a+i)-\psi_1(a+i)}{\Gamma(i)\Gamma (a+i) \Gamma (m-i+1) \Gamma (b-i+m+1)}\nonumber \\
&=\frac{\Gamma (a+b+2 m-1)}{\Gamma (m) \Gamma (a+m) \Gamma (b+m) \Gamma (a+b+m)}\big(\psi _1(a+b+2 m-1)-\psi _1(a+b+m)\nonumber\\
&~~\!~~-\psi _1(a+m)+\left(\psi _0(a+b+m)-\psi _0(a+b+2 m-1)+\psi _0(a+m)\right){}^2\big)
\end{align}

\begin{align} \label{eq:Bn1}
&\sum _{i=1}^m \frac{\psi _0(i)}{\Gamma (i) \Gamma (a+i) \Gamma (m-i+1) \Gamma (b-i+m+1)}\nonumber \\
&=-\frac{1}{\Gamma (m) \Gamma (b+m) \Gamma (a+b+m)}\sum _{i=1}^{m-1} \frac{\Gamma (a+b-i+2 m-1)}{\Gamma (a-i+m) i }\nonumber\\
&~~~\!~+\frac{\psi _0(m) \Gamma (a+b+2 m-1)}{\Gamma (m) \Gamma (a+m) \Gamma (b+m) \Gamma (a+b+m)}
\end{align}

\begin{align} \label{eq:Bn2}
&\sum _{i=1}^m \frac{\psi _0(i) \psi _0(b-i+m+1)}{\Gamma (i) \Gamma (a+i) \Gamma (m-i+1) \Gamma (b-i+m+1)}\nonumber\\
&=\frac{1}{\Gamma (m) \Gamma (b+m) \Gamma (a+b+m)}\sum _{i=1}^{m-1} \frac{\Gamma (a+b-i+2 m-1)}{\Gamma (a-i+m)i }(-\psi _0(a+b+m)\nonumber\\
&~~~\!~+\psi _0(a+b-i+2 m-1)-\psi _0(b+m))+\frac{\Gamma (a+b+2 m-1)\psi _0(m)}{\Gamma (m) \Gamma (a+m) \Gamma (b+m) } \nonumber\\
&~~\!~~\times \frac{1}{\Gamma (a+b+m)}(\psi _0(a+b+m)-\psi _0(a+b+2 m-1)+\psi _0(b+m))
\end{align}

\begin{align}\label{eq:Bn3}
&\sum _{i=1}^m \frac{\psi _0(i) \psi _0(a+i)}{\Gamma (i) \Gamma (a+i) \Gamma (m-i+1) \Gamma (b-i+m+1)}\nonumber\\
&=\frac{1}{\Gamma (m) \Gamma (b+m)\Gamma (a+b+m) }\sum _{i=1}^{m-1} \frac{\Gamma (a+b-i+2 m-1)}{\Gamma (a-i+m)i }(-\psi _0(a+b+m)\nonumber\\
&~~\!~~+\psi _0(a+b-i+2 m-1)-\psi _0(a-i+m))+\frac{\psi _0(m) \Gamma (a+b+2 m-1)}{\Gamma (m)  \Gamma (b+m) \Gamma (a+b+m)}\nonumber\\
&~~\!~~\times\frac{\psi _0(a+b+m)-\psi _0(a+b+2 m-1)+\psi _0(a+m)}{\Gamma (a+m)}
\end{align}

\begin{align}\label{eq:Bn4}
&\sum _{i=1}^m \frac{\psi _0^2(i)-\psi _1(i)}{\Gamma (i) \Gamma (a+i) \Gamma (m-i+1) \Gamma (b-i+m+1)}\nonumber\\
&=\frac{2 }{\Gamma (m) \Gamma (b+m) \Gamma (a+b+m)}\sum _{i=1}^{m-1} \frac{\Gamma (a+b-i+2 m-1)}{\Gamma (a-i+m)i }(\psi _0(i)-\psi _0(m)\nonumber\\
&~~~\!~-\psi _0(1))+\frac{ \Gamma (a+b+2 m-1)}{\Gamma (m) \Gamma (a+m) \Gamma (b+m) \Gamma (a+b+m)}\left(\psi _0^2(m)-\psi _1(m)\right)
\end{align}

\begin{align}\label{eq:Bn5}
&\sum _{i=1}^m \frac{\psi _0(i) \psi _0(m-i+1)}{\Gamma (i) \Gamma (a+i) \Gamma (m-i+1) \Gamma (b-i+m+1)}\nonumber\\
&=\frac{1}{\Gamma (m) \Gamma (a+b+m)}\left(-\frac{1}{\Gamma (a+m)}\sum _{i=1}^{m-1} \frac{\Gamma (a+b-i+2 m-1)}{ \Gamma (b-i+m)i^2}-\frac{1}{\Gamma (b+m)}\right.\nonumber\\
&~~\!~~\times\sum _{i=1}^{m-1} \frac{\Gamma (a+b-i+2 m-1)}{ \Gamma (a-i+m)i^2}-\frac{\psi _0(m)}{\Gamma (a+m)}\sum _{i=1}^{m-1} \frac{\Gamma (a+b-i+2 m-1)}{\Gamma (b-i+m)i }\nonumber\\
&~~\!~~-\left.\!\frac{\psi _0(m) }{\Gamma (b+m)}\sum _{i=1}^{m-1} \frac{\Gamma (a+b-i+2 m-1)}{\Gamma (a-i+m)i }\right)+\frac{\Gamma (a+b+2 m-1)}{\Gamma (m) \Gamma (a+m) \Gamma (b+m) }\nonumber\\
&~~\!~~\times\frac{1}{\Gamma (a+b+m)}\left(\psi _0^2(m)-\psi _1(m)+\psi _1(1)\right)
\end{align}

\begin{align}\label{eq:Bn6}
&\sum _{i=1}^m \frac{\Gamma (c-i+m) \Gamma (a+b+i+m)}{\Gamma (i) \Gamma (m-i+1) }\psi_0(i)\nonumber\\
&=-\frac{\Gamma (b+1) \Gamma (c) \Gamma (c+m)}{\Gamma (m) \Gamma (b+c+1)}\sum _{i=1}^{m-1} \frac{\Gamma (b+c-i+m)}{\Gamma (c-i+m)i }+\frac{\Gamma (b+1) \Gamma (c) \Gamma (b+c+m)}{\Gamma (m) \Gamma (b+c+1)}\nonumber\\
&~~~\!~\times\left(\psi _0(b+c+m)-\psi _0(b+c+1)+\psi _0(m)\right)
\end{align}

\begin{align}\label{eq:Bn7}
&\sum _{i=1}^m \frac{\Gamma (c-i+m) \Gamma (a+b+i+m)}{\Gamma (a+i) \Gamma (m-i+1)}\left(\psi_0^2 (i)-\psi _1(i)\right)\nonumber\\
&= \frac{2\Gamma (b+1) \Gamma (c) \Gamma (c+m)}{\Gamma (m) \Gamma (b+c+1)}\left(-\sum _{i=1}^{m-1} \frac{\Gamma (b+c-i+m) \psi _0(b+c-i+m)}{\Gamma (c-i+m)i }\right.\nonumber\\
&~~~\!~+\sum _{i=1}^{m-1} \frac{ \Gamma (b+c-i+m)\psi _0(i)}{\Gamma (c-i+m)i }-\left(-\psi _0(b+c+1)+\psi _0(m)+\psi _0(1)\right) \nonumber\\
&~~~\!~\times\!\left.\sum _{i=1}^{m-1} \frac{\Gamma (b+c-i+m)}{\Gamma (c-i+m)i }\right)+\frac{\Gamma (b+1) \Gamma (c) \Gamma (b+c+m)}{\Gamma (m) \Gamma (b+c+1)}\big(\psi _1(b+c+m)\nonumber\\
&~~~\!~-\psi _1(b+c+1)-\psi _1(m)+\left(\psi _0(b+c+m)-\psi _0(b+c+1)+\psi _0(m)\right){}^2\big)
\end{align}

\begin{align}\label{eq:Bn8}
&\sum _{i=1}^m \frac{\Gamma (c-i+m) \Gamma (a+b+i+m)}{\Gamma (a+i) \Gamma (m-i+1)i }\psi _0(a+b+i+m)\nonumber\\
&=\Gamma (c+m) \Gamma (a+b+m)\Bigg(\!\frac{\psi _0(a+b+m)}{\Gamma (a) \Gamma (m+1)}\left(\psi _0(c+m)-\psi _0(c)\right) +\Gamma (b+m+1) \nonumber\\
&~~~\!~\times \Gamma (c) \sum _{i=1}^m \frac{\psi _0(a+b+m)-\psi _0(b-i+m+1)+\psi _0(b+m+1)}{\Gamma (a+i) \Gamma (c+i) \Gamma (m-i+1) \Gamma (b-i+m+1)i }\Bigg)
\end{align}




\end{appendices}


\bibliography{sn-bibliography}



\end{document}